%% file: main.tex
\def\comments{0}
\def\supp{0}

\documentclass[letterpaper,11pt]{article}

\usepackage{preamble}

\newcommand{\llnorm}[1]{\left\lVert#1\right\rVert_2}
\newcommand{\fnorm}[1]{\left\lVert#1\right\rVert_F}

\newcommand{\abs}[1]{\left| #1 \right|}

\renewcommand{\epsilon}{\varepsilon}
\newcommand{\ball}[2]{\mathit{B}_{#2}\left( #1 \right)}

\newcommand{\maxsigma}{\sigma_{\max}}
\newcommand{\minsigma}{\sigma_{\min}}
\newcommand{\mw}{w_{\min}}

\newcommand{\Sym}{\mathsf{Sym}}

\newcommand{\id}{\mathbb{I}}
\newcommand{\Lap}{\mathrm{Lap}}
\newcommand{\opt}{\mathrm{opt}}

\newcommand{\RPGMP}{\ensuremath{\textsc{RPGMP}}}
\newcommand{\PGME}{\ensuremath{\textsc{PGME}}}

\newcommand{\PLoc}{\ensuremath{\textsc{PLoc}}}
\newcommand{\PGLoc}{\ensuremath{\textsc{PGLoc}}}
\newcommand{\PCount}{\ensuremath{\textsc{PCount}}}

\newcommand{\PGE}{\ensuremath{\textsc{PGE}}}
\newcommand{\PEGME}{\ensuremath{\textsc{PEGME}}}

\newcommand{\PSGE}{\ensuremath{\textsc{PSGE}}}
\newcommand{\PTB}{\ensuremath{\textsc{PTerrificBall}}}
\newcommand{\AboveThreshold}{\ensuremath{\textsc{AboveThreshold}}}
\newcommand{\KLSU}{\ensuremath{\textsc{M}_{\textsc{KLSU}}}}

\title{Differentially Private Algorithms for \\ Learning Mixtures of Separated Gaussians}
\author{
Gautam Kamath\thanks{Cheriton School of Computer Science, University of Waterloo. {\tt g@csail.mit.edu}. }
\and 
Or Sheffet\thanks{Department of Computer Science, Bar Ilan University.  {\tt or.sheffet@biu.ac.il}.  }
\and
Vikrant Singhal\thanks{Khoury College of Computer Sciences, Northeastern University. {\tt singhal.vi@husky.neu.edu}. }
\and
Jonathan Ullman\thanks{Khoury College of Computer Sciences, Northeastern University. {\tt jullman@ccs.neu.edu}. }
}

\begin{document}

\maketitle

\begin{abstract}
Learning the parameters of Gaussian mixture models is a fundamental and widely studied problem with numerous applications.  In this work, we give new algorithms for learning the parameters of a high-dimensional, well separated, Gaussian mixture model subject to the strong constraint of differential privacy.  In particular, we give a differentially private analogue of the algorithm of Achlioptas and McSherry.  Our algorithm has two key properties not achieved by prior work: (1) The algorithm's sample complexity matches that of the corresponding non-private algorithm up to lower order terms in a wide range of parameters.  (2) The algorithm does not require strong \emph{a priori} bounds on the parameters of the mixture components.
\end{abstract}
\newpage
\tableofcontents
\newpage

\input{intro}
\input{preliminaries}
\input{pca}
\input{easycase}

\input{hardcase}

\input{sa}

\ifnum\supp=0
\section*{Acknowledgments}
Part of this work was done while the authors were visiting the Simons Institute for Theoretical Computer Science.  
Parts of this work were done while GK was supported as a Microsoft Research Fellow, as part of the Simons-Berkeley Research Fellowship program, while visiting Microsoft Research, Redmond, and while supported by a University of Waterloo startup grant.
This work was done while OS was affiliated with the University of Alberta. OS gratefully acknowledges the Natural Sciences and Engineering Research Council of Canada
(NSERC) for its support through grant \#2017-06701.
JU and VS were supported by NSF grants CCF-1718088, CCF-1750640, and CNS-1816028.
\fi

\addcontentsline{toc}{section}{References}
\bibliographystyle{alpha}
\bibliography{biblio}

\appendix

\input{appendix}

\end{document}

%% file: intro.tex
\newcommand{\mypar}[1]{\medskip\textbf{#1}}

\section{Introduction}
The \emph{Gaussian mixture model} is one of the most important and widely studied models in Statistics---with roots going back over a century~\cite{Pearson94}---and has numerous applications in the physical, life, and social sciences.  In a Gaussian mixture model, we suppose that each sample is drawn by randomly selecting from one of $k$ distinct Gaussian distributions $G_1,\dots,G_k$ in $\R^d$ and then drawing a sample from that distribution.  The problem of \emph{learning a Gaussian mixture model} asks us to take samples from this distribution and approximately recover the parameters (mean and covariance) of each of the underlying Gaussians.  The past decades have seen tremendous progress towards understanding both the sample complexity and computational complexity of learning Gaussian mixtures~\cite{Dasgupta99, DasguptaS00, AroraK01, VempalaW02, AchlioptasM05, ChaudhuriR08a, ChaudhuriR08b, KumarK10, AwasthiS12, RegevV17, HopkinsL18, DiakonikolasKS18b, KothariSS18, KalaiMV10, MoitraV10, BelkinS10,HsuK13,AndersonBGRV14,BhaskaraCMV14, HardtP15,GeHK15, XuHM16, DaskalakisTZ17,AshtianiBHLMP18, FeldmanOS06,FeldmanOS08, DaskalakisK14, SureshOAJ14, DiakonikolasKKLMS16, LiS17, ChanDSS14a}.

In many of the applications of Gaussian mixtures models, especially those in the social sciences, the sample consists of sensitive data belonging to individuals.  In these cases, it is crucial that we not only learn the parameters of the mixture model, but do so while preserving these individuals' \emph{privacy}.  In this work, we study algorithms for learning Gaussian mixtures subject to \emph{differential privacy}~\cite{DworkMNS06}, which has become the \emph{de facto} standard for individual privacy in statistical analysis of sensitive data.  Intuitively, differential privacy guarantees that the output of the algorithm does not depend significantly on any one individual's data, which in this case means any one sample. Differential privacy is used as a measure of privacy for data analysis systems at Google~\cite{ErlingssonPK14},  Apple~\cite{AppleDP17}, and the U.S.~Census Bureau~\cite{DajaniLSKRMGDGKKLSSVA17}. Differential privacy and related notions of \emph{algorithmic stability} are also crucial for statistical validity even when individual privacy is not a primary concern, as they provide generalization guarantees in an adaptive setting~\cite{DworkFHPRR15, BassilyNSSSU16}.

The first differentially private algorithm for learning Gaussian mixtures comes from the work of Nissim, Raskhodnikova, and Smith~\cite{NissimRS07} as an application of their influential \emph{subsample-and-aggregate} framework.  However, their algorithm is a reduction from private estimation to non-private estimation that blows up the sample complexity by at least a $\poly(d)$ factor.

\medskip
The contribution of this work is new differentially private algorithms for recovering the parameters of an unknown Gaussian mixture provided that the components are sufficiently well-separated.  In particular we give differentially private analogues of the algorithm of Achlioptas and McSherry~\cite{AchlioptasM05}, which requires that the means are separated by a factor proportional to $\sqrt{k}$, but independent of the dimension $d$.  Our algorithms have two main features not shared by previous methods:
\begin{itemize}
\item The sample complexity of the algorithm matches that of the corresponding non-private algorithm up to lower order additive terms for a wide range of parameters.
\item The algorithms do not require strong \emph{a priori} bounds on the parameters of the mixture components.  That is, like many algorithms, we require that the algorithm is seeded with some information about the range of the parameters, but the algorithm's sample complexity depends only mildly (polylogarithmically) on the size of this range.
\end{itemize}

\subsection{Problem Formulation}

There are a plethora of algorithms for (non-privately) learning Gaussian mixtures, each with different learning guarantees under different assumptions on the parameters of the underlying distribution.\footnote{We remark that there are also many popular \emph{heurstics} for learning Gaussian mixtures, notably the EM algorithm~\cite{DempsterLR77}, but in this work we focus on algorithms with provable guarantees.}  In this section we describe the version of the problem that our work studies and give some justification for these choices.

We assume that the underlying distribution $\cD$ is a mixture of $k$ Gaussians in high dimension $d$.  The mixture is specified by $k$ \emph{components}, where each component $G_i$ is selected with probability $w_i \in [0,1]$ and is distributed as a Gaussian with mean $\mu_i \in \R^d$ and covariance $\Sigma_i \in \R^{d \times d}$.  Thus the mixture is specified by the tuple $\{ (w_i, \mu_i, \Sigma_i) \}_{i\in [k]}$.  

Our goal is to accurately recover this tuple of parameters. Intuitively, we would like to recover a tuple $\{ (\wh{w}_i, \wh\mu_i, \wh\Sigma_i) \}_{i \in [k]}$, specifying a mixture $\tilde\cG$ such that $\|\wh{w} - w\|_1$ is small and $\| \wh\mu_i - \mu_i \|_{\Sigma_i}$ and $\| \wh\Sigma_i - \Sigma_i \|_{\Sigma_i}$ are small for every $i \in [k]$.  Here, $\| \cdot \|_{\Sigma}$ is the appropriate vector/matrix norm that ensures $\cN(\mu_i, \Sigma_i)$ and $\cN(\wh\mu_i, \wh\Sigma_i)$ are close in total variation distance and we also compare $\wh{w}$ and $w$ in $\|\cdot \|_1$ to ensure that the mixtures are close in total variation distance.  Of course, since the labeling of the components is arbitrary, we can actually only hope to recover the parameters up to some unknown permutation $\pi : [k] \to [k]$ on the components.  We say that an algorithm \emph{learns a mixture of Gaussian using $n$ samples} if it takes $n$ i.i.d.\ samples from an unknown mixture $\cD$ and outputs the parameters of a mixture $\wh\cD$ satisfies these conditions.\footnote{To provide context, one might settle for a weaker goal of \emph{proper learning} where the goal is merely to learn \emph{some} Gaussian mixture, possibly with a different number of components, that is close to the true mixture, or \emph{improper learning} where it suffices to learn \emph{any} such distribution.} 

In this work, we consider mixtures that satisfy the \emph{separation} condition
\begin{equation} \label{eq:separated_intro}
    \forall i \neq j~~\|\mu_i - \mu_j\|_2 \geq
        \wt{\Omega}\left(\sqrt{k} + \sqrt{\frac{1}{w_i} + \frac{1}{w_j}}\right)\cdot
        \max\left\{\|\Sigma_i^{1/2}\|,
        \|\Sigma_j^{1/2}\|\right\}
\end{equation}
%\bigjnote{I'm not sure it's clear to say $\geq \Omega(\cdot)$.  This seems consistent with the condition $\geq k^8$, since $k^8 = \Omega(k^{1/2})$.  Also $\tilde{\Omega}$ usually is a \emph{smaller} quantity than $\Omega$?}\gnote{I'm personally fine it with -- we can handle any separation condition which is $\Omega(k^{1/2})$, including say $k^8$. If you have a better way to write it though, feel free.} \bigjnote{But it looks like we can also handle separation $\sqrt{k}/1000000\log^{100} k$ which we can't.  Personally I guess I always thought it looks better to write $n \geq \tilde{O}(\cdot)$ and read it as ``there isa quantity we could plug in that is $\tilde{O}(\cdot)$ such that $n$ is bigger than that.  But not everyone seems to like this.} \bigonote{I am with Gautam on this. It is fairly standard to use this separation, and in later chapters we give the exact formulation,~\eqref{eq:requiredSeparation}}
Note that the separation condition does not depend on the dimension $d$, only on the number of mixture components.  However,~\eqref{eq:separated_intro} is \emph{not} the weakest possible condition under which one can learn a mixture of Gaussians.  We focus on~\eqref{eq:separated_intro} because this is the regime where it is possible to learn the mixture components using \emph{statistical properties of the data}, such as the large principal components of the data and the centers of a good clustering, which makes this regime amenable to privacy.  In contrast, algorithms that learn with separation proportional to $k^{1/4}$~\cite{VempalaW02}, $k^{\eps}$~\cite{HopkinsL18,KothariSS18,DiakonikolasKS18b}, or even $\sqrt{\log k}$~\cite{RegevV17} use algorithmic machinery such as the sum-of-squares algorithm that has not been studied from the perspective of privacy, or are not computationally efficient.  In particular, a barrier to learning with separation $k^{1/4}$ is that the non-private algorithms are based on single-linkage clustering, which is not amenable to privacy due to its crucial dependence on the distance between individual points.  We remark that one can also learn without any separation conditions, but only with $\exp(k)$ many samples from the distributions~\cite{MoitraV10}.

\medskip
In this work, our goal is to design learning algorithms for mixtures of Gaussians that are also differentially private.  An \emph{$(\eps,\delta)$-differentially private}~\cite{DworkMNS06} randomized algorithm $A$ for learning mixtures of Gaussians is an algorithm that takes a dataset $X$ of samples and:
\begin{itemize}
\item For every pair of datasets $X,X'$ differing on a single sample, the distributions $A(X)$ and $A(X')$ are \emph{$(\eps,\delta)$-close} in a precise sense (Definition~\ref{def:dp}).
\item If $n$ is sufficiently large and $X_1,\dots,X_n \sim \cD$ for a mixture $\cD$ satisfying our assumptions, then $A(X)$ outputs an approximation to the parameters of $\cG$.
\end{itemize}
Note that, while the learning guarantees necessarily rely on the assumption that the data is drawn i.i.d.\ from some mixture of Gaussians, the privacy guarantee is worst-case.  It is important for privacy not to rely on distributional assumptions because we have no way of verifying that the data was truly drawn from a mixture of Gaussians, and if our assumption turns out to be wrong we cannot recover privacy once it is lost.

Furthermore, our algorithms assume certain
\emph{boundedness} of the mixture components.
Specifically, we assume that there are known
quantities $R, \maxsigma, \minsigma$ such that
\begin{equation} \label{eq:boundedness_intro}
    \forall i\in [k]~~\| \mu_i \|_2 \leq R~~\textrm{and}~~
    \minsigma^2 \leq \llnorm{\Sigma_i} \leq \maxsigma^2.
\end{equation}
%\begin{equation} \label{eq:boundedness_intro}
%\forall i\in [k]~~\| \mu_i \|_2 \leq R~~\textrm{and}~~\sigma_{\mathrm{min}}^2 \cdot \id \preceq \Sigma_i \preceq \sigma_{\mathrm{max}}^2 \cdot \id
%\end{equation}
These assumptions are to some extent necessary, as even the state-of-the-art algorithms for learning a single multivariate normal~\cite{KamathLSU19} require boundedness.\footnote{These boundedness conditions are also provably necessary to learn even a single univariate Gaussian for pure differential privacy, concentrated differential privacy, or R\'{e}nyi differential privacy, by the argument of~\cite{KarwaV18}.  One could only hope to avoid boundedness using the most general formulation of $(\eps,\delta)$-differential privacy.}  However, since $R$ and $\sigma_{\mathrm{max}}/\sigma_{\mathrm{min}}$ can be quite large---and even if they are not we cannot expect the user of the algorithm to know these parameter \emph{a priori}---the algorithm's sample complexity should depend only mildly on these parameters so that they can be taken to be quite large.

\subsection{Our Contributions}

The main contribution of our work is an algorithm
with improved sample complexity for learning mixtures
of Gaussians that are separated and bounded.
\begin{thm}[Main, Informal] \label{thm:main_intro}
    There is an $(\eps,\delta)$-differentially
    private algorithm that takes
    $$n = \left(
        \frac{d^2}{ \alpha^2 \mw} + \frac{d^2}{\alpha \mw \eps} + \frac{\poly(k)d^{3/2}}{\mw\epsilon}\right) \cdot \mathrm{polylog}\left( \frac{dkR(\maxsigma/\minsigma)}{\alpha \beta \eps \delta} \right)$$ 
    samples from an unknown mixture of $k$
    Gaussians $\cD$ in $\R^d$ satisfying
    \eqref{eq:separated_intro} and \eqref{eq:boundedness_intro},
    where $\mw = \min_{i} w_i$, and, with probability at least $1-\beta$, learns the parameters of $\cD$
    up to error $\alpha$.
\end{thm}

We remark that the sample complexity in Theorem~\ref{thm:main_intro} compares favorably to the sample complexity of methods based on subsample-and-aggregate.  In particular, when $\eps \geq \alpha$ and $k$ is a small polynomial in $d$, the sample complexity is dominated by $d^2/\alpha^2 \mw$, which is optimal even for non-private algorithms. In Section~\ref{sec:sa} we give an optimized version of the subsample-and-aggregate-based reduction from~\cite{NissimRS07} and show that we can learn mixtures of Gaussians with sample complexity roughly $\tilde{O}(\sqrt{d}/\eps)$ times the sample complexity of the corresponding non-private algorithm.  In contrast the sample complexity of our algorithm does not grow by dimension-dependent factors compared to the non-private algorithm on which it is based.

\medskip
At a high level, our algorithm mimics the approach of Achlioptas and McSherry~\cite{AchlioptasM05}, which is to use PCA to project the data into a low-dimensional space, which has the effect of projecting out much of the noise, and then recursively clustering the data points in that low-dimensional space.  However, where their algorithm uses a Kruskal-based clustering algorithm, we have to use alternative clustering methods that are more amenable to privacy. We develop our algorithm in two distinct phases addressing different aspects of the problem:

\medskip\noindent\emph{Phase I.} In Section~\ref{sec:easycase} we consider an
        ``easy case'' of Theorem~\ref{thm:main_intro},
        where we assume that: all components are
        spherical Gaussians, such that variances of
        each component lie in a small, known range
        (such that their ratio is bounded by a constant factor) and that
        the means of the Gaussians lie in a small
        ball around the origin. Under these assumptions,
        it is fairly straightforward to make the
        PCA-projection step~\cite{VempalaW02,AchlioptasM05}
        private. The key piece of the algorithm that
        needs to be private is computing the principal
        components of the data's covariance matrix.
        We can make this step private by adding
        appropriate noise to this covariance, and the
        key piece of the analysis is to analyze the
        effect of this noise on the principal components,
        extending the work of Dwork \etal~\cite{DworkTTZ14}
        on private PCA. Using the assumptions we make
        in this easy case, we can show that the projection
        shifts each component's mean by $O(\sqrt{k}\maxsigma)$,
        which preserves the separation of the data
        because all variances are within constant
        factor of one another.
        Then, we iteratively cluster the data using
        the $1$-cluster technique of~\cite{NissimSV16, NissimS18}. Lastly, we apply a simplified version of~\cite{KamathLSU19} (Appendix~\ref{sec:missingproofs-psge}) to learn each component's parameters.
        
\medskip\noindent\emph{Phase II.} We then consider the general case where the Gaussians can be non-spherical and wildly different from each other. In this case, if we directly add noise to the covariance matrix to achieve privacy, then the noise will scale polynomially with $\sigma_{\mathrm{max}}/\sigma_{\mathrm{min}}$, which is undesirable.  To deal with the general case, we develop a recursive algorithm, which repeatedly identifies an \emph{secluded} cluster in the data, and then recurses on this isolated cluster and the points outside of the cluster.   To that end we develop in Section~\ref{sec:TerrificBall} a variant of the private clustering algorithm of~\cite{NissimSV16,NissimS18} that finds a secluded ball---a set of points that lie inside of some ball $\ball{p}{r}$ such that the annulus $\ball{p}{10r}\setminus\ball{p}{r}$ is (essentially) empty.\footnote{Since~\cite{NissimSV16,NissimS18} call the ball found by their algorithm a \emph{good} ball, we call ours a \emph{terrific} ball.} 

    We obtain a recursive algorithm in the following way.  First we try to find a secluded ball in the unprojected data.  If we find one then we can split and recurse on the inside and outside of the ball.  If we cannot find a ball, then we can argue that the diameter of the dataset is $\mathrm{poly}(d,k,\sigma_{\max})$.  In the latter case, we can ensure that with $\mathrm{poly}(d,k)$ samples, the PCA-projection of the data preserves the mean of each component up to $O(\sqrt{k}\maxsigma)$, which guarantees that the cluster with the largest variance is secluded, so we can find the secluded ball and recurse.

\subsection{Related Work}
There has been a great deal of work on learning mixtures of distribution classes, particularly mixtures of Gaussians.
There are many ways the objective can be defined, including clustering~\cite{Dasgupta99, DasguptaS00, AroraK01, VempalaW02, AchlioptasM05, ChaudhuriR08a, ChaudhuriR08b, KumarK10, AwasthiS12, RegevV17, HopkinsL18, DiakonikolasKS18b, KothariSS18}, parameter estimation~\cite{KalaiMV10, MoitraV10, BelkinS10,HsuK13,AndersonBGRV14,BhaskaraCMV14, HardtP15,GeHK15, XuHM16, DaskalakisTZ17,AshtianiBHLMP18}, proper learning~\cite{FeldmanOS06,FeldmanOS08, DaskalakisK14, SureshOAJ14, DiakonikolasKKLMS16, LiS17}, and improper learning~\cite{ChanDSS14a}.

Some work on privately learning mixtures of Gaussians includes~\cite{NissimRS07} and~\cite{BunKSW19}. 
The former introduced the sample-and-aggregate method to convert non-private algorithms into private algorithms, and applied it to learning mixtures of Gaussians.
Our sample-and-aggregate-based method can be seen as a modernization of their algorithm, using tools developed over the last decade to handle somewhat more general settings and importantly, reduce the dependence on the range of the parameter space.
As discussed above, our main algorithm improves upon this approach by avoiding an increase in the dependence on the dimension, allowing us to match the sample complexity of the non-private algorithm in certain parameter regimes.
The latter paper (which is concurrent with the present work) provides a general method to convert from a \emph{cover} for a class of distributions to a private learning algorithm for the same class.
The work gets a near-optimal sample complexity of $\tilde O\left(kd^2\left(1/\alpha^2 + 1/\alpha \varepsilon\right)\right)$, but the algorithms have exponential running time in both $k$ and $d$ and their learning guarantees are incomparable to ours (the perform proper learning, while we do clustering and parameter estimation).

Other highly relevant works in private distribution learning include~\cite{KarwaV18,KamathLSU19}, which focus on learning a single Gaussian. 
There are also algorithms for learning structured univariate distributions in TV-distance~\cite{DiakonikolasHS15}, and learning arbitrary univariate distributions in Kolomogorov distance~\cite{BunNSV15}.  Upper and lower bounds for learning the mean of a product distribution over the hypercube in $\ell_\infty$-distance include~\cite{BlumDMN05, BunUV14, DworkMNS06, SteinkeU17a}.
\cite{AcharyaKSZ18} focuses on estimating properties of a distribution, rather than the distribution itself.
\cite{Smith11} gives an algorithm which allows one to estimate asymptotically normal statistics with minimal increase in the sample complexity.
%On the topic of private hypothesis testing, there has been a lot of work in both the asymptotic setting~\cite{VuS09,UhlerSF13,WangLK15,GaboardiLRV16,KiferR17,KakizakiSF17,CampbellBRG18,SwanbergGGRGB19,CouchKSBG19} and the minimax setting~\cite{CaiDK17,AcharyaSZ18,AliakbarpourDR18, AwanS18, CanonneKMSU19}.
There has also been a great deal of work on distribution learning in the local model of differential privacy~\cite{DuchiJW13,WangHWNXYLQ16,KairouzBR16,AcharyaSZ19, DuchiR18,JosephKMW18,YeB18,GaboardiRS19}.
%and hypothesis testing~\cite{GaboardiR18, Sheffet18, AcharyaCFT19}.

%Non-privately, there has been a significant amount of work on learning specific classes of distributions.
%The PAC-style formulation of the problem we consider originated in~\cite{KearnsMRRSS94}.
%While learning Gaussians and product distributions can be considered folklore at this point, some of the other classes we learn have enjoyed more recent study.
%For instance, learning sums of independent random variables was recently considered in~\cite{DaskalakisDS12b}, on the problem of learning Poisson Binomial Distributions (PBDs).
%Since then, there has been additional work on learning PBDs and various generalizations~\cite{DaskalakisKT15, DaskalakisDKT16, DiakonikolasKS16a, DiakonikolasKS16b, DiakonikolasKS16c,DeLS18}.
%
%Piecewise polynomials are a highly-expressive class of distributions, and they can be used to approximate a number of other univariate distribution classes, including distributions which are multi-modal, concave, convex, log-concave, monotone hazard rate, Gaussian, Poisson, Binomial, and more.
%Algorithms for learning such classes are considered in a number of papers, including~\cite{DaskalakisDS12a,ChanDSS14a,ChanDSS14b,AcharyaDK15, AcharyaDLS17}.

Within differential privacy, there are many algorithms for tasks that are related to learning mixtures of Gaussians, notably PCA~\cite{BlumDMN05,KapralovT13,ChaudhuriSS13,DworkTTZ14} and clustering~\cite{NissimRS07,GuptaLMRT10,NissimSV16,NissimS18,BalcanDLMZ17,StemmerK18,HuangL18}.  
Applying these algorithms na\"ively to the problem of learning Gaussian mixtures would necessarily introduce a polynomial dependence on the range of the data, which we seek to avoid.  
Nonetheless, private algorithms for PCA and clustering feature prominently in our solution, and build directly on these works.

%% file: preliminaries.tex
\section{Preliminaries}

\subsection{General Preliminaries}\label{subsec:general-prelims}
Let $\Sym_d^+$ denote set of all $d \times d$,
symmetric, and positive semidefinite matrices.
Let $\cG(d) = \{ \cN(\mu,\Sigma) :
    \mu \in \R^{d}, \Sigma \in \Sym_d^+\}$
be the family of $d$-dimensional
Gaussians. We can now define the
class $\cG(d,k)$ of mixtures of
Gaussians as follows.

\begin{defn} [Gaussian Mixtures]
    The class of \emph{Gaussian $k$-mixtures in $\R^{d}$} is
    $$\cG(d,k) \coloneqq \left\{\sum\limits_{i=1}^{k}{w_i G_i}:
       G_1,\dots,G_k \in \cG(d), w_1,\dots,w_k > 0,
       \sum_{i=1}^{k} w_i = 1 \right\}.$$
    We can specify a Gaussian mixture
    by a set of $k$ tuples as:
    $\{(\mu_1,\Sigma_1,w_1),\dots,(\mu_k,\Sigma_k,w_k)\},$
    where each tuple represents the mean,
    covariance matrix, and mixing weight
    of one of its components. Additionally,
    for each $i$, we refer to
    $\sigma_i^2 = \llnorm{\Sigma_i}$
    as the maximum directional variance
    of component $i$.
\end{defn}

We are given $n$ points in $d$ dimensions in
the form of a $(n\times d)$-matrix $X$, and
use $A$ to denote the corresponding matrix of
centers.
$(\alpha, \beta, \epsilon, \delta)$ are the
parameters corresponding to accuracy in
estimation (in total variation distance),
failure probability, and privacy parameters,
respectively. $R$ denotes the radius of a ball
(centered at the origin) which contains all
means, and $\kappa$ is the ratio of the upper
and the lower bound on the variances (for
simplicity, we generally assume the lower bound
is $1$). Also, $\maxsigma^2$ and $\minsigma^2$
are the maximum and minimum variance of
any single component, namely
$\sigma^2_{\max} = \max_i\{\sigma_i^2 \}$
and $\sigma^2_{\min}$ is defined symmetrically;
similarly $\mw$ denotes a lower bound on
the minimum mixing weight. 
We will use the notational convention that
$B^{d}_r(\vec{c})$ denotes the ball in $\R^d$
of radius $r$ centered at $\vec{c} \in \R^{d}$.
As $d$ will typically be clear from context,
we will often suppress $d$ and write $B_{r}(\vec{c})$.

In order to (privately) learn a Gaussian mixture,
we will need to impose two types of conditions on
its parameters---boundedness, and separation. 
\begin{defn}[Separated and Bounded Mixtures]
  \label{def:sep}
    For $s > 0$, a Gaussian mixtures $\cD \in \cG(d,k)$
    is \emph{$s$-separated} if
    $$\forall 1 \leq i < j \leq k,~~~
    \llnorm{\mu_i - \mu_j} \geq s\cdot\max\{\sigma_i,\sigma_j\}.$$
    For $R,\maxsigma,\minsigma,\mw > 0$, a Gaussian mixture
    $\cD \in \cG(d,k)$ is \emph{$(R,\minsigma,\maxsigma,\mw)$-bounded} if
    $$\forall 1 \leq i \leq k,~\| \mu_i \|_2 \leq R,~~~
        \min\limits_{i}\{\sigma_i\} \geq \minsigma,~~~
        \max\limits_{i}\{\sigma_i\} \leq \maxsigma,~~~
%        \min\limits_{i,j}\left\{\frac{\sigma^2_i}{\sigma^2_j}\right\}
%        \leq \kappa,~~~
        \text{and}~~~
        \min\limits_{i}\{w_i\} = \mw.$$
    We denote the family of separated
    and bounded Gaussian mixtures by
    $\cG(d,k,\minsigma,\maxsigma,R,\mw,s)$.
\end{defn}

We now have now established the necessary definitions
to define what it means to ``learn'' a Gaussian mixture
in our setting.
\begin{defn}[$(\alpha,\beta)$-Learning]
        \label{def:mixture-learning}
    Let $\cD \in \cG(d,k)$ be parameterized by
    $\{(\mu_1,\Sigma_1,w_1),\dots,(\mu_k,\Sigma_k,w_k)\}$.
    We say that an algorithm \emph{$(\alpha,\beta)$-learns}
    $\cD$, if on being given sample-access to $\cD$,
    it outputs with probablity at least $1-\beta$
    a distribution $\wh{\cD} \in \cG(d,k)$
    parameterized by $\{(\wh{\mu}_1,\wh{\Sigma}_1,\wh{w}_1),\dots,
    (\wh{\mu}_k,\wh{\Sigma}_k,\wh{w}_k)\}$,
    such that there exists a permutation
    $\pi:[k] \rightarrow [k]$, for which
    the following conditions hold.
    \begin{enumerate}
        \item For all $1 \leq i \leq k$,
            $\SD(\cN(\mu_i,\Sigma_i),
            \cN(\wh{\mu}_{\pi(i)},\wh{\Sigma}_{\pi(i)}))
            \leq O(\alpha)$.
        \item For all $1 \leq i \leq k$,
            $\abs{w_i - \wh{w}_{\pi(i)}} \leq
            O\left(\tfrac{\alpha}{k}\right)$.
    \end{enumerate}
    Note that the above two conditions together
    imply that $\SD(\cD,\wh{\cD}) \leq \alpha$.
\end{defn}

%Note that it is without loss of generality to
%assume that $\min_i \{\sigma_i\} = 1$, since
%we can just rescale the data to satisfy this
%condition, and adjust the parameters $s,R,\msigma$
%accordingly.

\subsubsection{Labelled Samples}
In our analysis, it will be useful to think
of sampling from a Gaussian mixture by the
following two-step process: first we select
a mixture component $\rho_i$ where $\rho_i = j$
with probability $w_j$, and then we choose a
sample $X_i \sim G_{\rho_i}$.  We can then
imagine that each point $X_i$ in the sample
has a \emph{label} $\rho_i$ indicating which
mixture component it was sampled from.

\begin{defn}[Labelled Sample]
    For $\cD \in \cG(d,k,\minsigma,\maxsigma,R,\mw,s)$,
    a \emph{labelled sample} is a set of tuples
    $X^L = ((X_1,\rho_1),\dots,(X_m,\rho_m))$
    sampled from $\cD$ according to the process
    above. We will write $X = (X_1,\dots,X_m)$
    to denote the (unlabelled) sample.
\end{defn}

We emphasize again that the algorithm does not
have access to the labelled sample $X^L$. In fact
given a fixed sample $X$, \emph{any} labelling has
non-zero probability of occurring, so from the
algorithm's perspective $X^L$ is not even well
defined.  Nonetheless, the labelled sample is a
well defined and useful construct for the analysis
of our algorithm, since it allows us to make sense
of the statement that the algorithm with high
probability correctly labels each point in the
sample by which mixture component it came from.

\subsubsection{Deterministic Regularity Conditions}

In order to analyze our algorithm, it will be
useful to establish several regularity conditions
that are satisfied (with high probability) by
samples from a Gaussian mixture. In this section
we will state the following regularity conditions.

The first condition asserts that each mixture
component is represented approximately the right
number of times.
\begin{condition} \label{cond:num-points}
    Let $X^{L} = ((X_1,\rho_1),\dots,(X_n,\rho_n))$ be
    a labelled sample from a Gaussian $k$-mixture $\cD$.
    For every label $1 \leq u \leq k$,
    \begin{enumerate}
        \item the number of points from  component $u$
            (i.e.~$| \{ X_i : \rho_i = u\} |)$ is in
            $[\tfrac{n w_u }{2} , \tfrac{3nw_u }{2}]$,
            and
        \item if $w_u \geq 4\alpha/9k$ then the number
            of points from  component $u$ is in
            $[n(w_u-\tfrac{\alpha}{9k}),n(w_u+\tfrac{\alpha}{9k})]$.
    \end{enumerate}
\end{condition}

%The final condition asserts that the empirical
%mean of the points in each mixture component
%is approximately correct.
%\begin{condition}\label{cond:gauss-mean-close}
%    Let $X^L = ((X_1,\rho_1),\dots,(X_n,\rho_n))$
%    be a labelled sample from a Gaussian mixture
%    $\cD$ with parameters $\{(\mu_1,\sigma_1,w_1),
%    \dots,(\mu_k,\sigma_k,w_k)\}$.  For each
%    $1 \leq u \leq k$, the empirical mean
%    $\wt{\mu}_i$ of the points with label $u$ is
%    within $\sqrt{d}/2$ of $\mu_u$.  That is, for
%    every $1 \leq u \leq k$, 
%    $$\llnorm{\mu_u - \wt{\mu}_u}
%        \leq \frac{\sqrt{d}}{2},$$
%    where
%    $\wt{\mu}_u =  (\sum_{X_i : \rho_i = u} X_i) / |\{ X_i : \rho_i = u \}|$.
%\end{condition}

In Appendix~\ref{sec:regularityproofs}, we prove
the following lemma, which states that if the number
of samples is sufficiently large, then with high
probability each of the above conditions is satisfied.
\begin{lemma} \label{lem:regularity-conditions-hold}
    Let $X^L = ((X_1,\rho_1),\dots,(X_n,\rho_n))$
    be a labelled sample from a Gaussian $k$-mixture
    $\cD \in \cG(d,k)$ with parameters
    $\{(\mu_1,\Sigma_1,w_1),\dots,(\mu_k,\Sigma_k,w_k)\}$.
    If
    $$n \geq \max\left\{\frac{12}{\mw}\ln(2k/\beta),
        \frac{405k^2}{2\alpha^2}\ln(2k/\beta)\right\},$$
    then with probability at least $1-\beta$, $X^L$
    (alternatively $X$) satisfies Condition
    \ref{cond:num-points}.
\end{lemma}

\subsection{Privacy Preliminaries}

In this section we review the basic definitions
of differential privacy, and develop the algorithmic
toolkit that we need.

\subsubsection{Differential Privacy}

%\ifnum\comments=1{\color{blue} \noindent Questions, comments, and action items:
%\begin{enumerate}
%\item What are datasets and neighboring datasets? %Gautam: done
%\item Do we ever want to use advanced composition to get better accuracy?  If so then it needs to be defined. %Gautam: Done, though I changed it to be a fixed epsilon
%\item The definition of composition didn't make sense.  I don't know what an ``adaptive composition of $M_1,\dots,M_{\tau}$'' is.  The whole point of adaptive composition is that the mechanisms are not fixed in advance. %Gautam: Maybe this was already done, but I think it's fine.
%\item The PCount thing needs to be completed. %Gautam: done
%\end{enumerate}
%}\fi

\begin{defn}[Differential Privacy (DP) \cite{DworkMNS06}]
    \label{def:dp}
    A randomized algorithm $M:\cX^n \rightarrow \cY$
    satisfies $(\eps,\delta)$- differential privacy
    ($(\eps,\delta)$-DP) if for every pair of
    neighboring datasets $X,X' \in \cX^n$ (i.e., datasets that differ in exactly one entry),
    $$\forall Y \subseteq \cY~~~
        \pr{}{M(X) \in Y} \leq e^{\eps}\cdot
        \pr{}{M(X') \in Y} + \delta.$$
\end{defn}

Two useful properties of differential privacy are
closure under post-processing and 
composition.
\begin{lemma}[Post Processing \cite{DworkMNS06}]
        \label{lem:post-processing}
    If $M:\cX^n \rightarrow \cY$ is
    $(\eps,\delta)$-DP, and $P:\cY \rightarrow \cZ$
    is any randomized function, then the algorithm
    $P \circ M$ is $(\eps,\delta)$-DP.
\end{lemma}

\begin{lemma}[Composition of DP \cite{DworkMNS06,DworkRV10}]
        \label{lem:composition}
    If $M_1,\dots,M_T$ are
    $(\eps_0,\delta_0)$-DP
    algorithms, and
    $$M(X) = (M_1(X),\dots,M_T(X))$$
    is the composition of these mechanisms, then $M$ is
    $(\eps,\delta)$-DP for
    \begin{itemize}
      \item (basic composition) $\eps = \eps_0T$ and $\delta = \delta_0T$;
      \item (advanced composition) $\eps = \eps_0\sqrt{2T \log(1/\delta')} + \eps_0(e^{\eps_0} -1)T$ and $\delta = \delta_0 T  + \delta'$, for any $\delta' > 0$.
    \end{itemize}
    Moreover, this property holds even if
    $M_1,\dots,M_T$ are chosen adaptively.
\end{lemma}

\subsubsection{Basic Differentially Private Mechanisms.}
We first state standard results on achieving
privacy via noise addition proportional to
sensitivity~\cite{DworkMNS06}.

\begin{defn}[Sensitivity]
    Let $f : \cX^n \to \R^d$ be a function,
    its \emph{$\ell_1$-sensitivity} and
    \emph{$\ell_2$-sensitivity} are
    $$\Delta_{f,1} = \max_{X \sim X' \in \cX^n} \| f(X) - f(X') \|_1
    ~~~~\textrm{and}~~~~\Delta_{f,2} = \max_{X \sim X' \in \cX^n} \| f(X) - f(X') \|_2$$
    respectively.
    Here, $X \sim X'$ denotes that $X$ and $X'$ are neighboring datasets (i.e., those that differ in exactly one entry).
\end{defn}

For functions with bounded $\ell_1$-sensitivity,
we can achieve $(\eps,0)$-DP by adding noise from
a Laplace distribution proportional to
$\ell_1$-sensitivity. For functions taking values
in $\R^d$ for large $d$ it is more useful to add
noise from a Gaussian distribution proportional
to the $\ell_2$-sensitivity, achieving $(\eps,\delta)$-DP.
\begin{lem}[Laplace Mechanism] \label{lem:laplacedp}
    Let $f : \cX^n \to \R^d$ be a function
    with $\ell_1$-sensitivity $\Delta_{f,1}$.
    Then the Laplace mechanism
    $$M(X) = f(X) + \Lap\left(\frac{\Delta_{f,1}}
        {\eps}\right)^{\otimes d}$$
    satisfies $(\eps,0)$-DP.
\end{lem}

One application of the Laplace mechanism is
private counting. The function $\PCount_{\eps}(X,T)$
is $\eps$-DP, and returns the
number of points of $X$ that lie
in $T$, i.e., $|X \cap T| + \Lap(1/\eps)$.
The following is immediate since the statistic $|X \cap T|$ is $1$-sensitive, and by tail bounds on Laplace random variables.
\begin{lemma}
  $\PCount$ is $\eps$-DP, and with probability at least $1 - \beta$, outputs an estimate of $|X \cap T|$ which is accurate up to an additive $O\left(\frac{\log (1/\beta)}{\varepsilon}\right)$.
\end{lemma}

\begin{lem}[Gaussian Mechanism] \label{lem:gaussiandp}
    Let $f : \cX^n \to \R^d$ be a function
    with $\ell_2$-sensitivity $\Delta_{f,2}$.
    Then the Gaussian mechanism
    $$M(X) = f(X) + \cN\left(0,\left(\frac{\Delta_{f,2}
        \sqrt{2\ln(2/\delta)}}{\eps}\right)^2 \cdot \id_{d \times d}\right)$$
    satisfies $(\eps,\delta)$-DP.
\end{lem}

%\gnote{Also, this needs to be integrated into the actual body later, which says we do SVT, if we instead use AboveThreshold.}
Now consider a setting where there are
multiple queries to be performed on a dataset,
but we only want to know if there is a query
whose answer on the dataset lies above a certain
threshold. We introduce a differentially private
algorithm from \cite{DworkR14,DworkNRRV09} which does that.
\begin{theorem}[Above Threshold]\label{thm:above-threshold}
    Suppose we are given a dataset $X$, a sequence of queries
    $f_1,\dots,f_t$, with $\ell_1$-sensitivity
    $\Delta$, and a threshold $T$. There
    exists an $\eps$-differentially private
    algorithm $\AboveThreshold_{\eps}$ that outputs
    a stream of answers $a_1,\dots,a_{t'}\in\{\bot,\top\}$,
    where $t' \leq t$, such that if $a_{t'} = \top$,
    then $a_1,\dots,a_{t'-1} = \bot$. Then
    the following holds with probability at
    least $1-\beta$.
    If $a_{t'} = \top$, then
    $$f_{t'}(X) \geq T - \Gamma,$$
    and for all $1 \leq i \leq t'$, if
    $a_{i} = \bot$, then
    $$f_i(X) \leq T + \Gamma,$$
    where
    $$\Gamma = \frac{8\Delta(\log{t} + \log(2/\beta))}{\eps}.$$
\end{theorem}

\subsection{Technical Preliminaries}

\begin{lem}[Hanson-Wright inequality~\cite{HansonW71}]
  \label{lem:HW}
  Let $X \sim \cN(0,\mathbb{I}_{d \times d})$ and let $A$ be a $d \times d$ matrix.
  Then for all $t > 0$, the following two bounds hold:
  \[
    \pr{}{X^TAX - \tr(A) \geq 2 \|A\|_F \sqrt{t} + 2\|A\|_2t} \leq \exp(-t);
  \]
  \[
    \pr{}{X^TAX - \tr(A) \leq -2 \|A\|_F \sqrt{t} } \leq \exp(-t).
  \]
\end{lem}
\iffalse
\begin{lem}[Follows from Lemma 4.4 of~\cite{DiakonikolasKKLMS16}]
	\label{lem:meanofgaussiansubset}
    Let $X_1, \dots, X_n$ be i.i.d.\ samples from
    $\cN(0, \mathbb{I}_{d \times d})$, where
    $n = \Omega\left(\frac{d + \log(1/\beta)}{\tau^2}\right)$,
    for $\tau > 0$ less than some absolute constant.
    Then with probability at least $1 - \beta$,
    simultaneously for all $\mathcal{S} \subseteq [n]$
    such that $|\mathcal{S}| \geq (1 - \tau)n$, we have
    that
    $$\left\|\frac{1}{|S|}\sum_{i \in \mathcal{S}} X_i
        \right\|_2 \leq O\left(\tau \sqrt{\log(1/\tau)}\right).$$
\end{lem}
\fi
The following are standard concentration results
for the empirical mean and covariance of a set of
Gaussian vectors (see, e.g.,~\cite{DiakonikolasKKLMS16}).
\begin{lem}\label{lem:spectralnormofgaussians}
    Let $X_1, \dots, X_n$ be i.i.d.\ samples
    from $\cN(0, \mathbb{I}_{d \times d})$.
    Then we have that 
    \[
        \pr{}{\left\|\frac{1}{n} \sum_{i \in [n]} X_i \right\|_2 \geq t}
            \leq 4\exp(c_1 d - c_2 nt^2);
    \]
    \[
        \pr{}{\left\|\frac{1}{n} \sum_{i \in [n]} X_iX_i^T  - I \right\|_2 \geq t}
            \leq 4\exp(c_3 d - c_4 n\min(t,t^2)),
    \]
    where $c_1, c_2, c_3, c_4 >0$ are some absolute constants.
\end{lem}

We finally have a lemma that translates closeness
of Gaussian distributions from one metric to another.
It is a combination of Corollaries 2.13 and 2.14
of \cite{DiakonikolasKKLMS16}.
\begin{lemma}\label{lem:gaussian-norm-translation}
    Let $\alpha > 0$ be smaller than some absolute
    constant. Suppose that
    $$\llnorm{\Sigma^{-1/2}(\mu - \wh{\mu})} \leq O(\alpha)~~~
        \text{and}~~~
        \fnorm{\id - \Sigma^{-1/2}\wh{\Sigma}\Sigma^{-1/2}} \leq O(\alpha),$$
    where $\cN(\mu,\Sigma)$ is a Gaussian distribution
    in $\R^d$, $\mu \in \R^d$, and $\wh{\Sigma} \in \R^{d \times d}$
    is a PSD matrix. Then
    $$\SD\left(\cN(\mu,\Sigma),\cN(\wh{\mu},\wh{\Sigma})\right)
        \leq \alpha.$$
\end{lemma}

%% file: pca.tex
\section{Robustness of PCA-Projection to Noise}
\label{sec:PCA}

One of the main tools used in learning mixtures of Gaussians under separation is Principal Component Analysis (PCA).
In particular, it is common to project onto the top $k$ principal components (a subspace which will contain the means of the components).
In some sense, this eliminates directions which do not contain meaningful information while preserving the distance between the means, thus allowing us to cluster with separation based on the ``true'' dimension of the data, $k$, rather than the ambient dimension $d$.
In this section, we show that a similar statement holds, even after perturbations required for privacy.

%Specifically, we show that the perturbed (noisy) $k$-PCA projection of the data is useful for the purpose of approximating the means of the distributions. 
Before showing the result for perturbed PCA, we reiterate the (very simple) proof of Achlioptas and McSherry~\cite{AchlioptasM05}. Fixing a cluster $i$, denoting its empirical mean as $\bar \mu_i$, the mean of the resulting projection as $\hat \mu_i$, $\Pi$ as the $k$-PCA projection matrix, $u_i\in \{0,1\}^n$ as the vector indicating which datapoint was sampled from cluster $i$, and $n_i$ as the number of datapoints which were sampled from cluster $i$, we have
\[  \|\bar\mu_i -\hat\mu_i \|_2 = \|\tfrac 1 {n_i} \left(X^T- (X\Pi)^T \right)u_i \|_2 \leq \|X- X\Pi\|_2 \tfrac{\|u_i\|_2}{n_i} \leq \tfrac 1 {\sqrt{n_i}} \|X-A\|_2,  \] where
the last inequality follows from the $X\Pi$ being the best $k$-rank approximation of $X$ whereas $A$ is any rank-$k$ matrix. 
In particular, we could choose $A$ to be the matrix where each row of $X$ is replaced by the (unknown) center of the component which generated it (as we do in Lemma~\ref{lem:PCA_AM_style} below).
We now extend this result to a perturbed $k$-PCA projection as given by the following lemma.

\begin{lemma}\label{lem:PCA_AM_style}
	Let $X\in \R^{n\times d}$ be a collection of
    $n$ datapoints from $k$ clusters each centered
    at $\mu_1, \mu_2,...,\mu_k$. Let $A\in\R^{n\times d}$
    be the corresponding matrix of (unknown) centers
    (for each $j$ we place the center $\mu_{c(j)}$
    with $c(j)$ denoting the clustering point $X_j$
    belongs to). Let $\Pi_{V_k}\in \R^{d\times d}$
    denote the $k$-PCA projection of $X$'s rows. Let
    $\Pi_U\in \R^{d\times d}$ be a projection such
    that for some bound $B\geq 0$ it holds that
    $\|X^T X - (X\Pi_U)^T(X\Pi_U)\|_2 \leq
    \|X^T X - (X\Pi_{V_k})^T(X\Pi_{V_k})\|_2 + B$.
    Denote $\bar{\mu_i}$ as the empirical mean of all
    points in cluster $i$ and denote $\hat{\mu_i}$ as
    the projection of the empirical mean
    $\hat{\mu_i} = \Pi_U \bar{\mu_i}$. Then
	\[ \|\bar{\mu_i} - \hat{\mu_i}\|_2 \leq
        \tfrac 1 {\sqrt {n_i}}\|X-A\|_2+ \sqrt{\tfrac B {n_i}}  \] 
\end{lemma}
\begin{proof}
    Fix $i$. Denote $u_i\in \R^{n}$ as the indicating
    vector of all datapoints in $X$ that belong to
    cluster $i$. Assuming there are $n_i$ points from
    cluster $i$ it follows that $\|u_i\|^2 = n_i$, thus
    $\tfrac 1 {\sqrt n_i} u_i$ is a unit-length vector.
    Following the inequality $\tr(AB)\leq \|A\|_2 \|B\|_F$
    for PSD matrices, we thus have that
    %\bigvnote{slightly modified version}
    %\gnote{Something seems funny with parts of this proof}
    %\bigonote{We've talked about this Gautam: $\Pi$ is a projection matrix, thus it is (a) symmetric and (b) equal to its square $\Pi = \Pi^2$. Therefore the cross-terms vanish. Moreover, an alternative view is to see that $\id-\Pi_U = \Pi_{U^\perp}$ which is also a projection matrix.}
    \begin{align*}
        \llnorm{\bar{\mu_i}-\hat{\mu_i}}^2 &=
                \llnorm{\frac{1}{n_i} X^T u_i-\frac{1}{n_i} (X\Pi_U)^T u_i}^2 =
                \frac{1}{n_i^2} \llnorm{\left( X(\id-\Pi_U)\right)^T u_i}^2\\
            &\leq \frac{1}{n_i^2}\llnorm{\left(X(\id-\Pi_U)\right)^T}^2\llnorm{u_i}^2
                \stackrel{(\ast)}=
                \frac{1}{n_i^2}\llnorm{\left(X(\id-\Pi_U)\right)^T\left(X(\id-\Pi_U)\right)}
                \cdot n_i\\
            &\leq \frac{1}{n_i} \left( \llnorm{\left( X(\id-\Pi_{V_k}) \right)^T
                \left( X(\id-\Pi_{V_k}) \right)} + B \right) =
                \frac{1}{n_i} \left( \| X(\id-\Pi_{V_k}) \left(X(\id-\Pi_{V_k}) \right)^T
                \|_2 + B \right)\\
            &= \frac{1}{n_i} \left(\llnorm{XX^T - X\Pi_{V_k}\Pi_{V_k}^TX^T} + B\right)
                \stackrel{(\ast\ast)}= \frac{1}{n_i} \left(\llnorm{X - X\Pi_{V_k}}^2 + B\right)
                \stackrel{(\ast\ast\ast)}\leq \frac{1}{n_i} \left(\llnorm{X - A}^2 + B\right)
    \end{align*}
%\begin{align*}
%    \|\bar{\mu_i}-\hat{\mu_i}\|^2_2 &=
%    \|\tfrac 1 {n_i} X^T u_i-\tfrac 1 {n_i}
%    (X\Pi_U)^T u_i\|^2_2 = \tfrac 1 {n_i^2} \|
%    \left( X(I-\Pi_U)  \right)^T u_i\|^2_2
%    \cr & = \tfrac 1 {n_i^2} u_i^T \left( X(I-\Pi_U) \right)
%    \left( X(I-\Pi_U)  \right)^T u_i
%    = \tfrac 1 {n_i^2}\tr\left( \left( X(I-\Pi_U) \right)
%    \left( X(I-\Pi_U)  \right)^T u_i u_i^T \right) 
%    \cr &\leq \left\| \left( X(I-\Pi_U)  \right)
%    \left( X(I-\Pi_U)  \right)^T \right\|_2\cdot
%    \tfrac 1 {n_i^2} \|u_i u_i^T \|_F
%    \stackrel{(\ast)}=  \left\| \left( X(I-\Pi_U)\right)^T
%    \left( X(I-\Pi_U)  \right) \right\|_2\cdot \tfrac {n_i} {n_i^2}
%    \cr &  \leq \frac 1 {n_i} \left( \|\left( X(I-\Pi_{V_k})\right)^T
%    \left( X(I-\Pi_{V_k})  \right)\|_2 + B \right) =
%    \frac 1 {n_i} \left( \|  X(I-\Pi_{V_k}) \left(X(I-\Pi_{V_k}) \right)^T
%    \|_2 + B \right)
%    \cr &= \tfrac 1 {n_i} \left(\|XX^T - X\Pi_{V_k}\Pi_{V_k}^TX^T\|_2 + B\right)
%    \stackrel{(\ast\ast)}= \tfrac 1 {n_i} \left(\|X - X\Pi_{V_k}\|_2^2 + B\right)
%    \stackrel{(\ast\ast\ast)}\leq \tfrac 1 {n_i} \left(\|X - A\|_2^2 + B\right)
%\end{align*}
where the equality marked with $(\ast)$ follows from the fact that for any matrix $M$ we have $\|MM^T\|_2=\|M^T M\|_2=\sigma_1(M)^2$ with $\sigma_1(M)$ denoting $M$'s largest singular value; the equality marked with $(\ast\ast)$ follows from the fact that for any matrix $M$ we have $\|MM^T- M\Pi_{V_k}M\|_2=\sigma_{k+1}(M)^2$ with $\sigma_k(M)$ denoting $M$'s $(k+1)$th-largest singular value; and the inequality $(\ast\ast\ast)$ follows from the fact that $A$ is a rank-$k$ matrix. 
  In this proof, we also used the fact that $\Pi$ is a projection matrix, implying that it is symmetric and equal to its square -- these facts allow us to cancel various ``cross terms.'' 
  The inequality $\sqrt{a+b}\leq \sqrt a + \sqrt b$ concludes the proof.
\end{proof}

The above is a general statement for any type of clustering problem, which we instantiate in the following lemma for mixtures of Gaussians.

\begin{lemma}\label{lem:data-spectral}
    Let $X \in \R^{n \times d}$ be a sample
    from $\cD \in \cG(d,k)$,
    and let $A \in \R^{n \times d}$ be the matrix
    where each row $i$ is the (unknown) mean
    of the Gaussian from which $X_i$ was sampled.
    For each $i$, let $\sigma^2_i$ denote the
    maximum directional variance of component
    $i$, and $w_i$ denote its mixing weight.
    Define $\sigma^2 = \max\limits_{i}\{\sigma^2_i\}$
    and $\mw = \min\limits_{i}\{w_i\}$. If
    $$n \geq \frac{1}{\mw} \left(\xi_1 d +
        \xi_2\log\left(\frac{2k}{\beta}\right)\right),$$
    where $\xi_1,\xi_2$ are universal constants,
    then with probability at least $1-\beta$,
    $$\frac{\sqrt{n\mw}\sigma}{4} \leq \llnorm{X-A}
        \leq 4\sqrt{n\sum\limits_{i=1}^{k}{w_i\sigma^2_i}}.$$
\end{lemma}
\begin{proof}
    Let $C_i \in \R^{n_i \times d}$ be the matrix formed by concatenating the rows drawn from Gaussian component $i$, and let $A_i$ be the corresponding operation applied to $A$.
    Let $\Sigma_i$ denote the covariance matrix of Gaussian component $i$.

    We first prove the lower bound on the norm.
    Let $C^*$ be the matrix in  $\{C_1,\dots,C_k\}$ with the largest direction variance (i.e., $C^* = C_i$, where $i = \arg\max_i {\sigma_i^2}$),
    let $A^*$ be the submatrix of $A$ corresponding to the same rows of $C^*$, and let $\Sigma$ be the covariance matrix of the Gaussian component corresponding to these rows.
    Then each
    row of $(C^*-A^*)\Sigma^{-\tfrac{1}{2}}$ is
    an independent sample from
    $\cN(\vec{0},\id_{d \times d})$.
    
    We know that the number of rows in $C^*$ is
    at least $\tfrac{n\mw}{2}$. Using
    Cauchy-Schwarz inequality, Theorem 5.39 of
    \cite{Vershynin12}, and our bound on $n$, we get
    that with probability at least $1-\tfrac{\beta}{k}$.
    \begin{align*}
        \llnorm{C^*-A^*}\llnorm{\Sigma^{-\frac{1}{2}}} &\geq
                \llnorm{(C^*-A^*)\Sigma^{-\frac{1}{2}}}\\
            &\geq \sqrt{\frac{n\mw}{2}} - C_1\sqrt{d} -
                C_2\sqrt{\log\left(\frac{2k}{\beta}\right)}\\
            &\geq \frac{\sqrt{n\mw}}{4},
    \end{align*}
    where $C_1$ and $C_2$ are absolute constants.
    Since $\llnorm{\Sigma^{\tfrac{1}{2}}} = \sigma$,
    we get,
    $$\llnorm{C^*-A^*} \geq \frac{\sqrt{n\mw}\sigma}{4}.$$
    Since, the spectral norm of $X-A$ has to
    be at least the spectral norm of $C^*-A^*$,
    the lower bound holds.

    Now, we prove the upper bound. For each
    $i$, the number of points in the submatrix
    $C_i$ is at most $\tfrac{3nw_i}{2}$.
    Using Cauchy-Schwarz inequality, Theorem
    5.39 of \cite{Vershynin12}, and our bound
    on $n$ again, we get the following with
    probability at least $1-\tfrac{\beta}{k}$.
    \begin{align*}
        \llnorm{C_i - A_i} &=
                \llnorm{(C_i - A_i)\Sigma_i^{-\frac{1}{2}}\Sigma_i^{\frac{1}{2}}}\\
            &\leq \llnorm{(C_i - A_i)\Sigma_i^{-\frac{1}{2}}}
                \llnorm{\Sigma_i^{\frac{1}{2}}}\\
            &\leq \left(\sqrt{\frac{3nw_i}{2}} + C_1\sqrt{d} +
                C_2\sqrt{\log\left(\frac{2k}{\beta}\right)}\right)
                \llnorm{\Sigma_i^{\frac{1}{2}}}\\
            &\leq 4\sqrt{nw_i}\sigma_i
    \end{align*}
    Now,
    \begin{align*}
        \llnorm{X-A}^2 &= \llnorm{(X-A)^T(X-A)}\\
            &= \llnorm{\sum\limits_{i=1}^{k}{(C_i-A_i)^T(C_i-A_i)}}\\
            &\leq \sum\limits_{i=1}^{k}{\llnorm{(C_i-A_i)^T(C_i-A_i)}}\\
            &\leq 16n\sum\limits_{i=1}^{k}{w_i\sigma^2_i}.
    \end{align*}
    The second equality can be seen by noting that each entry of $(X-A)^T(X-A)$ is the inner product of two columns of $X-A$ -- by grouping terms in this inner product, it can be be seen as the sum of inner products of the corresponding columns of $C_i - A_i$, since the indices form a partition of the rows.
    The first inequality is the triangle inequality.

    We finally apply the union bound over
    all $i$ to complete the proof.
\end{proof}

%% file: easycase.tex
\section{A Warm Up: Strongly Bounded Spherical Gaussian Mixtures}\label{sec:easycase}

We first give an algorithm to learn mixtures
of spherical Gaussians, whose means lie in a
small ball, whose variances are within
constant factor of one another, and whose
mixing weights are identical. Before we get
to that, we formally define such a family of
mixtures of spherical Gaussians. Let
$$\cS(d) = \{ \cN(\mu,\sigma^2
    \mathbb{I}_{d \times d}) :
    \mu \in \R^{d}, \sigma^2 > 0\}$$
be the family of $d$-dimensional
spherical Gaussians. As before, we can define
the class $\cS(d,k)$ of mixtures of spherical
Gaussians as follows.
\begin{defn} [Spherical Gaussian Mixtures]
    The class of \emph{Spherical Gaussian $k$-mixtures in $\R^{d}$} is
    $$\cS(d,k) \coloneqq \left\{\sum\limits_{i=1}^{k}{w_i G_i}:
       G_1,\dots,G_k \in \cS(d),  \sum_{i=1}^{k} w_i = 1 \right\}$$
\end{defn}

Again, we will need to impose two types of conditions
on its parameters---boundedness and separation---that
are defined in slightly different ways from our
initial definitions in Section~\ref{subsec:general-prelims}.
We also introduce another condition that says
that all mixing weights are equal.
\begin{defn}[Separated, Bounded, and Uniform (Spherical) Mixtures]
    For $s > 0$, a spherical Gaussian mixture
    $\cD \in \cS(d,k)$ is \emph{$s$-separated} if
    $$\forall 1 \leq i < j \leq k,~~~
    \| \mu_i - \mu_j \|_2 \geq s\cdot\max\{\sigma_i,\sigma_j\}.$$
    For $R,\maxsigma,\mw > 0$, a Gaussian mixture $\cD \in \cS(d,k)$
    is \emph{$(R,\minsigma,\kappa)$-bounded} if
    $$\forall 1 \leq i \leq k,~\| \mu_i \|_2 \leq R,~~~
        \min\limits_{i}\{\sigma_i\} = \minsigma,~~~
        \text{and}~~~
        \max\limits_{i,j}\left\{\frac{\sigma^2_i}{\sigma^2_j}\right\} \leq \kappa.$$
    A Gaussian mixture $\cD \in \cS(d,k)$ is
    \emph{uniform} if
    $$\forall 1 \leq i \leq k,~ w_i = \frac{1}{k}.$$
    We denote the family of separated,
    bounded, and uniform spherical Gaussian mixtures
    by $\cS(d,k,\minsigma,\kappa,s,R)$.
    We can specify a Gaussian mixture
    in this family by a set of $k$ tuples as:
    $\{(\mu_1,\sigma_1),\dots,(\mu_k,\sigma_k)\},$
    where each tuple represents the mean
    and standard deviation of one of its
    components.
\end{defn}

\begin{defn}
    We define the following family of separated,
    bounded, and uniform mixtures of spherical Gaussians
    that have similar variances and lie in a small
    ball around the origin.
    $$\cS(d,k,\kappa,s) \equiv
        \bigcup\limits_{\minsigma > 0}
        {\cS(d,k,\minsigma,\kappa,s,k\sqrt{d\kappa}\minsigma)}$$
\end{defn}
We define the quantity $s$ in the
statement of the main theorem of this
section. The above definition could be generalized
to have the means lie in a small ball that
is not centered at the origin, but because
it is easy to privately find a small ball
that would contain all the points, we can
omit that for simplicity.

\begin{algorithm}[h!] \label{alg:pegm}
\caption{Private Easy Gaussian Mixture Estimator
    $\PEGME_{\eps, \delta, \alpha, \beta,
        \kappa, \minsigma}(X)$}
\KwIn{Samples $X_1,\dots,X_{2n} \in \R^{d}$.
    Ratio of maximum and minimum variances: $\kappa$.
    Minimum variance of a mixture component: $\minsigma^2$.
    Parameters $\eps, \delta, \alpha, \beta > 0$.}
\KwOut{A mixture of Gaussians $\wh{G}$, such that
    $\SD(G, \wh{G}) \leq \alpha$.} \vspace{10pt}

Set parameters:
$
    \Lambda \gets 2k\sqrt{d\kappa}\minsigma~~~
    \Delta_{\eps,\delta} \gets
        \frac{2\Lambda^2\sqrt{2\ln(1.25/\delta)}}{\eps}~~~
        \ell \gets \max\left\{k,
        O\left(\log\left(\frac{n}{\beta}\right)\right)\right\}
$\\\vspace{10pt}

Throw away all $X_i \in X$ such that
    $\llnorm{X_i} > \Lambda$, and call this
    new dataset $X$ as well\\
Let $Y \gets (X_1,\dots,X_n)$ and
    $Z \gets (X_{n+1},\dots,X_{2n})$
\vspace{10pt}

\tcp{Privately run PCA}
Let $E \in \R^{d \times d}$ be a symmetric matrix,
    where each entry $E_{i,j}$ for $j \geq i$ is an independent draw from
    from $\cN(0, \Delta^2_{\eps,\delta})$\\
Let $\wh{V}_{\ell}$ be the $\ell$-dimensional principal singular
    subspace of $Y^T Y + E$\\
Project points of $Z$ on to $\wh{V}_{\ell}$ to get
    the set $Z'_{\ell}$\\
Rotate the space to align with the axes of $\wh{V}_{\ell}$
    to get the set $Z_{\ell}$ from $Z'_{\ell}$
\vspace{10pt}

\tcp{Privately locate individual components}
Let $S_{\ell} \gets Z_{\ell}$ and $i = 1$\\
\While{$S_{\ell} \neq \emptyset$ and $i \leq k$}{
    $(c_i, r'_i) \gets \PGLoc(S_{\ell},\frac{n}{2k};
        \frac{\eps}{O(\sqrt{k\ln(1/\delta)})},
        \frac{\delta}{2k},R+8\sqrt{\ell\kappa}\minsigma,\minsigma,\sqrt{\kappa}\minsigma)$\\
    Let $r_i \gets 4\sqrt{3}r'_i$
        and $S_{i} \gets S_{\ell} \cap \ball{c_i}{r_i}$\\
%    $\wh{n}_i \gets \PCount_{\frac{\eps}{6k}}(S_v, \ball{\wh{c}_i}{r_i})$\\
    $S_{\ell} \gets S_{\ell} \setminus S_{i}$\\
    $i \gets i + 1$
}
\If{$\abs{C} < k$}{
    \Return $\bot$
}
\For{$i \gets 1,\dots,k$}{
    Rotate $c_i$ back to get $\wh{c}_i$\\
    Set $\wh{r}_i \gets r_i + 10\sqrt{\ell\kappa}\minsigma
        + 2r_i\sqrt{\frac{3d}{\ell}}$\\
    Let $\wh{S}_i$ be points in
        $Z \cap \ball{\wh{c}_i}{\wh{r}_i}$, whose
        corresponding points lie in
        $Z_{\ell} \cap \ball{c_i}{r_i}$
}
\vspace{10pt}
%Let $\wh{n} \gets \sum\limits_{i=1}^{k}{\wh{n}_i}$\\

\tcp{Privately estimate each Gaussian}
\For{$i \gets 1,\dots,k$}{
    $(\wh{\mu}_i, \wh{\sigma}^2_i) \gets
        \PSGE(\wh{S}_i;
        \wh{c}_i,\wh{r}_i,\eps,\delta)$\\
%    $\wh{w}_i \gets \tfrac{\wh{n}_i}{\wh{n}}$\\
%    $\wh{G} \gets \wh{G} \cup \{(\wh{\mu}_i, \wh{\sigma}^2_i\id_{d \times d},
%        \wh{w}_i)\}$\\
    $\wh{G} \gets \wh{G} \cup
        \{(\wh{\mu}_i, \wh{\sigma}^2_i, \frac{1}{k})\}$\\
}

\Return $\wh{G}$
\vspace{10pt}
\end{algorithm}

The following is our main theorem for this
sub-class of mixtures, which quantifies the
guarantees of Algorithm~\ref{alg:pegm}.
\begin{theorem}\label{thm:easycase}
    There exists an $(\eps,\delta)$-differentially
    private algorithm, which if given $n$ independent
    samples from $\cD \in \cS(d,k,\kappa,C\sqrt{\ell})$,
    such that $C = \xi + 16\sqrt{\kappa}$,
    where $\xi,\kappa \in \Theta(1)$
    and $\xi$ is a universal constant,
    $\ell=\max\{512\ln(nk/\beta),k\}$, and
    $$n \geq O\left(
        \frac{dk}{\alpha^2}
        + \frac{d^{\frac{3}{2}}k^3\sqrt{\ln(1/\delta)}}{\eps}
        + \frac{dk\sqrt{\ln(1/\delta)}}{\alpha\eps}
        + \frac{\sqrt{d}k\ln(k/\beta)}{\alpha\eps}
        \right) + n',$$
    where
    $$n' \geq O\left(
        \frac{k^2}{\alpha^2}\ln(k/\beta)
        + \frac{\ell^{\frac{5}{9}}k^{\frac{5}{3}}}
            {\eps^{\frac{10}{9}}}
            \cdot \polylog\left(\ell,k,
            \frac{1}{\eps},\frac{1}{\delta},
            \frac{1}{\beta}\right)
            \right),$$
    %$$n \geq n' + O\left(\frac{d}{\alpha^2}
    %    + \frac{d\sqrt{\ln(1/\delta)}}{\alpha\eps}
    %    + \frac{\sqrt{d}\ln(k/\beta)}{\alpha\eps}
    %    + \frac{\ln(k/\beta)}{\alpha^2}\right),$$
    %where
    %$$n' \geq O\left(\frac{d^{\frac{3}{2}}k^3\sqrt{\ln(1/\beta)}}{\eps}
    %    + \frac{k^2}{\alpha^2}\ln(k/\beta)
    %    + \frac{\ell^{\frac{5}{9}}k^{\frac{30}{9}}}
    %        {\eps^{\frac{10}{9}}}
    %        \cdot \polylog\left(\ell,k,
    %        \frac{1}{\eps},\frac{1}{\delta},
    %        \frac{1}{\beta}\right)\right),$$
    then it $(\alpha,\beta)$-learns $\cD$.
\end{theorem}

The algorithm itself is fairly simple to describe. 
First, we run a private version of PCA, and project to the top $k$ PCA directions.
By Lemma~\ref{lem:PCA_AM_style}, this will reduce the dimension from $d$ to $k$ while (approximately) preserving the separation condition.
Next, we repeatedly run an algorithm which (privately) finds a small ball containing many points (essentially the 1-cluster algorithm of~\cite{NissimS18}) in order to cluster the points such that all points generated from a single Gaussian lie in the same cluster.
Finally, for each cluster, we privately estimate the mean and variance of the corresponding Gaussian component.

\subsection{Privacy}

We will first analyze individual components
of the algorithm, and then use composition
(Lemma~\ref{lem:composition}) to reason
about privacy.

The PCA section of the algorithm is $(\eps,\delta)$-DP
with respect to $Y$. This holds because the
$\ell_2$ sensitivity of the function $Y^TY$ is
$2\Lambda^2$, because all points are guaranteed
to lie in a ball of radius $\Lambda$ around the
origin, and by Lemma \ref{lem:gaussiandp},
we know that adding Gaussian noise proportional
to $\Delta_{\eps,\delta}$ is enough to have
$(\eps,\delta)$-DP.

In the second step, we run $\PGLoc$ (which
is $(\eps,\delta)$-differentially private) on
$Z$ repeatedly with parameters $({\eps}/{O(\sqrt{k\ln(1/\delta)})},
{\delta}/{2})$, after which we only perform computation on the output of this private algorithm
So, the whole
process, by advanced composition and post-processing
(Lemmata~\ref{lem:composition} and \ref{lem:post-processing}),
is $(\eps,\delta)$-DP with respect to $Z$.

In the final step, we apply $\PSGE$ (which
is $(\eps,\delta)$-differentially private) with
parameters $\eps,\delta$ on disjoint datasets
$\wh{S}_i$. Therefore, this step is
$(\eps,\delta)$ private with respect to $Z$.

Finally, applying the composition lemma again,
we have $(2\eps,2\delta)$-DP for $Z$.
Combined with $(\eps, \delta)$-DP for $Y$, and the fact that $X$ is the union of these two disjoint sets, we have $(2\eps, 2\delta)$-DP for $X$.
Rescaling the values of $\eps$ and $\delta$ by $2$ gives the desired result.

\subsection{Accuracy}
As indicated in our outline above, 
the algorithm is composed of three blocks:
private PCA, isolating individual Gaussians,
and learning the isolated Gaussians. We divide
the proof of accuracy in a similar way.

\subsubsection{Deterministic Regularity Conditions
    for Spherical Gaussians}

We first give two regularity conditions for
mixtures of spherical Gaussians in the family
mentioned above.

The first condition asserts that each mixture
component is represented by points that lie in
a ball of approximately the right radius.
\begin{condition}\label{cond:easy-intra-gaussian}
    Let $X^L = ((X_1,\rho_1),\dots,(X_n,\rho_n))$
    be a labelled sample from a Gaussian mixture
    $\cD \in \cS(\ell,k,\kappa,s)$,
    where $\ell \geq 512\ln(nk/\beta)$ and $s > 0$.
    For every $1 \leq u \leq k$, the radius
    of the smallest ball containing the set of
    points with label $u$ (i.e.~$\{X_i : \rho_i = u\}$)
    is in $[\sqrt{\ell}\sigma_u/2, \sqrt{3\ell}\sigma_u]$.
\end{condition}

The second condition says that if the means
of the Gaussians are ``far enough'', then
the inter-component distance (between points
from different components) would also be large.
\begin{condition}\label{cond:easy-inter-gaussian}
    Let $X^{L} = ((X_1,\rho_1),\dots,(X_n,\rho_n))$
    be a labelled sample from a Gaussian mixture
    $\cD \in \cS(\ell,k,\kappa,C\sqrt{\ell})$,
    where $\ell \geq \max\{512\ln(nk/\beta),k\}$
    and $C > 1$ is a constant. For every
    $\rho_i \neq \rho_j$,
    $$\llnorm{X_i - X_j} \geq \frac{C}{2}
        \sqrt{\ell} \max\{\sigma_{\rho_i}, \sigma_{\rho_j}\}.$$
\end{condition}

    The following lemma is immediate from from Lemmata
    \ref{lem:regularity-conditions-hold} and
    \ref{lem:easy-intra-gaussian}.

\begin{lemma}\label{lem:easy-dr-conditions}
    Let $Y$ and $Z$ be datasets sampled from
    $\cD \in \cS(d,k,\kappa,(\xi+16\sqrt{\kappa})\sqrt{\ell})$
    (with $Y^L$ and $Z^L$ being their respective
    labelled datasets)
    as defined within the algorithm, such that
    $d \geq \ell \geq \max\{512\ln(nk/\beta),k\}$
    where $\xi,\kappa \in \Theta(1)$
    and $\xi>1$ is a universal contant. If
    $$n \geq O\left(\frac{k^2}{\alpha^2}\ln(k/\beta)\right),$$
    then with probability at least $1-4\beta$,
    $Y^L$ and $Z^L$ (alternatively $Y$ and $Z$)
    satisfy Conditions \ref{cond:num-points}
    and \ref{cond:easy-intra-gaussian}.
\end{lemma}

\subsubsection{PCA}

The following result stated in \cite{DworkTTZ14},
though used in a setting where $\Delta_{\eps,\delta}$
was fixed, holds for any value of
$\Delta_{\eps,\delta}$.

\begin{lemma}[Theorem 9 of~\cite{DworkTTZ14}]\label{lem:ag-spectral}
    Let $\wh{\Pi}_{\ell}$ be the top $\ell$
    principal subspace obtained using $Y^TY + E$
    in Algorithm \ref{alg:pegm}. Suppose
    $\Pi_{\ell}$ is the top $\ell$ subspace
    obtained from $Y^TY$. Then with probability
    at least $1-\beta$,
    $$\llnorm{Y^TY - (Y\wh{\Pi}_{\ell})^T(Y\wh{\Pi}_{\ell})}
        \leq \llnorm{Y^TY - (Y\Pi_{\ell})^T(Y\Pi_{\ell})}
        + O\left(\Delta_{\eps,\delta}\sqrt{d}\right).$$
\end{lemma}

The main result here is that the PCA step
shrinks the Gaussians down in a way that
their means, after being projected upon
the privately computed subspace, are close
to their original locations.

\begin{lemma}\label{lem:easy-pca}
    Let $Y$ be the dataset, and $\wh{V}_{\ell}$
    be the subspace as defined
    in Algorithm \ref{alg:pegm}. Suppose $\mu_1,\dots,\mu_k$
    are the means of the Gaussians, and $\mu'_1,\dots,\mu'_k$
    are their respective projections on to $\wh{V}_{\ell}$.
    If
    $$n \geq O\left(\frac{d^{\frac{3}{2}}k^3\sqrt{\ln(1/\delta)}}{\eps}
        + k\ln(k/\beta)\right),$$
    and $Y$ satisfies Condition \ref{cond:num-points},
    then with probability at least $1-4\beta$,
    \begin{enumerate}
        \item for every $i$, we have
            $\llnorm{\mu_i - \mu'_i} \leq 8\sqrt{\ell}\maxsigma$,
            where $\maxsigma = \max\limits_{i}\{\sigma_i\}$,
            and
        \item for every $i \neq j$,
            $\llnorm{\mu'_i - \mu'_j} \geq
                \llnorm{\mu_i - \mu_j} - 16\sqrt{\ell}\maxsigma$.
    \end{enumerate}
\end{lemma}
\begin{proof}
    Let $\wt{\mu}_i$ be the empirical
    mean of component $i$ using the respective
    points in $Y$, and let
    $\dot{\mu}_i = \wh{\Pi}_{\ell}\wt{\mu}_i$.
    Using Lemmata \ref{lem:PCA_AM_style},
    \ref{lem:data-spectral}, and \ref{lem:ag-spectral},
    we know that with probability at least
    $1-\beta$, for all $1 \leq i \leq k$,
    \begin{align*}
        \llnorm{\dot{\mu}_i - \wt{\mu}_i} &\leq
                4\sqrt{2\sum\limits_{j=1}^{k}{\sigma^2_i}} +
                O\left(\sqrt{\frac{\Delta_{\eps,\delta}\sqrt{d}k}{n}}\right)\\
            &\leq 4\sqrt{2k}\maxsigma +
                O\left(\sqrt{\frac{d^{\frac{3}{2}}k^3
                \minsigma^2\sqrt{\ln(1/\delta)}}{\eps n}}\right).
                \tag{$\kappa \in \Theta(1)$}
    \end{align*}
    Because $n \geq \tfrac{d^{\frac{3}{2}}k^3\sqrt{\ln(1/\delta)}}{\eps}$
    and $\tfrac{\maxsigma^2}{\minsigma^2} \leq
    \kappa \in \Theta(1)$, we have
    \begin{align*}
        O\left(\sqrt{\frac{d^{\frac{3}{2}}k^3
                \minsigma^2\sqrt{\ln(1/\delta)}}{\eps n}}\right)
                &\leq O\left(\maxsigma\right)\\
        \implies \llnorm{\dot{\mu}_i - \wt{\mu}_i}
            &\leq 6 \sqrt{k} \maxsigma\\
            &\leq 6 \sqrt{\ell} \maxsigma.
    \end{align*}
    Since $n \geq O\left(\tfrac{dk}{\ell} +
    \tfrac{k\ln(k/\beta)}{d}
    + \tfrac{k\ln(k/\beta)}{\ell}\right)$, using
    Lemma \ref{lem:spectralnormofgaussians},
    we have that with probability at least
    $1-2\beta$, for all $1 \leq i \leq k$,
    $$\llnorm{\wt{\mu}_i - \mu_i} \leq \sqrt{\ell}\sigma_i
        \leq \sqrt{\ell}\maxsigma~~~ \text{and}~~~
        \llnorm{\dot{\mu}_i - \mu'_i} \leq \sqrt{\ell}\sigma_i
        \leq \sqrt{\ell}\maxsigma.$$
    Finally, we get the required results
    by using triangle inequality.
\end{proof}

\subsubsection{Clustering}

After the PCA step, individual Gaussian
components shrink (i.e., the radius decreases from
$O(\sqrt{d}\sigma)$ to $O(\sqrt{\ell}\sigma)$), but the means do not shift a
lot. Given the large initial separation,
we can find individual components
using the private location algorithm, and
learn them separately using a private
learner. In this section, we show that
our algorithm is able to achieve the first
goal.

First, we prove that our data is likely to satisfy some
of the conditions have have already defined.

\begin{lemma}
  \label{lem:all-dr-conditions}
    Let $Y,Z,Z_{\ell}$ be datasets as defined
    within the algorithm (with $Z^L_{\ell}$
    being the corresponding labeled dataset
    of $Z_{\ell}$), where $Y,Z$ are sampled
    from $\cD \in \cS(d,k,\kappa,(\xi+16\sqrt{\kappa})\sqrt{\ell})$,
    such that $d \geq \ell \geq \max\{512\ln(nk/\beta),k\}$,
    and $\xi,\kappa \in \Theta(1)$ and $\xi>1$
    is a universal constant. If
    $$n \geq O\left(\frac{d^{\frac{3}{2}}k^3\sqrt{\ln(1/\delta)}}{\eps}
        + \frac{k^2}{\alpha^2}\ln(k/\beta)\right),$$
    and $Y$ and $Z$ satisfy Condition
    \ref{cond:num-points}, then with probability
    at least $1-7\beta$, $Z^L_{\ell}$ (alternatively
    $Z_{\ell}$) satisfies Conditions \ref{cond:num-points},
    \ref{cond:easy-intra-gaussian}, and
    \ref{cond:easy-inter-gaussian}.
\end{lemma}
\begin{proof}
    Let $\mu_1,\dots,\mu_k$ be means of the
    Gaussians in $\cD$, and $\mu'_1,\dots,\mu'_k$
    be their respective projections onto
    $\wh{V}_{\ell}$, and let $\sigma_1^2,\dots,\sigma_k^2$
    be their respective variances. Because $Y$
    satisfies Condition \ref{cond:num-points},
    we know from Lemma \ref{lem:easy-pca} that
    with probability at least $1-4\beta$,
    for each $1 \leq i \neq j \leq k$
    \begin{align*}
        \llnorm{\mu'_i-\mu'_j} &\geq
                \llnorm{\mu_i-\mu_j} - 16\sqrt{\ell}\maxsigma\\
            &\geq (\xi + 16\sqrt{\kappa})\sqrt{\ell}\max\{\sigma_i,\sigma_j\}
                - 16\sqrt{\ell}\maxsigma\\
            &\geq \xi\sqrt{\ell}\max\{\sigma_i,\sigma_j\}.
    \end{align*}
    Therefore, points in $Z_{\ell}$ are
    essentially points from some
    $\cD' \in \cS(\ell,k,\kappa,\xi\sqrt{\ell})$.
    $Z_{\ell}$ clearly satisfies Condition
    \ref{cond:num-points}. Using Lemmata
    \ref{lem:easy-intra-gaussian} and
    \ref{lem:easy-inter-gaussian}, we have
    that $Z_{\ell}$ satisfies the other
    two conditions as well. Using the union
    bound over these three events, we get
    the required result.
\end{proof}

The following theorem guarantees the existence
of an algorithm that finds approximately smallest
balls containing almost the specified number of
points. 
This is based off of the 1-cluster algorithm of
Nissim and Stemmer~\cite{NissimS18}.
We provide its proof, and state such an
algorithm in Section \ref{sec:missingproofs-pgloc}.
%While the leading factor of the sample complexity
%may seem arcane, we note that it can be replaced
%by $\left(\frac{\sqrt{\ell} k}{\gamma \eps}\right)^{1 + \tau}$,
%for any constant $\tau > 0$.

\begin{thm}[Private Location for GMMs, Extension of~\cite{NissimS18}]\label{thm:pgloc}
    There is an $(\eps,\delta)$-differentially
    private algorithm
    $\PGLoc(X,t; \eps,\delta,R,\minsigma,\maxsigma)$
    with the following guarantee.
    Let $X = (X_1,\dots,X_n) \in \R^{n \times d}$
    be a set of $n$ points drawn from a mixture
    of Gaussians $\cD \in \cG(\ell,k,R,\minsigma,\maxsigma,\mw,s)$.
    Let $S \subseteq X$ such that $\abs{S} \geq t$, and
    let $0 < a < 1$ be any small absolute constant (say, one can take $a = 0.1$).
    If $t = \gamma n$, where $0 < \gamma \leq 1$, and
    \begin{align*}
      n &\geq \left(\frac{\sqrt{\ell}}
            {\gamma\eps}\right)^{\frac{1}{1-a}}
            \cdot 9^{\log^*\left(\sqrt{\ell}\left(\frac{R\maxsigma}
            {\minsigma}\right)^{\ell}\right)}
            \cdot \polylog\left(\ell,
            \frac{1}{\eps},\frac{1}{\delta}
            \frac{1}{\beta},\frac{1}{\gamma}\right)
            + O\left(\frac{\ell + \log(k/\beta)}{\mw}\right),
    \end{align*}
    then for some absolute constant $c > 4$ that
    depends on $a$, with
    probability at least $1-\beta$, the algorithm
    outputs $(r,\vec{c})$ such that the following hold:
    \begin{enumerate}
        \item $B_{r}(\vec{c})$ contains at least
            $\tfrac{t}{2}$ points in $S$, that is,
            $\abs{B_r(\vec{c}) \cap S} \geq
            \tfrac{t}{2}$.
        \item If $r_{\opt}$ is the radius of
            the smallest ball containing at least
            $t$ points in $S$, then $r \leq c
            \left(r_{\opt}+\frac{1}{4}\sqrt{\ell}\minsigma\right).$
    \end{enumerate}
\end{thm}
%\begin{thm}[Private Location for GMMs, Extension of~\cite{NissimS18}]\label{thm:pgloc}
%    There is an $(\eps,\delta)$-differentially
%    private algorithm
%    $\PGLoc(X,t; \eps,\delta,R,\minsigma,\maxsigma)$
%    with the following guarantee.
%    Let $X = (X_1,\dots,X_n) \in \R^{n \times d}$
%    be a set of $n$ points drawn from a mixture
%    of Gaussians $\cD \in \cG(\ell,k,R,\minsigma,\maxsigma,\mw,s)$.
%    If $t = \gamma n$, where $0 < \gamma \leq 1$, and
%    \begin{align*}
%        n &\geq \frac{\ell^{\frac{5}{9}}k^{\frac{10}{9}}}
%            {\gamma^{\frac{10}{9}}\eps^{\frac{10}{9}}}
%            \cdot 9^{\log^*\left(\sqrt{\ell}\left(\frac{R\maxsigma}
%            {\minsigma}\right)^{\ell}\right)}
%            \cdot \polylog\left(\ell,
%            \frac{1}{\eps},\frac{1}{\delta},
%            \frac{1}{\beta},\frac{1}{\gamma}\right)
%    \end{align*}
%    then for some absolute constant $c > 4$, with
%    probability at least $1-\beta$, the algorithm
%    outputs $(r,\vec{c})$ such that the following hold:
%    \begin{enumerate}
%        \item $B_{r}(\vec{c})$ contains at least
%            $\tfrac{t}{2}$ points in $X$, that is,
%            $\abs{B_r(\vec{c}) \cap X} \geq
%            \tfrac{t}{2}$.
%        \item If $r_{\opt}$ is the radius of
%            the smallest ball containing at least
%            $t$ points in $X$, then $r \leq c
%            \left(r_{\opt}+\frac{1}{4}\sqrt{\ell}\minsigma\right).$
%    \end{enumerate}
%\end{thm}

Since the constant $a$ can be arbitrarily small, for simplicity,
we fix it to $0.1$ for the remainder of this section.

The first lemma we prove says that individual
components are located correctly in
the lower dimensional subspace, which
is to say that we find $k$ disjoint balls,
such that each ball completely contains
exactly one component. 

We first define the
following events.
\begin{enumerate}
    \item $E_{Y,Z}$: $Y$ and $Z$ satisfy Condition
        \ref{cond:num-points}
    \item $E_Z$: $Z$ satisfies Conditions \ref{cond:num-points}
        and \ref{cond:easy-intra-gaussian}
    \item $E_{Z_{\ell}}$: $Z_{\ell}$
        satisfies Conditions \ref{cond:num-points},
        \ref{cond:easy-intra-gaussian}, and
        \ref{cond:easy-inter-gaussian}.
\end{enumerate}

\begin{lemma}\label{lem:easy-ld-clustering}
    Let $Z_{\ell}$ be the dataset as defined
    in the algorithm, and let $Z^L_{\ell}$ be
    its corresponding labelled dataset. If
    $$n \geq O\left(\frac{d^{\frac{3}{2}}k^3\sqrt{\ln(1/\delta)}}{\eps}
        + \frac{k^2}{\alpha^2}\ln(k/\beta)
        + \frac{\ell^{\frac{5}{9}}k^{\frac{5}{3}}}
            {\eps^{\frac{10}{9}}}
            \cdot \polylog\left(\ell,k,
            \frac{1}{\eps},\frac{1}{\delta},
            \frac{1}{\beta}\right)\right),$$
    and events $E_{Y,Z}$ and $E_{Z_{\ell}}$
    happen, then with probability
    at least $1-5\beta$, at the end of the
    first loop,
    \begin{enumerate}
        \item $i = k+1$, that is, the algorithm
            has run for exactly $k$ iterations;
        \item for all $1 \leq i \leq k$, if
            $u,v \in \ball{c_i}{r_i}$, and
            $(u,\rho_u),(v,\rho_v) \in Z^L_{\ell}$,
            then $\rho_u = \rho_v$;
        \item for all $1 \leq i \neq j \leq k$, if
            $u \in \ball{c_i}{r_i}$, $v \in \ball{c_j}{r_j}$,
            and $(u,\rho_u),(v,\rho_v) \in Z^L_{\ell}$,
            then $\rho_u \neq \rho_v$;
        \item for all $1 \leq i \neq j \leq k$,
            $\ball{c_i}{r_i} \cap \ball{c_j}{r_j} = \emptyset$;
        \item $S_1 \cup \dots \cup S_k = Z_{\ell}$;
        \item for all $1 \leq i \leq k$, if
            $u \in \ball{c_i}{r_i}$, and
            $(u,\rho_u) \in Z^L_{\ell}$, then
            $r_i \in \Theta(\sqrt{k}\sigma_{\rho_u})$.
    \end{enumerate}
\end{lemma}
\begin{proof}
    Throughout the proof, we will omit the
    conditioning on events $E_{Y,Z}$ and
    $E_{Z_{\ell}}$ for
    brevity. To prove this lemma, we first
    prove a claim that at the end of iteration
    $i$, the balls we have found so far are
    disjoint, and that each ball completely
    contains points from exactly one Gaussian,
    no two balls contain the same component,
    and the radius of each ball is not too
    large compared to the radius of the optimal
    ball that contains the component within
    it.
    \begin{claim}
        Let $E_i$ be the event that at the end
        of iteration $i$,
        \begin{enumerate}
            \item the number of components found so far is $i$;
            \item for all $1 \leq a \leq i$, if
                $u,v \in \ball{c_a}{r_a}$, and
                $(u,\rho_u),(v,\rho_v) \in Z^L_{\ell}$,
                then $\rho_u = \rho_v$;
            \item for all $1 \leq a \neq b \leq i$, if
                $u \in \ball{c_a}{r_a}$, $v \in \ball{c_b}{r_b}$,
                and $(u,\rho_u),(v,\rho_v) \in Z^L_{\ell}$,
                then $\rho_u \neq \rho_v$;
            \item for all $1 \leq a \neq b \leq i$,
                $\ball{c_a}{r_a} \cap \ball{c_b}{r_b} = \emptyset$;
              \item if $B_i = \ball{c_1}{r_1} \cup \dots \cup \ball{c_i}{r_i}$,
                then for all $u \in B_i \cap Z_{\ell}$ and
                $v \in Z_{\ell} \setminus B_i$, such that
                $(u,\rho_u),(v,\rho_v) \in Z^L_{\ell}$,
                it holds that $\rho_u \neq \rho_v$.
            \item for all $1 \leq a \leq i$, if
                $u \in \ball{c_a}{r_a}$, and
                $(u,\rho_u) \in Z^L_{\ell}$, then
                $\tfrac{\sqrt{\ell}\sigma_{\rho_u}}{2} \leq
                r_a \leq 48c\sqrt{3\ell}\sigma_{\rho_u}$.
        \end{enumerate}
        If
        $$n \geq O\left(\frac{d^{\frac{3}{2}}k^3\sqrt{\ln(1/\delta)}}{\eps}
            + \frac{k^2}{\alpha^2}\ln(k/\beta)
            + \frac{\ell^{\frac{5}{9}}k^{\frac{5}{3}}}
                {\eps^{\frac{10}{9}}}
                \cdot \polylog\left(\ell,k,
                \frac{1}{\eps},\frac{1}{\delta},
                \frac{1}{\beta}\right)\right),$$
        then
        $$\pr{}{E_i} \geq 1 - \frac{i\beta}{k}$$
    \end{claim}
    \begin{proof}
        We prove this by induction on $i$. Suppose
        the claim holds for $i-1$. Then it is
        sufficient to show the following.
        $$\pr{}{E_i|E_{i-1}} \geq 1 - \frac{\beta}{k}$$
        Note that the points are now from
        $\ell$ dimensional Gaussians owing to
        being projected upon an $\ell$ dimensional
        subspace. Conditioning on $E_{i-1}$
        entails that at the beginning of
        iteration $i$, there is no label for
        points in $S_{\ell}$ that occurs for points
        in $B_{i-1} \cap Z_{\ell}$. Therefore,
        the number of points having the remaining
        labels is the same as in the beginning.

        For any two labels $\rho_u,\rho_v$,
        suppose $\mu''_{\rho_u}, \mu''_{\rho_v}$
        are the respective means (in the space
        after projection and rotation),
        $\sigma^2_{\rho_u},\sigma^2_{\rho_v}$
        are the respective variances,
        $\mu_{\rho_u},\mu_{\rho_v}$ are the original
        means, and $\mu'_{\rho_u},\mu'_{\rho_v}$
        are the projected, unrotated means.
        Then by conditioning on $E_{Y,Z}$ and
        $E_{Z_{\ell}}$, and using Lemma
        \ref{lem:easy-pca}, we have
        \begin{align*}
            \llnorm{\mu_{\rho_u} - \mu_{\rho_v}} &\geq
                (400c+16\kappa)\sqrt{\ell}\max\{\sigma_{\rho_u},\sigma_{\rho_v}\}\\
            \implies \llnorm{\mu'_{\rho_u} - \mu'_{\rho_v}} &\geq
                400c\sqrt{\ell}\max\{\sigma_{\rho_u},\sigma_{\rho_v}\}\\
            \implies \llnorm{\mu''_{\rho_u} - \mu''_{\rho_v}} &\geq
                400c\sqrt{\ell}\max\{\sigma_{\rho_u},\sigma_{\rho_v}\},
        \end{align*}
        where the final inequality holds because the 
        $\ell_2$ norm is rotationally invariant.
        Therefore by conditioning on event $E_{Z_{\ell}}$
        again, for any two points with labels $\rho_1$ and $\rho_2$,
        the distance between the two is strictly
        greater than
        $$200c\sqrt{\ell}\max\{\sigma_{\rho_1},\sigma_{\rho_2}\}.$$

        Now, conditioning on $E_{Z_{\ell}}$,
        we know that the radii of individual clusters
        are bounded. Because the components in
        this subspace are well-separated, the
        smallest ball containing at least $\tfrac{n}{2k}$
        points cannot have points from two different
        components, that is, the radius of the smallest
        ball containing at least $\tfrac{n}{2k}$
        points has to be the radius of the smallest
        ball that contains at least $\tfrac{n}{2k}$
        points from a single component.
        Let that component have label $\rho_1$
        (without loss of generality), and let its
        radius be $r_{opt}$. By Theorem
        \ref{thm:pgloc} (by setting $S=S_{\ell}$),
        with probability at least
        $1-\tfrac{\beta}{k}$, we obtain a ball of
        radius at most
        $$c\left(r_{opt} + \frac{1}{4}\sqrt{\ell}\minsigma\right)
            \leq 2cr_{opt},$$
        that contains at least $\tfrac{n}{4k}$
        points from the component. Since $Z_{\ell}$
        satisifes Condition~\ref{cond:easy-intra-gaussian},
        multiplying this radius by $4\sqrt{3}$ (to get
        $r_i$) is enough to cover all points of that
        component, hence, we have that
        \begin{align}
            \frac{\sqrt{\ell}\sigma_{\rho_1}}{2} \leq r_{opt} \leq
                r_i \leq 8\sqrt{3}cr_{opt} \leq 24c\sqrt{\ell}\sigma_{\rho_1}.
                \label{eq:clustering-radius}
        \end{align}

        We want to show that $\ball{c_i}{r_i}$
        contains points from exactly one component
        among all points in $Z_{\ell}$. There are two
        cases to consider. First, where the ball
        contains points that have label $\rho_1$,
        that is, the ball returned contains points
        only from the cluster with the smallest radius.
        Second, where the ball returned contains
        points from some other component that has a
        different label $\rho_2$ (without loss of generality),
        that is, the ball returned does not contain
        any points from the smallest cluster, but has
        points from a different cluster. As we will
        show later, this can only happen if the radius
        of this other cluster is ``similar'' to that
        of the smallest cluster.

        In the first case, $\ball{c_i}{r_i}$ completely
        contains all points with label $\rho_1$. But
        because this radius is less than its distance
        from every other point in $Z_{\ell}$ with a
        different label, it does not contain any points
        from any other labels in $Z_{\ell}$. This proves
        (1), (2), (3), (5), and (6) for this case.
        Now, consider any index $a \leq i-1$. Then
        $\ball{c_a}{r_a}$ contains exactly one component,
        which has label (without loss of generality)
        $\rho_3$. Let $u \in \ball{c_a}{r_a}$ and
        $v \in \ball{c_i}{r_i}$. Then we have the
        following.
        \begin{align*}
            \llnorm{c_i - c_a} &= \llnorm{(u-v) - (v-c_i) - (c_a-u)}\\
                &\geq \llnorm{u-v} - \llnorm{v-c_i} - \llnorm{c_a-u}\\
                &> 200c\sqrt{\ell}\max\{\sigma_{\rho_i},\sigma_{\rho_a}\} -
                    48c\sqrt{3\ell}\sigma_{\rho_a} - 24c\sqrt{\ell}\sigma_{\rho_i}\\
                &> r_i + r_a
        \end{align*}
        This shows that $\ball{c_a}{r_a}$ and
        $\ball{c_i}{r_i}$ do not intersect. This
        proves (4) for this case.

        In the second case, let $r_{\rho_2}$ be the
        radius of the smallest ball containing points
        only from the component with label $\rho_2$. Since
        $r_{opt}$ is the radius of the component with label
        $\rho_1$, and $Z_{\ell}$ satisfies
        Condition~\ref{cond:easy-intra-gaussian},
        it must be the case that
        $$\frac{\sqrt{\ell}\sigma_{\rho_1}}{2} \leq r_{opt}
            \leq r_{\rho_2} \leq \sqrt{3\ell}\sigma_{\rho_2}
            \implies \sqrt{3\ell}\sigma_{\rho_2} \geq
            \frac{\sqrt{\ell}\sigma_{\rho_1}}{2}
            \implies 2\sqrt{3}\sigma_{\rho_2} \geq \sigma_{\rho_1},$$
        otherwise the smallest component would
        have been the one with label $\rho_2$.
        Combined with Inequality~\ref{eq:clustering-radius},
        this implies that $r_i \leq 48c\sqrt{3\ell}\sigma_{\rho_2}$,
        which proves (6) for this case. The arguments
        for other parts for this case are identical
        to those for the first case.

        Since the only randomness comes
        from running $\PGLoc$, using Theorem
        \ref{thm:pgloc}, we get,
        $$\pr{}{E_i|E_{i-1}} \geq 1-\frac{\beta}{k}.$$
        Therefore, by the inductive hypothesis,
        $$\pr{}{E_i} \geq 1-\frac{i\beta}{k}.$$
        The argument for the base case is the
        same, so we omit that for brevity.
    \end{proof}
    We complete the proof by using the above
    claim by setting $i=k$.
\end{proof}

The next lemma states that given that the
algorithm has isolated individual components
correctly in the lower dimensional subspace,
it can correctly classify the corresponding
points in the original space. To elaborate,
this means that for each component, the algorithm
finds a small ball that contains the component,
and is able to correctly label all points in
the dataset.

\begin{lemma}\label{lem:easy-hd-clustering}
    Let $Z$, $Z_{\ell}$, and $\wh{S}_i,\dots,\wh{S}_k$
    be datasets as defined in the the algorithm,
    and let $Z^L$ be the corresponding labelled
    dataset of $Z$. If
    $$n \geq O\left(\frac{d^{\frac{3}{2}}k^3\sqrt{\ln(1/\delta)}}{\eps}
        + \frac{k^2}{\alpha^2}\ln(k/\beta)
        + \frac{\ell^{\frac{5}{9}}k^{\frac{5}{3}}}
            {\eps^{\frac{10}{9}}}
            \cdot \polylog\left(\ell,k,
            \frac{1}{\eps},\frac{1}{\delta},
            \frac{1}{\beta}\right)\right),$$
    and events $E_{Y,Z}$, $E_Z$, and $E_{Z_{\ell}}$
    happen then with probability at least $1-11\beta$,
    \begin{enumerate}
        \item for all $1 \leq i \leq k$, if
            $u,v \in \wh{S}_i$, and
            $(u,\rho_u),(v,\rho_v) \in Z^L$,
            then $\rho_u = \rho_v$;
        \item for all $1 \leq i \neq j \leq k$, if
            $u \in \wh{S}_i$, $v \in \wh{S}_j$,
            and $(u,\rho_u),(v,\rho_v) \in Z^L$,
            then $\rho_u \neq \rho_v$;
        \item $\wh{S}_1 \cup \dots \cup \wh{S}_k = Z$;
        \item for all $1 \leq i \leq k$, if
            $u \in \wh{S}_i$, and
            $(u,\rho_u) \in Z^L$, then
            $r_i \in \Theta(\sqrt{k}\sigma_{\rho_u})$.
    \end{enumerate}
\end{lemma}
\begin{proof}
    Again for brevity, we will omit conditioning
    on events $E_{Y,Z}$, $E_Z$, and $E_{Z_{\ell}}$.
    We will prove the following
    claim that says that by the end of iteration $i$ of
    the second loop, each set formed so far completely contains
    exactly one component, and the ball that contains
    it is small. But before that, note that: (1) from
    Lemma~\ref{lem:easy-pca}, with
    probability at least $1-4\beta$, for all components,
    the distance between the original mean and its
    respective projection onto $\wh{V}_{\ell}$
    is small; and (2) from Lemma~\ref{lem:easy-ld-clustering},
    with probability at least
    $1 - 5\beta$, the balls found in the lower dimensional
    subspace in the previous step isolate individual
    components in balls, whose radii are small. We
    implicitly condition the following claim on them.
    \begin{claim}
        Let $E_i$ be the event that at the end of
        iteration $i$,
        \begin{enumerate}
            \item for all $1 \leq a \leq i$, if
                $u,v \in \wh{S}_i$, and
                $(u,\rho_u),(v,\rho_v) \in Z^L$,
                then $\rho_u = \rho_v$;
            \item for all $1 \leq a \neq b \leq i$, if
                $u \in \wh{S}_i$, $v \in \wh{S}_j$,
                and $(u,\rho_u),(v,\rho_v) \in Z^L$,
                then $\rho_u \neq \rho_v$;
            \item if $T_i = \wh{S}_1 \cup \dots \cup \wh{S}_i$,
                then for all $u \in T_i \cap Z$ and
                $v \in Z \setminus T_i$, such that
                $(u,\rho_u),(v,\rho_v) \in Z^L$,
                it holds that $\rho_u \neq \rho_v$.
            \item for all $1 \leq a \leq i$, if
                $u \in \wh{S}_a$, and
                $(u,\rho_u) \in Z^L_{\ell}$, then
                $r_a \in \Theta(\sqrt{d}\sigma_{\rho_u})$.
%                $\tfrac{\sqrt{\ell}\sigma_{\rho_u}}{2} \leq
%                r_a \leq 48c\sqrt{3\ell}\sigma_{\rho_u}$.
        \end{enumerate}
        If
        $$n \geq O\left(\frac{d^{\frac{3}{2}}k^3\sqrt{\ln(1/\delta)}}{\eps}
            + \frac{k^2}{\alpha^2}\ln(k/\beta)
            + \frac{\ell^{\frac{5}{9}}k^{\frac{5}{3}}}
                {\eps^{\frac{10}{9}}}
                \cdot \polylog\left(\ell,k,
                \frac{1}{\eps},\frac{1}{\delta},
                \frac{1}{\beta}\right)\right),$$
        then
        $$\pr{}{E_i} \geq 1-\frac{2i\beta}{k}.$$
    \end{claim}
    \begin{proof}
        We prove this by induction on $i$. Suppose
        the claim holds for $i-1$. Then it is
        sufficient to show the following.
        $$\pr{}{E_i|E_{i-1}} \geq 1-\frac{2\beta}{k}$$
        From Lemma \ref{lem:easy-ld-clustering}, we
        know that $\ball{c_i}{r_i}$ completely
        contains exactly one component from $Z_{\ell}$.
        This implies that the empirical mean of
        those points ($\dot{\mu}$) also lies within
        $Z_{\ell}$. Suppose the mean of their
        distribution in $\wh{V}_{\ell}$ is $\mu'$
        and its variance is $\sigma^2$. We know from
        Lemma \ref{lem:spectralnormofgaussians}, and
        because
        $n \geq O\left(\tfrac{\ln(k/\beta)}{\ell}\right)$,
        that with probability at least $1-\beta/k$,
        $$\llnorm{\dot{\mu}-\mu'} \leq \sqrt{\ell}\sigma.$$
        Let $\mu$ be the mean of the Gaussian in
        the original subspace. We know from Lemma~\ref{lem:easy-pca}
        that,
        $$\llnorm{\mu-\mu'} \leq 8\sqrt{\ell}\maxsigma.$$
        Let the empirical mean of the corresponding
        points in $Z$ be $\wt{\mu}$. Again, from
        Lemma \ref{lem:spectralnormofgaussians}, and
        because
        $n \geq O\left(\tfrac{\ln(k/\beta)}{\ell} + \tfrac{d}{\ell}\right)$,
        we know that
        with probability at least $1-\beta/k$,
        $$\llnorm{\mu-\wt{\mu}} \leq \sqrt{\ell}\sigma.$$
        Therefore, by triangle inequality,
        $\ball{\wh{c}_i}{r_i + 10\sqrt{\ell\kappa}\minsigma}$
        contains $\wt{\mu}$. 
        Now, from the proof of
        Lemma \ref{lem:easy-intra-gaussian}, we know
        that all points of the Gaussian in $Z$ will be
        at most $\sqrt{3d}\sigma$ away from $\wt{\mu}$.
        But since $\ball{c_i}{r_i}$ contains all
        points from the Gaussian in $Z_{\ell}$,
        we know from conditioning on $E_{Z_{\ell}}$
        that
        $r_i \geq \tfrac{\sqrt{\ell}\sigma}{2}$,
        which means that the each of the corresponding
        points in $Z$ is at most $2r_i\sqrt{\tfrac{3d}{\ell}}$
        away from $\wt{\mu}$.
        Therefore, by triangle inequality, the distance
        between $\wh{c}_i$ and any of these
        points in $Z$ can be at most
        $$r_i + 10\sqrt{\ell\kappa}\minsigma + 2r_i\sqrt{\tfrac{3d}{\ell}}.$$
        Because this is exactly how $\wh{r}_i$ is
        defined, $\ball{\wh{c}_i}{\wh{r}_i}$
        completely contains all points from the
        component in $Z$. Since $S_i$ contains all
        the corresponding points from $Z_{\ell}$,
        it must be the case that $\wh{S}_i$ contains
        all the points from the component.
        
        Finally, we prove that the radius of the ball
        enclosing the component in $Z$ is small enough.
        \begin{align*}
            &\wh{r}_i = r_i + 10\sqrt{\ell\kappa}\minsigma
                    + 2r_i\sqrt{\frac{3d}{\ell}}\\
            \implies r_i + 10\sqrt{\ell\kappa}\minsigma
                    + \Omega(\sqrt{d}\sigma) \leq
                    &\wh{r}_i
                    \leq r_i + 10\sqrt{\ell\kappa}\minsigma
                    + O(\sqrt{d}\sigma)
                    \tag{$r_i \in \Theta(\sqrt{\ell}\sigma)$}\\
            \implies &\wh{r}_i \in \Theta(\sqrt{d}\sigma)
                \tag{$\kappa \in \Theta(1)$}.
        \end{align*}

        Therefore,
        $$\pr{}{E_i|E_{i-1}} \geq 1-\frac{2\beta}{k}.$$
        Hence, by the inductive hypothesis,
        $$\pr{}{E_i} \geq 1-\frac{2i\beta}{k}.$$
        The argument for the base case is identical,
        so we skip it for brevity.
    \end{proof}
    We complete the proof by setting $i=k$
    in the above claim, and using the union
    bound.
\end{proof}

\subsubsection{Estimation}

Once we have identified individual components
in the dataset, we can invoke a differentially
private learning algorithm on each one of them
separately. The next theorem establishes the
existence of one such private learner that is
tailored specifically for spherical Gaussians,
and is accurate even when the number of points
in the dataset constitutes sensitive information.
We provide its proof, and state such an algorithm
in Section \ref{sec:missingproofs-psge}.

\begin{theorem}\label{thm:easy-learn-gaussian}
    There exists an $(\eps,\delta)$-differentially
    private algorithm $\PSGE(X; \vec{c},r,\alpha_{\mu},\alpha_{\sigma},
    \beta,\eps,\delta)$ with the following guarantee.
    If $B_{r}(\vec{c}) \subseteq \R^{\ell}$ is a ball,
    $X_1,\dots,X_n \sim \cN(\mu, \sigma^2 \id_{\ell \times \ell})$,
    and
    $n \geq \frac{6\ln(5/\beta)}{\eps} +
        n_{\mu} + n_{\sigma},$
    where
    \begin{align*}
        n_{\mu} &= O\left(
            \frac{\ell}{\alpha^2_{\mu}} +
            \frac{\ln(1/\beta)}{\alpha^2_{\mu}} +
            \frac{r\ln(1/\beta)}
            {\alpha_{\mu}\eps\sigma} +
            \frac{r\sqrt{\ell
            \ln(1/\delta)}}{\alpha_{\mu}\eps\sigma}
            \right),\\
        n_{\sigma} &= O\left(
            \frac{\ln(1/\beta)}{\alpha^2_{\sigma}\ell} +
            \frac{\ln(1/\beta)}
            {\alpha_{\sigma}\eps} +
            \frac{r^2\ln(1/\beta)}
            {\alpha_{\sigma}\eps\sigma^2 \ell}\right),
    \end{align*}
    then with probability at least $1-\beta$, the
    algorithm returns $\wh{\mu},\wh{\sigma}$
    such that if $X$ is contained in $B_{r}(\vec{c})$
    (that is, $X_i \in B_{r}(\vec{c})$ for every $i$)
    and $\ell \geq 8\ln(10/\beta)$,
    then
    $$\| \mu - \wh{\mu} \|_2 \leq \alpha_{\mu} \sigma \qquad \textrm{and} \qquad
        (1-\alpha_\sigma) \leq \frac{\wh\sigma^2}{\sigma^2}
        \leq(1+\alpha_\sigma).$$
\end{theorem}

With the above theorem in our tookit, we can
finally show that the each estimated Gaussian
is close to its respective actual Gaussian to
within $\alpha$ in TV-distance.

\begin{lemma}\label{lem:easy-individual-gaussian}
    Suppose $\mu_1,\dots,\mu_k$ are the means,
    and $\sigma^2_1,\dots,\sigma^2_k$ are the
    variances of the Gaussians of the target
    distribution $\cD \in \cS(d,k,\kappa,(\xi + 16\sqrt{\kappa})\sqrt{\ell})$,
    where $\xi,\kappa \in \Theta(1)$ and
    $\xi$ is a universal constant, and
    $\ell = \max\{512\ln(nk/\beta),k\}$.
    Let $\wh{\mu}_1,\dots,\wh{\mu}_k$, and
    $\wh{\sigma}^2_1,\dots,\wh{\sigma}^2_k$ be
    their respective estimations produced by
    the algorithm. If
    $$n \geq n_{\textsc{CLUSTERING}} + O\left(\frac{dk}{\alpha^2}
        + \frac{dk\sqrt{\ln(1/\delta)}}{\alpha\eps}
        + \frac{\sqrt{d}k\ln(k/\beta)}{\alpha\eps}
        + \frac{k\ln(k/\beta)}{\alpha^2}\right),$$
    where
    $$n_{\textsc{CLUSTERING}}
        \geq O\left(\frac{d^{\frac{3}{2}}k^3\sqrt{\ln(1/\delta)}}{\eps}
        + \frac{k^2}{\alpha^2}\ln(k/\beta)
        + \frac{\ell^{\frac{5}{9}}k^{\frac{5}{3}}}
            {\eps^{\frac{10}{9}}}
            \cdot \polylog\left(\ell,k,
            \frac{1}{\eps},\frac{1}{\delta},
            \frac{1}{\beta}\right)\right),$$
    and events $E_{Y,Z}$ and $E_{Z}$ happen,
    then with probability at least $1-12\beta$,
    $$\forall~ 1 \leq i \leq k,~~~
        \SD\left(\cN(\mu_i,\sigma^2_i\id_{d \times d}),
        \cN(\wh{\mu}_i,\wh{\sigma}^2_i\id_{d \times d})\right)
        \leq \alpha.$$
\end{lemma}
\begin{proof}
    We know that from Lemma \ref{lem:easy-hd-clustering}
    that all points in $Z$ get correctly classified
    as per their respective labels by the algorithm
    with probability at least $1-11\beta$, and that
    we have centers and radii of private balls
    that completely contain one unique component
    each. In other words, for all $1 \leq i \leq k$,
    we have a set $\wh{S}_i$ that contains all points
    from component $i$, such that
    $\wh{S}_i \subset \ball{\wh{c}_i}{\wh{r}_i}$,
    where $\wh{r}_i \in \Theta(\sqrt{d}\sigma_i)$.

    Now, from Theorem \ref{thm:easy-learn-gaussian}
    and our bound on $n$, since $Z$ satisfies Condition
    \ref{cond:num-points}, we know that for each
    $1 \leq i \leq k$, with probability at least
    $1-\tfrac{\beta}{k}$, we output $\wh{\mu}_i$
    and $\sigma_i$, such that
    $$\llnorm{\mu_i-\wh{\mu}_i} \leq O(\alpha)~~~
        \text{and}~~~
        1-O\left(\frac{\alpha}{\sqrt{d}}\right) \leq
        \frac{\wh{\sigma}^2_i}{\sigma^2_i} \leq
        1+O\left(\frac{\alpha}{\sqrt{d}}\right).$$
    This implies from Lemma
    \ref{lem:gaussian-norm-translation} that
    $$\SD\left(\cN(\wh{\mu}_i,\wh{\sigma}^2_i \id_{d \times d}),
        \cN(\mu_i,\sigma^2\id_{d \times d})\right) \leq \alpha.$$
    By applying the union bound, finally, we
    get the required result.
\end{proof}

\subsubsection{Putting It All Together}

Given all the results above, we can finally
complete the proof of Theorem~\ref{thm:easycase}.

\begin{proof}[Proof of Theorem~\ref{thm:easycase}]
    Let $G_1,\dots,G_k$ be the actual Gaussians
    in $\cD$, and let $\wh{G}_1,\dots,\wh{G}_k$
    be their respective estimations. Now that
    the individual estimated Gaussians are close
    to within $\alpha$ in TV-distance to the
    actual respective Gaussians, we can say the
    following about closeness of the two mixtures.
    \begin{align*}
        \SD(\wh{G},\cD) &= \max\limits_{S \in \R^d}
                \abs{\wh{G}(S) - \cD(S)}\\
            &= \frac{1}{k}\max\limits_{S \in \R^d}
                \abs{\sum\limits_{i=1}^{k}{\wh{G}_i(S)-G_i(S)}}\\
    %        &\leq \frac{1}{k}\max\limits_{S \in \R^d}
    %            \left\{\sum\limits_{i=1}^{k}\abs{{\wh{G}_i(S)-G_i(S})}\right\}\\
            &\leq \frac{1}{k}\sum\limits_{i=1}^{k}{\max\limits_{S \in \R^d}
                \abs{{\wh{G}_i(S)-G_i(S})}}\\
            &= \frac{1}{k}\sum\limits_{i=1}^{k}{\SD(\wh{G}_i,G_i)}\\
            &\leq \alpha
    \end{align*}
    Here, the last inequality follows with
    probability at least $1-12\beta$ from
    Lemma \ref{lem:easy-individual-gaussian}.
    Note that our algorithms and theorem
    statements along the way required various
    regularity conditions. By
    Lemma~\ref{lem:easy-dr-conditions} and~\ref{lem:all-dr-conditions},
    these events all occur with probability at
    least $1-11\beta$, which gives us the success
    probability of at least $1-23\beta$. This
    completes our proof of Theorem~\ref{thm:easycase}.
\end{proof}

%% file: hardcase.tex
\section{An Algorithm for Privately Learning Mixtures of Gaussians}
\label{sec:hardcase}

In this section, we present our main algorithm for
privately learning mixtures of Gaussians and prove the
following accuracy result.
\begin{theorem}\label{thm:pgme-accurate}
    For all $\eps,\delta,\alpha,\beta > 0$,
    there exists an $(\eps,\delta)$-DP
    algorithm that takes $n$ samples from
    $\cD \in \cG(d,k,\minsigma,\maxsigma,R,\mw,s)$,
    where $\cD$ satisfies (\ref{eq:conditionVariances}),
    and $s$ is defined by
    the following separation condition
    \[ \forall i,j,~~~~\|\mu_i-\mu_j\|_2 \geq
        100 (\sigma_i+\sigma_j) \left( \sqrt{k\log(n)} +
        \frac 1 {\sqrt{w_i}} + \frac 1 {\sqrt{w_j}}\right) \]
    and returns
    a mixture of $k$ Gaussians $\wh{\cD}$,
    such that if 
    \begin{align*}
        n \geq \max\Big\{  &\tilde\Omega\left( \left(\frac
            {\sqrt{dk}\log(dk/\beta)\log^{\frac{3}{2}}(1/\delta)\log\log
            ((R+\sqrt{d}\maxsigma)/\minsigma)}{w_{\min}
            \epsilon}\right)^{\frac 1 {1-a}}\right)
            \textrm{ for an arbitrary constant }a>0 ,~~
            \cr & \Omega\left( \frac{k^{9.06}d^{3/2}
            \log(1/\delta)\log(k/\beta)}{w_{\min}\epsilon}\right), \cr
            & \Omega\left(\frac{k^{\frac{3}{2}}
            \log(k\log((R+\sqrt{d}\maxsigma)/\minsigma)/\beta)\log(1/\alpha)
            \log(1/\delta)}
            {\alpha\epsilon w_{\min}}\right),~~~
            \frac{1}{\mw}(n_1 + n_2)\Big\} 
    \end{align*}
    where,
    \begin{align*}
        n_1 &\geq \Omega\left(\frac{(d^2 + \log(\frac{k}{\beta}))
            \log^2(1/\alpha)}{\alpha^2}
            + \frac{(d^2\polylog(\frac{dk}{\alpha\beta\eps\delta}))}{\alpha\eps}
            + \frac{d^\frac{3}{2}\log^{\frac{1}{2}}(\frac{\maxsigma}{\minsigma})
            \polylog(\frac{dk\log(\maxsigma/\minsigma)}{\beta\eps\delta})}{\eps}\right)\\
        n_2 &\geq \Omega\left(\frac{d\log(\frac{dk}{\beta})\log^2(1/\alpha)}{\alpha^2}
            + \frac{d\log(\frac{dk\log{R}\log(1/\delta)}{\beta\eps})
            \log^{\frac{1}{2}}(\frac{1}{\delta})\log^2(1/\alpha)}{\alpha\eps}
            + \frac{\sqrt{d}\log(\frac{Rdk}{\beta})\log^{\frac{1}{2}}(\frac{1}{\delta})}
            {\eps}\right)
    \end{align*}
    then it $(\alpha,\beta)$-learns $\cD$.
\end{theorem}

\input{TerrificBall}

\subsection{The Algorithm}
\label{subsec:algorithm_description}

We now introduce our overall algorithm for GMM clustering
which mimics the approach of Achlioptas and McSherry~\cite{AchlioptasM05}.
Recall, Achlioptas and McSherry's algorithm correctly learns
the model's parameters, provided that 
\begin{equation}\label{eq:requiredSeparation}
    \forall i,j,~~~~\|\mu_i-\mu_j\|_2 \geq
        C (\sigma_i+\sigma_j) \left( \sqrt{k\log(nk/\beta)} +
        \frac 1 {\sqrt{w_i}} + \frac 1 {\sqrt{w_j}}\right) 
\end{equation}
for some constant $C>0$, and that $n\geq \poly(n,d,k)$.
We argue that under the same separation condition (albeit
we replace the constant of~\cite{AchlioptasM05} with a much
larger constant, say $C=100$) a $(\eps,\delta)$-differentially
private algorithm can also separate the $k$-Gaussians and
learn the model parameters. Alas, we are also forced to
make some additional assumptions. First, we require a bound
on the distribution parameters, that is, the norm of all
means, and the eigenvalues of covariances; this is due to
standard results in DP literature (\cite{BunNSV15}) showing
the necessity  of such bounds for non-trivial DP algorithms.
The good news is that
we only require $\log\log$-dependence on these parameters.
Namely, our sample complexity is now
$\poly(n,k,d,\frac 1\epsilon, \log(1/\delta),\log(1/\beta),
\log\log((R+\sqrt{d}\maxsigma)/\minsigma))$. Secondly, because
we replace the non-private Kruskal based algorithm with an
algorithm that locates a terrific ball, we are forced to use
one additional assumption --- that the Gaussians are not
``too flat''. The reason for this requirement is that we are
incapable of dealing with the case that the pairwise distances
of points drawn from the \emph{same} Gaussian have too much
variation in them (see Property~\ref{property:DistTwoPoints}
below). Formally, we require the following interplay between
$n$ (the number of datapoints), $\beta$ (probability of error),
and the norms of each Gaussian variance:
\begin{equation}\label{eq:conditionVariances} 
    \forall i, ~~~ \|\Sigma_i\|_F\sqrt{\log(nk/\beta)} \leq
    \tfrac 1 8 \tr(\Sigma_i) ,~~\textrm{ and } ~~
    \|\Sigma_i\|_2{\log(nk/\beta)} \leq
    \tfrac 1 8 \tr(\Sigma_i).
\end{equation}
Note that for a spherical Gaussian (where
$\Sigma_i = \sigma_i^2 I_{d\times d}$) we have that
$\tr(\Sigma_i) = d\sigma_i^2$, while
$\|\Sigma_i\|_F = \sigma_i^2\sqrt d$ and $\|\Sigma_i\|_2=\sigma_i^2$,
thus, the above condition translates to requiring a
sufficiently large dimension. This assumption was explicitly
stated in the non-private work regarding learning spherical
Gaussians~\cite{VempalaW02}
(also, we ourselves made such an assumption in the
simplified version detailed in Section~\ref{sec:easycase}).

We now detail the main component of our GMM learner.
Algorithm~\ref{alg:LearnMixture} takes $X$, the given
collection of datapoints, and returns a $k$-partition
of $X$ in the form of list of subsets. We thus focus
on finding the correct partition. Note that intrinsically,
the returned partition of indices cannot preserve privacy
(it discloses the cluster of datapoint $i$), so once this
$k$-partition is done, one must apply the existing
$(\eps,\delta)$-DP algorithms for finding the mean and
variance of each cluster, as well as apply $(\eps,0)$-DP
histogram in order to assess the cluster weights. This
is the overall algorithm ($\PGME$) given in
Algorithm~\ref{alg:OverallLearnerOfMixture}.

In our analysis, we prove that our algorithm is 
$(O(\eps\sqrt{k\log(1/\delta)}),O(k\delta))$-DP, and
accurate with probability at least $1-O(k\beta)$.
The desired $(\eps,\delta)$-DP and error probability $1-\beta$ guarantees are achieved by appropriately rescaling the values of $\eps, \delta$, and $\beta$.

\begin{algorithm}[h!] \label{alg:LearnMixture}
	\caption{Recursive Private Gaussian Mixture Partitioning 
		}
    \KwIn{Dataset $X \in \R^{n\times d}$ coming from a
        mixture of at most $k$ Gaussians.
        Upper bound on the number of components $k$.
        Bounds on the parameters of the GMM $\mw, \minsigma, \maxsigma$.
        Privacy parameters $\eps, \delta>0$.}
	\KwOut{Partition of $X$ into clusters.} 
	\vspace{10pt}
	$\RPGMP(X; k, R, \mw, \minsigma, \maxsigma, \epsilon,\delta)$:
	\begin{enumerate}
		\parskip=0pt
		\item \If {($k=1$)}{
                \qquad Skip to last step (\#8).
            }
        \item Find a small ball that contains $X$, and bound
            the range of points to within that ball:\\
            Set $n' \gets \abs{X} + \Lap(2/\eps) - \tfrac{n\mw}{20}$.\\
            $\ball{p}{r''} \gets \PGLoc(X,n';
            \tfrac{\eps}{2}, \delta, R,\minsigma,\maxsigma)$.\\
            Set $r \gets 12r''$.\\
            Set $X \gets X \cap \ball{p}{r}$.
        \item Find $5$-{\tt PTerrificBall} in $X$ with
            $t = \tfrac{n\mw}{2}$:\\
            $\ball{p'}{r'} \gets \PTB(X, \tfrac{n\mw}{2}, c=5, largest={\tt FALSE};
            \eps, \delta, r, \tfrac{\sqrt{d}\minsigma}{2})$.
	    \item If the data is separable already, we recurse on each part:\\
            \If {($\ball{p'}{r'}\neq \bot$)}{
                \qquad Partition $X$ into $A = X\cap  \ball{p'}{r'}$ and
                    $B = X\setminus \ball{p'}{5r'}$.\\
                \qquad Set $C_A \gets
                    \RPGMP(A; k-1, R, \mw, \minsigma, \maxsigma, \eps, \delta)$.\\
                \qquad Set $C_B \gets
                    \RPGMP(B; k-1, R, \mw, \minsigma, \maxsigma, \eps, \delta)$.\\
	            \qquad \Return{$C_A \cup C_B$}.
            }
    	\item Find a private $k$-PCA of $X$:\\	
	        Sample $N$, a symmetric matrix whose entries are
            taken from $\cN(0, \frac{4r^4\ln(2/\delta)}{\eps^2})$.\\
            $\Pi \in \R^{d\times d} \gets k\textrm{-PCA projection of } X^TX + N$,
            where $\Pi$ is a rank $k$ matrix. 
        \item Find $5$-{\tt PTerrificBall} in $X\Pi$ with
            $t = \tfrac{n\mw}{2}$:\\
            $\ball{p'}{r'} \gets \PTB(X\Pi, \tfrac{n\mw}{2}, c=5, largest={\tt TRUE};
            \eps, \delta, r, \tfrac{\sqrt{k}\minsigma}{2})$.\\
	    \item If the projected data is separable, we recurse on each part:\\
            \If {($\ball{p'}{r'}\neq \bot$)}{
                \qquad Partition $X$ into $A = \{x_i \in X:~ \Pi x_i \in \ball{p'}{r'}\}$
                    and $B = \{x_i \in X:~ \Pi x_i \not\in \ball{p'}{5r'}\}$.\\
                \qquad Set $C_A \gets
            	    \RPGMP(A; k-1, R, \mw, \minsigma, \maxsigma, \eps, \delta)$.\\
                \qquad Set $C_B \gets
                    \RPGMP(B; k-1, R, \mw, \minsigma, \maxsigma, \eps, \delta)$.\\
                \qquad \Return{$C_A \cup C_B$}.
	        }
	    \item Since the data isn't separable, we treat it as a single Gaussian:\\
            Set a single cluster: $C \gets \{i: x_i\in X\}$.\\
            \Return{$\{C\}$.}
	\end{enumerate}
\end{algorithm}

\begin{algorithm}[h!] \label{alg:OverallLearnerOfMixture}
	\caption{Private Gaussian Mixture Estimation
		}
    \KwIn{Dataset $X \in \R^{n\times d}$ coming from
        a $k$-Gaussian mixture model.
        Upper bound on the number of components $k$.
        Bounds on the parameters of the GMM $\mw, \minsigma, \maxsigma$.
        %, such that each $x_i\in X$ has length $\leq O(R + \maxsigma \sqrt{d})$.
		Privacy parameters $\eps, \delta>0$;
        failure probability $\beta > 0$.}
	\KwOut{Model Parameters Estimation} 
	\vspace{10pt}
	$\PGME(X; k, R, w_{\min}, \minsigma, \maxsigma,
        \epsilon,\delta,\beta)$:
	\begin{enumerate}
		\parskip=0pt
        \item Truncate the dataset so that for all points,
            $\|X_i\|_2 \leq O(R + \maxsigma\sqrt{d})$ 
		\item $\{C_1,..,C_k\} \gets \RPGMP(X;k,R,w_{\min},
            \minsigma, \maxsigma, \eps,\delta)$
		\item \For {$j$ from $1$ to $k$}{
			    \qquad $(\mu_j, \Sigma_j) \gets
                    \PGE( \{x_i: i\in C_j\}; R, \minsigma, \maxsigma, \eps,\delta)$.\\			
    			\qquad $\tilde n_j \gets |C_j|+\Lap(1/\epsilon)$.
			}
		\item Set weights such that for all $j$,
            $w_j \gets \tilde n_j / (\sum_{j}\tilde n_j)$.
		\item \Return{ $\langle \mu_j, \Sigma_j, w_j\rangle_{j=1}^k$ }
	\end{enumerate}
\end{algorithm}

\begin{theorem}\label{thm:HardcaseAlgPrivacy}
	Algorithm~\ref{alg:OverallLearnerOfMixture} satisfies
    $\big(2\epsilon+8\epsilon\sqrt{2k\log(1/\delta)}),
    (8k+1)\delta\big)$-differential privacy.
\end{theorem}
\begin{proof}
	Consider any two neighboring datasets $X$ and $X'$,
    which differ on at most one point. In each of
    the levels of the recursion in $\RPGMP$ we apply three
    $(\eps,\delta)$-DP, one $(\tfrac{\eps}{2},\delta)$-DP,
    and one $(\tfrac{\eps}{2},0)$-DP procedures to the data, as well as
    fork into one of the two possible partitions and invoke
    recursive calls on the part of the data the contains
    the point that lies in at most one of $X$ and $X'$.
    Since the recursion can be applied
    at most $k$ times, it follows that the point plays a
    role in at most $k$ rounds, each of which is $(4\eps,4\delta)$-DP.
    By advanced composition (Lemma~\ref{lem:composition}), overall $\PGME$ is
    $(4\eps \cdot \sqrt{8k\log(1/\delta)},8k\delta)$-DP.
    The additional calls for learning the mixture model
    parameters on each partition are $(2\eps,\delta)$-DP.
    Summing the two parts together yields the overall
    privacy-loss parameters.
\end{proof}

Since the latter steps of Algorithm~\ref{alg:OverallLearnerOfMixture}
rely on existing algorithms with proven utility, our utility analysis
boils down to proving the correctness of Algorithm~\ref{alg:LearnMixture}.
We argue that given $X$ taken from a Gaussian Mixture Model with
sufficiently many points and sufficient center separation, then with
high probabiliy, Algorithm~\ref{alg:LearnMixture} returns a $k$-partition
of the data which is \emph{laminar} with the data. In fact, the algorithm
might omit a few points in each level of the recursion, however, our
goal is to argue that the resulting $k$-subsets are pure~--- holding
only datapoints from a single cluster.

\begin{defn}
	\label{def:LaminarPartition}
	Given that $X$ is drawn from a $k'$-GMM, we call two disjoint
    sets $A$, $B \subseteq X$ a \emph{laminar partition} of $X$ if (i) there exists
    a partition of the set $\{1,2,..,k'\}$ into $S$ and $T$ such that
    all datapoints in $A$ are drawn from the Gaussians indexed by some
    $i\in S$ and all datapoints in $B$ are drawn from the Gaussians
    indexed by some $i\in T$, and (ii)
    $|X\setminus (A\cup B)|\leq \tfrac{n\mw\alpha}{10k'\log(1/\alpha)}$.
\end{defn}

Towards the conclusion that the partition returned by
Algorithm~\ref{alg:LearnMixture} is laminar, we require
multiple claims regarding the correctness of the algorithm
through each level of the recursion. These claims, in turn,
rely on the following properties of Gaussian data.

\begin{enumerate}
    \item The Hanson-Wright inequality (Lemma~\ref{lem:HW}):
        $\forall i, \forall x\sim \cN(\mu_i, \Sigma_i)$ we have that
        \[\tr(\Sigma_i) - 2\|\Sigma_i\|_F\sqrt{\log(n/\beta)}
            \leq \|x-\mu_i\|_2^2 \leq \tr(\Sigma_i) +
            2\|\Sigma_i\|_F\sqrt{\log(n/\beta)} + 2\|\Sigma_i\|_2\log(n/\beta).\]
        Using the assumption that
        $2\|\Sigma_i\|_F\sqrt{\log(n/\beta)}\leq \tfrac 1 4 \tr(\Sigma_i)$
        and that $2\|\Sigma_i\|_2\log(n/\beta) \leq \tfrac 1 4 \tr(\Sigma_i)$
        (Equation~\ref{eq:conditionVariances}) we get that
        $\tfrac{3}{4} \tr(\Sigma_i) \leq \|x-\mu_i\|^2 \leq\tfrac 3 2 \tr(\Sigma_i)$.
        \label{property:DistFromCenter}
	\item Similarly, for every $i$ and for any
        $x,y \sim \cN(\mu_i,\Sigma_i)$ we have that
        $x-y\sim \cN(0,2\Sigma_i)$ implying under the same
        concentration bound that
        $\tfrac{3}{2}\tr(\Sigma_i)\leq \|x-y\|_2^2 \leq 3\tr(\Sigma_i)$.
        \label{property:DistTwoPoints}
	\item For every $i$, its empirical spectral norm is close
        to the expected spectral norm, assuming for all $i$ we
        have $w_i n = \Omega(d + \log(1/\beta))$ (see
        Lemma~\ref{lem:spectralnormofgaussians}). Namely,
        \[\tfrac 1 {n_i} \llnorm{\sum\limits_{\{x\in X:~ x\sim \cN(\mu_i,\Sigma_i)}
            (x-\mu_i)(x-\mu_i)^T} \in \left(\tfrac 1 2 \|\Sigma_i\|_2,
            \tfrac 3 2\|\Sigma_i\|_2 \right).\]
	\item Under center separation, we have that
        $\forall i, \forall x\sim\cN(\mu_i,\Sigma_i)$ and for
        all $j\neq i$ we have that $\|x-\mu_i\|_2 \leq \|x-\mu_j\|_2$.
        This is a result of the fact that when projecting $x$ onto
        the line that connects $\mu_i$ and $\mu_j$ the separation
        between $\mu_i$ and $\mu_j$ is overwhelmingly larger than
        $(\sigma_i+\sigma_j)\sqrt{\ln(n/\beta)}$ (a consequence of the separation condition combined with (\ref{eq:conditionVariances})). \label{property:NearestCenter}
	\item Lastly, we also require that for any $i$ the number
        of points drawn from the $i$th Gaussian component is
        roughly $w_i n$. In other words, we assume that for all
        $i$ we have $\{x\in X:~ x\textrm{ drawn from }i\}$ is
        in the range $(\frac {3w_i n} 4, \frac{5w_i n}4)$.
        Standard use of the multiplicative Chernoff bound suggests
        that when $n = \Omega(\log(1/\beta)/w_i)$ then for each
        cluster $i$ the required holds.
\end{enumerate}
Thus, we have (informally) argued that, with probability
$\geq 1- 5\beta k$, the given dataset satisfies all of
these properties. Next we argue a few structural propositions
that follow.

\begin{prop}\label{prop:ABallWithCenterAndPoint}
	Let $\ball p r$ be a ball such that for
    some $i$, both $\mu_i$ and some
    $x\sim\cN(\mu_i,\Sigma_i)$ lie in $\ball p r$.
    Then $\ball p {4r}$ holds all datapoints drawn
    from $\cN(\mu_i,\Sigma_i)$.
\end{prop}
\begin{proof}
	Based on Property~\ref{property:DistFromCenter}
    it follows that $\ball{\mu_i}{\sqrt 2\|\mu_i-x\|_2}$
    holds all datapoints drawn from the $i$th Gaussian.
    As $\|\mu_i-x\|_2\leq 2r$ we have that
    $\ball{\mu_i}{\sqrt 2\|\mu_i-x\|_2}\subset\ball{p}{4r}$.
\end{proof}

\begin{prop}\label{prop:ABallWithTwoPoints}
	Let $\ball p r$ be a ball such that for some
    $i$, two points $x,y~\sim\cN(\mu_i,\Sigma_i)$
    lie in $\ball p r$. Then $\ball p {4r}$ holds
    all datapoints drawn from $\cN(\mu_i,\Sigma_i)$.
\end{prop}
\begin{proof}
	Based on Property~\ref{property:DistTwoPoints},
    we have that the ball $\ball{x}{\sqrt 2\|x-y\|_2}$
    holds all datapoints drawn from the $i$th Gaussian.
    As $\|x-y\|_2\leq 2r$ it follows that
    $\ball{x}{\sqrt 2 \|x-y\|_2} \subset\ball p {4r}$.
\end{proof}

The remainder of the utility proof also relies on the
assumption that at all levels of the recursion,
all subprocedures (which are probabilistic in nature)
do not terminate per-chance or with a non-likely output.
Note that this assumption relies in turn on a sample
size assumption: we require that $\frac{w_{\min}n}2$
is large enough for us to find a terrific ball, namely,
we require
$$n = \tilde\Omega\left( \left(\frac
    {\sqrt d\log(d/\beta)\log(1/\delta)\log\log((R+\sqrt{d}\maxsigma)/\minsigma)}
    {w_{\min}\epsilon}\right)^{\frac 1 {1-a}}\right)$$
for some small constant $a>0$ (say $a=0.1$).
Recall that at any level of the recursion we deal with
at most two such subprocedures.
As we next argue, (a) at each level of the recursion we
partition the data into two parts, which are each laminar with a mixture
of $k'<k$ Gaussians, and (b) when the given input contains
solely points from a single Gaussian we reach the bottom
level of the recursion. It follows that the recursion tree
(a binary tree with $k$ leaves) has $2k$ nodes.
Since each instantiation of the recusion has failure probability at most $6\beta$ (a consequence of Claim~\ref{clm:pgloc}, Claim~\ref{clm:TerrificBallIsLaminar}, Claim~\ref{clm:NoTerrificBallThenBoundedRadius}, Corollary~\ref{cor:CloseProjectedCenter}, Claim~\ref{clm:ProjectedDataTerrificBallLaminar}, and Claim~\ref{clm:OneComponentNoTerrificBall}), this implies that that with probability $\geq 1- 12k\beta$, all subroutines
calls ever invoked by $\RPGMP$ are successful. Note that
the overall success probability of Algorithm~\ref{alg:OverallLearnerOfMixture}
is thus $1-19k\beta$ (with probability $\leq 5k\beta$,
there exists some cluster, whose datapoints don't satisfy
properties 1-5; with probability $\leq 12k\beta$, a call
made through Step $2$ of Algorithm~\ref{alg:OverallLearnerOfMixture}
fails; and with probability $\leq 2k\beta$, a failure occurs
at Steps $3$ and $4$ of Algorithm~\ref{alg:OverallLearnerOfMixture}).
We continue assuming no such failure happens.

Subject to successfully performing each subprocedure,
the following claims yield the overall algorithm's
correctness.

\begin{claim}\label{clm:pgloc}
    Let $(X_1,\dots,X_n)$ be samples from
    a mixture of $k$ Gaussians in $\PGME$.
    Suppose $X$ (a subset of the samples) is the
    input dataset in $\RPGMP$, such that (1) for each
    $1 \leq i \leq k$, the number of points
    in $X$ belonging to component $i$ is either
    equal to $0$ or greater than $\Omega(n\mw)$; and
    (2) $X$ contains points from at least one component.
    If
    $$n \geq \left(\frac{\sqrt{d}k}
        {\eps}\right)^{\frac{1}{1-a}}
        \cdot 9^{\log^*\left(\sqrt{d}\left(\frac{R\maxsigma}
        {\minsigma}\right)^{d}\right)}
        \cdot \polylog\left(d,
        \frac{1}{\eps},\frac{1}{\delta}
        \frac{1}{\beta},\frac{1}{\gamma}\right)
        + O\left(\frac{d + \log(k/\beta)}{\mw}
        + \frac{\log(1/\beta)}{\eps\mw}\right),$$
    where $a>0$ is an arbitrarily small constant,
    then with probability at least $1-\beta$,
    the ball $\ball{r}{p}$ found in Step 2 of
    $\RPGMP$ contains all points in $X$, and
    $r \leq cr_{opt}$, where $c>0$ is a constant (which depends on $a$),
    and $r_{opt}$ is the radius of the smallest
    ball that contains all points in
    $X$.
\end{claim}
\begin{proof}
    First, note that from Lemma~\ref{lem:lap-conc},
    we know that with probability at least $1-\tfrac{\beta}{2}$,
    the magnitude of the Laplace noise added in
    Step 2 of $\RPGMP$ is at most $\tfrac{2}{\eps}\log(2/\beta)$.
    Since for every component that has
    points in $X$, there are at least $\Omega(n\mw)$ points, we know from our bound on
    $n$ that this magnitude of Laplace noise is
    at most $\tfrac{n\mw}{20}$. Therefore,
    \begin{align*}
        n' &= n + \Lap(2/\eps) - \frac{n\mw}{20}\\
            &\geq n\left(1 - \frac{\mw}{10}\right) \geq n/2,
    \end{align*}
    which means that by asking $\PGLoc$ to search
    for $n'$ points, we are covering at least half
    of every component that has points in $X$.

    Next, by Theorem~\ref{thm:pgloc}, we know that
    with probability at least $1-\tfrac{\beta}{2}$,
    if the radius of the smallest ball covering $n'$
    points is $r'$, we will get a ball of
    radius $r'' = O(r')$ that covers at least $\tfrac{n'}{2}$
    points. Let $\ball{p}{r''}$ be the returned ball.
    There are two cases to consider: (1) the ball
    covers at least two points from all components
    in $X$; and (2) it covers at most one point
    from at least one component in $X$ (which we call
    ``completely uncovered'' for brevity).
    In the first case, because all points in $X$ satisfy
    Property~\ref{property:DistTwoPoints},
    multiplying $r''$ by $4$ would cover
    all the points in $X$. In the second case, since
    $n'$ is very large, there must be at least two points
    from every completely uncovered component in $X \setminus \ball{p}{r''}$,
    that lie together in some optimal ball containing $n'$
    points, and one point in $X \cap \ball{p}{r''}$ that
    lies in the same optimal ball as those two points.
    Consider any completely uncovered component,
    and let $y,z$ be two such points from it, with $x$ being
    a corresponding point in $X \cap \ball{p}{r''}$. Then
    by triangle inequality,
    \begin{align*}
        \llnorm{p-y} &\leq \llnorm{p-x} + \llnorm{x-y}\\
            &\leq r'' + 2r'\\
            &\leq 3r'',
    \end{align*}
    which also holds for $\llnorm{p-z}$.
    Therefore, multiplying $r''$ by $3$ would cover
    both $y$ and $z$. Since the choice of completely uncovered
    components was arbitrary, this holds for all completely
    uncovered components. Again, since all points in $X$
    satisfy Property~\ref{property:DistTwoPoints},
    multiplying $3r''$ by $4$ (to get $r$) fully covers
    all points in completely uncovered components, and all
    other points in $X$ as well.
    
    Finally, $r' \leq r_{opt}$, which means that
    $r \in O(r_{opt})$. This completes our proof of
    the claim.
\end{proof}

We now describe our strategy for proving the main theorem of this subsection (which will be Theorem~\ref{thm:privateMainThmHardCase} below).
We will begin by establishing a series of claims (first in prose, then with formal statements), and show how they imply our main theorem statement.
We then conclude the subsection by substantiating our claims with proofs.

First, we argue that if the mixture contains at least two components and Step 3 of the algorithm finds a terrific ball to split the data (in the original, non-projected space), then the partition that is induced by this ball will split the data in a laminar fashion.
That is, it will call the algorithm recursively on two disjoint subsets of the data, such that for each component, all samples which were drawn from this component end up in only one of the two subsets.
This will ensure that the two recursive calls operate on valid (sub)instances of the GMM learning problem.

\begin{claim}\label{clm:TerrificBallIsLaminar}
	If the dataset $X$ contains samples from a mixture
    of $k\geq 2$ Gaussians, and Step 3 of $\RPGMP$
    finds a terrific ball over the data, then with
    probability at least $1-\beta$, the
    partition formed by the terrific ball is laminar.	
\end{claim} 

On the other hand, if we still have at least two components and the algorithm is unable to find a terrific ball, we argue that all points lie in a ball of limited radius.
Since the sample complexity of our private PCA algorithm (Lemma~\ref{lem:PCA_AM_style} and Lemma~\ref{lem:data-spectral}) depends on the radius of this ball, this will allow us to bound the sample complexity of this step.

\begin{claim}\label{clm:NoTerrificBallThenBoundedRadius}
	If the dataset $X$ contains samples from a mixture
    of $k\geq 2$ Gaussians with $\maxsigma^2$ denoting
    the largest directional variance of any
    component in the mixture, and Step 3 of $\RPGMP$
    does not find a terrific ball over
    the data, then with probability at least
    $1-\beta$, the radius $r$ of the entire
    instance (found in Step 2) is at most
    $k^{4.53}\sqrt{d}\maxsigma$, that is,
    all points of $X$ lie in a ball of radius
    $k^{4.53}\sqrt{d}\maxsigma$.
\end{claim} 

In the same setting, after the PCA projection, we have that the projected means are close to the original means, and that the resulting data will have a terrific ball (due to the separation between components).

\begin{coro}\label{cor:CloseProjectedCenter}
	Under the same conditions as in
    Claim~\ref{clm:NoTerrificBallThenBoundedRadius},
    we have that if
    $$n\geq O\left(\frac{k^{9.06}d^{3/2}\sqrt{\log(2/\delta)}
        \log(1/\beta)}{w_{\min}\epsilon}\right),$$
    then with probability at least $1-\beta$, we have that
    for each center $\mu_i$, the corresponding projected
    center $\hat\mu_i$ is within distance
    $\leq \tfrac{3\maxsigma}{\sqrt{w_i}}$. As a result,
    under our center-separation condition, the projected
    data $X\Pi$ has a $21$-terrific ball.
\end{coro}

Similar to Claim~\ref{clm:TerrificBallIsLaminar}, we conclude this case by arguing that the partition formed by the resulting terrific ball to split the data (in the projected space) is laminar, resulting in the two recursive calls being made to valid (sub)instances of the problem.

\begin{claim}\label{clm:ProjectedDataTerrificBallLaminar}
	Under the same conditions as in
    Claim~\ref{clm:NoTerrificBallThenBoundedRadius},
    with probability at least $1-\beta$,
    the partition formed by the terrific ball found
    on the projected data $X\Pi$ is laminar.
\end{claim}

Finally, if our dataset is generated from a single component, then we will not locate a terrific ball, and the recursion ceases.
\begin{claim}\label{clm:OneComponentNoTerrificBall}
	For any $i$, if the dataset $X$ is composed of
    at least
    \[\frac{400 w_i d \log(1/\beta w_{\min})}{w_{\min}} + \Omega\left(\frac{1}{\varepsilon}\left(\log \log \left((R + \sqrt{d}\maxsigma)/\minsigma\right) + \log(1/\beta)\right)\right)\]
    points drawn only from $\cN(\mu_i,\Sigma_i)$,
    then with probability at least $1-\beta$,
    neither Step 3 nor Step 6 will locate a terrific ball. %$X$ nor $X\Pi$ have a terrific ball.
\end{claim}

Putting together all of these claims and the entire
discussion, we have that following theorem.
\begin{theorem}\label{thm:privateMainThmHardCase}
    Algorithm~\ref{alg:LearnMixture} satisfies
    $\big(8\epsilon\sqrt{2k\log(1/\delta)}), 8k\delta\big)$-DP,
    and under the center-separation of~\eqref{eq:requiredSeparation},
    with probability $\geq 1- 12k\beta$, it returns $k$-subsets
    which are laminar with the $k$ clusters, while omitting
    no more than $\tfrac{\alpha w_{\min} n}{20\log(1/\alpha)}$
    % \geq 2k\Gamma = O(\frac k \epsilon \log(\log(R/\sigma_{\min})/\beta))$ 
	of the datapoints, provided that
	\begin{align*} 
		n &= \tilde\Omega\left( \left(\frac
            {\sqrt d\log(d/\beta)\log(1/\delta)\log\log((R+\sqrt{d}\maxsigma)/\minsigma)}
            {w_{\min}\epsilon}\right)^{\frac 1 {1-a}}\right)
            \textrm{ for an arbitrary constant }a>0 
        \cr n &= \Omega\left( \frac{k^{9.06}d^{3/2}\sqrt{\log(2/\delta)}
            \log(1/\beta)}{w_{\min}\epsilon}\right)  
		\cr n &= \Omega\left(\frac{d \log(1/\beta w_{\min})}{w_{\min}}\right)
        \cr n &= \Omega\left(
            \frac{k\log(1/\alpha)\log(\log((R+\sqrt{d}\maxsigma)/\minsigma)/\beta)}
            {\alpha\epsilon w_{\min}}\right).
	\end{align*}
\end{theorem}
The proof of Theorem~\ref{thm:privateMainThmHardCase}
follows from all the above mentioned properties of the
data and the claims listed above. It argues that in
each level of the recursion we are forced to make a
laminar partition with probability at least $1-6\beta$
(conditioned on the success in the previous levels)
until we reach a subset (of sufficient
size) which is contained in a single cluster, then we
halt. Since this implies that we ultimately return $k$
clusters, it means that there are at most $2k$ nodes
in the recursion tree, so the failure probability adds
up to $12k\beta$. The sample complexity bounds are the bounds required
for all claims and properties 1-5, where the last bound
guarantees that the total number of points omitted in the
overall execution of the algorithm doesn't exceed
$\frac {n\mw\alpha}{20\log (1/\alpha)}$ (at most
$O(k\Gamma)=O(\frac k \epsilon \log(\log((R+\sqrt{d}\maxsigma)/\minsigma)/\beta))$,
since the recursion tree has at most $k$ non-leaf nodes).
%By Lemmas~\ref{lem:meanofgaussiansubset}
%and~\ref{lem:spectralnormofgaussians}, this implies
%that even if all of the points were removed from the
%component with the smallest mixing weight, the overall
%shift in the estimated mean and variance of each
%component is bounded by $\alpha$.

\begin{proof}[Proof of Claim~\ref{clm:TerrificBallIsLaminar}]
    Let $\ball{p}{r}$ be the ball returned by Step 3 of
    the $\RPGMP$ algorithm. Let $x,y$ be two datapoints
    that lie inside the ball and belong to the same
    component $i$. It follows from Proposition~\ref{prop:ABallWithTwoPoints}
    that all datapoints from cluster $i$ lie inside
    $\ball{p}{4r}$, but since the annulus
    $\ball{p}{5r}\setminus \ball{p}{r}$ is effectively empty
    (contains significantly less than
    $\tfrac{n\mw\alpha}{20k\log(1/\alpha)}$ points), it should
    be the case that
    (almost) all of these datapoints lie in $\ball{p}{r}$
    itself, and in particular, no point from component $i$ lies
    outside of $\ball{p}{5r}$.

    It follows that any component $i$ with at least two datapoints
    that fall inside $\ball{p}{r}$ belongs to one side of the
    partition, and moreover, since the ball contains $>\tfrac{n\mw}{4}$
    datapoints, there exists at least one component $i$, such that all of
    its datapoints lie inside $\ball{p}{4r}$.
    
    Next, suppose that
    for some $j\neq i$, some datapoint $z$ drawn from the $j^{\text{th}}$
    component lies inside $\ball{p}{r}$. It follows
    %from Proposition~\ref{prop:ABallWithCenterAndPoint}
    that both $z$ and $\mu_i$ lie inside $\ball{p}{r}$, and so
    \[2r \geq \|z-\mu_i\|_2\geq \|z-\mu_j\|_2\geq \sqrt{\frac 3 4 \tr(\Sigma_j)},\]
    thus $r\geq \sqrt{\frac{3}{16}\tr(\Sigma_j)}$. Thus the ball
    of radius $4r$ centered at $z$, which is fully contained
    inside $\ball{p}{5r}$, has radius $\geq \sqrt{3\tr(\Sigma_j)}$.
    This is large enough to include all datapoints from cluster
    $j$ as well. Again, the fact that the annulus
    $\ball{p}{5r}\setminus\ball{p}{r}$ is effectively empty implies
    that (effectively) all points from cluster $j$ also belong to
    $\ball{p}{r}$.

    Lastly, note that $\geq nw_{\min}/4$ datapoints are left
    outside $\ball{p}{5r}$. This implies that at least some
    component is left outside of $\ball{p}{5r}$. Therefore, the
    partition $(A,B)$ formed by the terrific ball is a laminar
    partition of the dataset $X$.

    The failure probability of this claim is $\beta$ using
    Lemma~\ref{lem:terrificball} because the success of the
    event in this claim rests on $\PTB$ functioning correctly.
\end{proof}

\begin{proof}[Proof of Claim~\ref{clm:NoTerrificBallThenBoundedRadius}]
    Based on Property~\ref{property:DistFromCenter}, for
    any datapoint in $X$ and its respective center we have
    that their distance is at most
    $s\stackrel{\rm def}=\sqrt{1.5d}\sigma_{\max}$. Now,
    since we know that our procedure that looks for a
    $5$-terrific ball failed, then by
    Claim~\ref{clm:TerrificRadiusCorrectness_necessary} we
    know that the data holds no $21$-terrific ball. 

    Consider the following function $f:\N \rightarrow [k]$,
    \[f(i) = \min\limits_{\substack{T \subseteq [k],\\
        X \subseteq \bigcup\limits_{j \in T}{\ball{\mu_j}{is}}}}
        \abs{T},\]
    that is, the minimum number of balls of
    radius $i\cdot s$ centered at some $\mu_j$ that are required
    to cover all datapoints in $X$. The above paragraph implies
    that $f(1)\leq k$, so it is enough to show that
    $f(k^{\log_2(23)})\leq 1$, as it would imply that a ball of
    radius $O(k^{\log_2(23)}\sqrt{d}\maxsigma)$ covers the entire
    instance. We argue that if there exists
    no $21$-terrific ball in $X$, then for all $i$ such
    that $f(i)>1$, then $f(23i) \leq f(i)/2$, which by
    induction leads to
    $f(23^{\log_2(k)}) = f(k^{\log_2(23)}) \leq 1$.

    Fix $i$. Assume $f(i)>1$, otherwise we are done. By
    definition, there exists a subset $\{\mu_{j_1}, \mu_{j_2},..., \mu_{j_{f(i)}}\}$
    of the $k$ centers, such that $X$ is contained in
    $\bigcup_t \ball{\mu_{j_t}}{is}$. Pick any center
    $\mu$ in this subset. We know that $\ball{\mu}{is}$
    is not a $21$-terrific ball. Since it holds enough points (at least
    $n\mw$) and leaves out enough points (since $f(i)>1$),
    it must be the case that $\ball \mu {21is}\setminus \ball{\mu}{is}$
    is non-empty, that is, there exists a point $x \in X$
    that resides in $\ball{\mu}{21is}\setminus\ball{\mu}{is}$. This means
    that $\ball{\mu}{22is}$ holds the center $\mu'$ of the
    ball $\ball{\mu'}{is}$ that covers $x$, and therefore
    $\ball{\mu'}{is} \subset \ball{\mu}{23is}$. Note that
    by symmetry it also holds that $\mu \in \ball{\mu'}{22is}$
    and so  $\ball{\mu}{is} \subset \ball{\mu'}{23is}$.
    
    Now, draw a graph containing $f(i)$ nodes (one node per center
    $\mu_{j_t}$), and connect two nodes $\mu$ and $\mu'$ if
    $\ball{\mu}{is} \subset \ball{\mu'}{23is}$. This is a
    graph in which each node has degree $\geq 1$ because
    there is no $21$-terrific ball centered at the corresponding
    mean, and therefore,
    has a dominating set of size $\leq f(i)/2$. Hence, $X$
    is covered by balls of radius $23is$ centered at each
    node in this dominating set, implying that $f(23i)\leq f(i)/2$. 

    Finally, by Claim~\ref{clm:pgloc}, we have that
    the radius $r$ found is Step 2 is in
    $O(k^{\log_2(23)}\sqrt{d}\maxsigma) \in O(k^{4.53}\sqrt{d}\maxsigma)$.

    As before, the failure probability of this claim is $\beta$
    using Lemma~\ref{lem:terrificball}.
\end{proof}

\begin{proof}[Proof of Corollary~\ref{cor:CloseProjectedCenter}.]
    Under the same conditions as in
    Claim~\ref{clm:NoTerrificBallThenBoundedRadius},
    it holds that the added noise is such that
    $\|N\|_2 \leq \tfrac {2r^2\sqrt{d \ln(2/\delta)}\log(1/\beta)} {\epsilon}$ with probability $\geq 1 - \beta$.
    Using the fact that $r = k^{4.53}\sqrt{d}\maxsigma$,
    Lemmata~\ref{lem:PCA_AM_style} and~\ref{lem:data-spectral},
    along with our bound on $n$ (hence, on $n_i$), imply
    that with probability at least $1-\beta$, for any $i$,
    \begin{align*}
        \|\bar{\mu}_i- \Pi\bar{\mu}_i\| &\leq
                \sqrt{\frac 1 {n_i} \|A-C\|^2} +
                \sqrt{\frac {2r^2\sqrt {d \ln(2/\delta)}\log(1/\beta)}
                {\epsilon n_i}}\\
            &\leq \frac{1}{\sqrt{w_{i}}} \left( 4\sqrt{2}\maxsigma +
                \maxsigma \right)\\
            &\leq 7 \frac{\sigma_{\max}}{\sqrt{w_{i}}}\\
            &\leq 7 \frac{\sigma_{\max}}{\sqrt{w_{\min}}}.
    \end{align*}

    Without loss of generality, assume cluster $1$
    is the Gaussian of largest variance. It follows
    that for all $i\neq 1$, we have
    \[ \|\Pi\mu_1-\Pi\mu_i \|_2 \geq
        100(\sigma_1+\sigma_i)\left(\sqrt {k\log(n)} +
        \frac{1}{\sqrt{w_1}} + \frac{1}{\sqrt{w_i}}\right) -
        7\left(\frac{\sigma_{1}}{\sqrt{w_1}} +
        \frac{\sigma_{1}}{\sqrt{w_i}} \right)
        \geq 100 \sigma_1\sqrt{k\log(n)}.\]
    Yet, similar to the analysis in~\cite{VempalaW02, AchlioptasM05},
    we have that $\|\Pi (x-\mu_i) \|_2 \leq \sqrt {k \sigma_i^2 \log(n)}$
    for each datapoint and its respective cluster
    $i$.\footnote{In fact, due to the fact that
    $\Pi$ was derived via in a differentially private
    manner, is is easier to argue this than in
    the~\cite{AchlioptasM05} paper, see the
    following~\href{http://windowsontheory.org/
    2014/02/04/differential-privacy-for-measure-concentration}{blog post}.} 
    Roughly, the argument is that, if we draw a sample from a Gaussian in a $k$ dimensional space, its $\ell_2$ distance from its mean is $O(\sigma \sqrt{k \log n})$.
    The same argument holds if the sample is drawn in a $d$ dimensional space and then projected into $k$ dimensions, as long as the projection is independent of the sample.
    While the projection is dependent on the data, the fact that it was computed in a differentially private manner allows us to act as though it is independent.

    This implies that all points that belong to
    cluster $1$ fall inside the ball
    $\ball{\Pi\mu_1}{4\sqrt{k\sigma_{\max}^2\log(n)}}$,
    whereas any point from a different cluster
    must fall outside the ball
    $\ball{\Pi\mu_1}{90\sqrt{k\sigma_{\max}^2\log(n)}}$,
    with each side containing at least
    $\tfrac {n\mw}{2}$ datapoints.
\end{proof}

\begin{proof}[Proof of Claim~\ref{clm:ProjectedDataTerrificBallLaminar}.]
    Again, as in the proof of Corollary~\ref{cor:CloseProjectedCenter}, we leverage the fact that using
    the projection, we have that
    $\|\Pi (x-\mu_i)  \|_2\leq \sqrt {k \sigma_i^2 \log(n)}$
    for each datapoint $x$ belonging to
    cluster $i$. Note that we run the {\tt PTerrificBall}
    procedure using the flag $largest = {\tt TRUE}$.
    As a result, based on
    Claim~\ref{clm:TerrificRadiusCorrectness_necessary}
    and Corollary~\ref{cor:CloseProjectedCenter},
    we are guaranteed that
    the radius of the ball returned is at least as large
    as the radius of the terrific ball that holds all
    points from the cluster having the largest directional
    variance $\maxsigma^2$ (without loss of generality,
    let that be cluster $1$). In other words, we have
    that the radius of the terrific ball is at least
    $4\sqrt{k\sigma_{\max}^2\log(n)}$.
    
    Therefore, for any cluster $i$,
    the radius is large enough so that the ball of
    radius $4r$ holds all datapoints from cluster $i$.
    Again, the annulus $\ball{p}{5r}\setminus\ball{p}{r}$
    holds very few points, and so almost all
    datapoints of cluster $i$ lie inside $\ball{p}{r}$,
    and none of its points lie outside of $\ball{p}{5r}$.

    Again, the failure probability of this claim is $\beta$
    using Lemma~\ref{lem:terrificball}.
\end{proof}

\begin{proof}[Proof of Claim~\ref{clm:OneComponentNoTerrificBall}.]
    At a high level, the proof goes as follows.
    Suppose our dataset was generated according to a single Gaussian, and that Step 3 or 6 produces a terrific ball.
    Perform the following thought experiment: take the data, apply some univariate projection, and partition the line into 5 intervals (corresponding to the ball of radius $r$/diameter $2r$, the two surrounding intervals of diameter $4r$, and the two infinite intervals which extend past those points).
    For the ball to be terrific, there must be a projection such that the first interval has significantly many points, the second and third intervals have almost no points, and the last two intervals (collectively) have many points. 
    Given the structure of a Gaussian, we know that no such projection and decomposition could exist when considering the \emph{true} probability measure assigned by the Gaussian distribution, since the distribution is unimodal with quickly decaying tails.
    However, with some small probability, this could happen with respect to the empirical set of samples.
    For a given projection, we bound this probability by partitioning the line into intervals of appropriate width, and then applying a Chernoff and union bound style argument.
    We bound the overall probability by taking a union bound over a net of possible directions.
    A more formal argument follows.

    First, due to the proof of Claim~\ref{clm:TerrificBallIsLaminar}
    we know that if a terrific ball $\ball p r$ is
    found in Step 3 
    %\bigvnote{Step 3?} 
    then its radius is large enough
    to hold %\bigvnote{almost?}\bigonote{Here I mean all.} 
    all datapoints from the same cluster.
    Therefore, in the case where all datapoints are
    taken from the same Gaussian we have that no points
    lie outside of $\ball p {5r}$.
    
    Next, we argue something slightly stronger. Note
    that a $4$-terrific ball, either on the data on
    over the projected data, implies that there's some
    direction (unit-length vector) $v$~--- namely,
    a line going from the origin %\bigonote{Clearer?} 
    through the center of the ball~---
    such that on projecting the data onto $v$ we have an
    interval of length $2r$ 
    %\bigvnote{2r?} 
    holding at least $\tfrac{n\mw}{3}$
    datapoints, surrounded by intervals of length
    $3r$ of both sides with very few points (quantified by the guarantees of Lemma~\ref{lem:terrificball}),
    and the remainder of the line also holds $\tfrac{n\mw}{3}$
    datapoints.
    %\bigvnote{How do we conclude this? Is
    %there a lemma that we are referencing?}\bigonote{Just simple geometry, project all points to ${\rm span}(v)$. If it is simpler for you, draw a 2D ball with many point inside, empty annulus and many points outside the annulus...}. 
    (The same holds for a ball for the
    projected points since this ball lies in some
    $k$-dimensional subspace.) However, since the
    datapoints are all drawn from the same Gaussian,
    its projection over $v$ also yields a (one-dimensional)
    Gaussian, with the property that for any three
    ordered intervals from left to right of the same
    length $I_1, I_2, I_3$, the probability mass held
    in $I_2$ is greater than the minimum between the probability held in $I_1$ and the probability mass held in $I_3$.
    %mass held in either $I_1$ or at $I_3$ \bigvnote{Not
    %true when $I_2$ does not include the mean?}\bigonote{Yes true. Read the statement carefully again: $\Pr[I_2] \geq \min\{ \Pr[I_1], \Pr[I_3]$. If the mean is smaller then the leftmost point of $I_2$ then $I_3$ has the smallest probability mass, symmetrically for $I_2$ if the mean is greater then the rightmost point of $I_2$.}. 
    We thus
    have that a necessary condition for the existence
    of such terrific ball is that there exists a
    direction $v$ and an interval $\ell$ which should
    have a probability mass of at least $w_{\min}/10$
    %\bigvnote{umm, how do we say something about the
    %probability mass?}\bigonote{For the actual Gaussian, no such $I_1, I_2, I_3$ exist. To have a terrific ball, such $I_1, I_2, I_3$ must exist on the \emph{empirical draws from the distribution}, thus $I_2$ is ill-represented.},
    yet contains less than $w_{\min}/20$ fraction of
    the empirical mass. 
    %\bigvnote{This could use some
    %more explanation to translate from $\ball{p}{r}$,$\ball{p}{5r}$
    %notations to the interval notation.}\bigonote{I don't necessarily disagree, but I leave it to you...}
	
	We now apply classical reasoning: a Chernoff and
    union argument. Fix a cluster $i$. If we
    are dealing with a set of datapoints drawn only
    from the $i^{\text{th}}$ Gaussian, then this set has no
    more than $2w_i n$ points, and so the ratio
    $\lambda = \tfrac{\mw}{20w_i}$ represents
    the fraction of points that ought to fall inside
    the above-mentioned interval. We thus partition
    the distribution projected onto direction $v$ into
    $2/\lambda$ intervals each holding $\lambda/2$
    probability mass. The Chernoff bound guarantees
    that if the number of points from cluster $i$ is
    at least $\tfrac {200d\log(1/\beta\lambda)}\lambda$
    then a given interval has empirical probability
    sufficiently close (up to a multiplicative factor
    of $2$) to its actual probability mass.
    Applying the union-bound over all $\tfrac 2\lambda$
    intervals times all $2^{O(d)}$ unit-length vectors
    in a cover of the unit-sphere, we get that such an
    interval exists with probability $\leq \beta$.
    %\bigvnote{Can we make these calculations more explicit?
    %This intuitively sort of makes sense, but wouldn't
    %hurt to have a more clear picture since this is the
    %full-version.}\bigonote{Again, go ahead and do so.}
    Note that the number of points of cluster $i$ is
    at least $\tfrac{nw_i}{2}$, so by re-plugging the value of
    $\lambda$ into the bound we get that it is enough
    to have
    $\tfrac{nw_i}{2} \geq \tfrac {200w_id\log(1/\beta w_{\min})}{w_{\min}}$,
    implying that
    $n\geq \tfrac{400 d \log(1/\beta w_{\min})}{w_{\min}}$
    is a sufficient condition to have that no such
    ill-represented interval exists, and as a result,
    no terrific ball exists.
    %\bigvnote{We need to take into account for $X\Pi$
    %that the number oif points from each Gaussian in
    %the projected data could be less than $nw_i/2$
    %because points get lost while finding the terrific
    %ball in $X$.}\bigonote{I don't think so. The properties 1-5 -- namely, 5 -- and the fraction of points omitted should suffice to have more than $w_{\min} n /2$...}
\end{proof}

\subsection{Estimation}

We want to show that once we have clusters
$C_1,\dots,C_k$ from $\RPGMP$,
such that each cluster contains points from
exactly one Gaussian, no two clusters contain
points from the same Gaussian, and that the
fraction of points lost from the dataset is
at most
$$\tau = O\left(\frac{\mw\alpha}{\log(1/\alpha)}\right),$$
we can learn individual Gaussians and mixing
weights accurately. We use the learner for
$d$-dimensional Gaussians from \cite{KamathLSU19}
in the process.

\begin{theorem}[Learner from \cite{KamathLSU19}]
        \label{thm:klsu}
    For every $\alpha,\beta,\eps,\delta,\minsigma,\maxsigma,R > 0$,
    there exists an $(\eps,\delta)$-differentially
    private algorithm $\KLSU$, which if given $m$
    samples from a $d$-dimensional Gaussian $\cN(\mu,\Sigma)$,
    such that,
    $m \geq m_1 + m_2,$
    where,
    \begin{align*}
        m_1 &\geq O\left(\frac{d^2 + \log(\frac{1}{\beta})}{\alpha^2}
            + \frac{d^2\polylog(\frac{d}{\alpha\beta\eps\delta})}{\alpha\eps}
            + \frac{d^\frac{3}{2}\log^{\frac{1}{2}}(\frac{\maxsigma}{\minsigma})
            \polylog(\frac{d\log(\maxsigma/\minsigma)}{\beta\eps\delta})}{\eps}\right)\\
        m_2 &\geq O\left(\frac{d\log(\frac{d}{\beta})}{\alpha^2}
            + \frac{d\log(\frac{d\log{R}\log(1/\delta)}{\alpha\beta\eps})
            \log^{\frac{1}{2}}(\frac{1}{\delta})}{\alpha\eps}
            + \frac{\sqrt{d}\log(\frac{Rd}{\beta})\log^{\frac{1}{2}}(\frac{1}{\delta})}
            {\eps}\right)
    \end{align*}
    and $\minsigma^2 \preceq \Sigma \preceq \maxsigma^2$
    and $\llnorm{\mu} \leq R$, outputs $\wh{\mu},\wh{\Sigma}$,
    such that with probability at least $1-\beta$,
    $$\SD(\cN(\mu,\Sigma),\cN(\wh{\mu},\wh{\Sigma})) \leq \alpha.$$
\end{theorem}

Now, we are in a situation where at most $n\tau$
samples get lost from each component in the
clustering process. So, we need a more robust
version of $\KLSU$ that works even when we lose
a small fraction of the points. The following
theorem guarantees the existence of one such
robust learner.

\begin{theorem}\label{thm:pge}
    For every $\alpha,\beta,\eps,\delta,\minsigma,\maxsigma,R > 0$,
    there exists an $(\eps,\delta)$-differentially
    private algorithm $\PGE$ with the following guarantee.  
    Let $(X_1,\dots,X_n)$ be independent samples from a
    $d$-dimensional Gaussian $\cN(\mu,\Sigma)$,
    where $\minsigma^2 \preceq \Sigma \preceq \maxsigma^2$
    and $\llnorm{\mu} \leq R$ and 
    $n \geq \frac{1}{\mw}(n_1 + n_2),$
    for
    \begin{align*}
        n_1 &\geq O\left(\frac{(d^2 + \log(\frac{1}{\beta}))
            \log^2(1/\alpha)}{\alpha^2}
            + \frac{(d^2\polylog(\frac{d}{\alpha\beta\eps\delta}))}{\alpha\eps}
            + \frac{d^\frac{3}{2}\log^{\frac{1}{2}}(\frac{\maxsigma}{\minsigma})
            \polylog(\frac{d\log(\maxsigma/\minsigma)}{\beta\eps\delta})}{\eps}\right)\\
        n_2 &\geq O\left(\frac{d\log(\frac{d}{\beta})\log^2(1/\alpha)}{\alpha^2}
            + \frac{d\log(\frac{d\log{R}\log(1/\delta)}{\beta\eps})
            \log^{\frac{1}{2}}(\frac{1}{\delta})\log^2(1/\alpha)}{\alpha\eps}
            + \frac{\sqrt{d}\log(\frac{Rd}{\beta})\log^{\frac{1}{2}}(\frac{1}{\delta})}
            {\eps}\right).
    \end{align*}
    For every $S \subseteq [n]$ with $|S| \geq n(1- O(\alpha/\log(1/\alpha)))$,
    when $\PGE$ is given $\{X_i\}_{i \in S}$ as input, with probability at least $1-\beta$,
    $$\SD(\cN(\mu,\Sigma),\cN(\wh{\mu},\wh{\Sigma})) \leq \alpha.$$
\end{theorem}

The proof of this theorem follows effectively the
same structure as that of Theorem~\ref{thm:klsu},
the primary difference in the setting being that a
miniscule fraction of points have been removed.
However, fortunately, the proof of Theorem~\ref{thm:klsu}
uses the Gaussianity of the data essentially only to
show concentration of the empirical mean and covariance
in various norms, which are preserved even against an
adversary who can delete a small fraction of the points
(see, i.e., Lemma 4.3, 4.4, and Corollary 4.8 of~\cite{DiakonikolasKKLMS16}).
Substituting these statements into the proof, it follows
essentially the same.

Now, we give a lemma that says that the
mixing weights are accurately estimated.

\begin{lemma}
    Suppose $w_1,\dots,w_k$ are the mixing
    weights of the Gaussians of the target
    distribution $\cD \in \cG(d,k)$.
    Let $\wh{w}_1,\dots,\wh{w}_k$ be their
    respective estimations produced by the
    algorithm. If
    $$n \geq O\left(\frac{k^2}{\eps\alpha}\ln(k/\beta)\right),$$
    then with probability at least $1-O(\beta)$,
    $$\forall~ 1 \leq i \leq k,~~~
        \abs{\wh{w}_i - w_i} \leq \frac{\alpha}{3k}.$$
\end{lemma}
\begin{proof}
    Let $w'_1,\dots,w'_k$ be the empirical weights
    of the Gaussians produced using the points
    in $X$. We have the following claim for
    them.
    \begin{claim}
        Let $X$ be the dataset as in the
        algorithm. If
        $$n \geq O\left(\frac{k^2}{\eps\alpha}\ln(k/\beta)\right),$$
        and $X$ satisfies Condition \ref{cond:num-points},
        then for $w_i \geq \tfrac{4\alpha}{9k}$,
        $$\abs{w_i-w'_i} \leq \frac{\alpha}{9k},$$
        and for $w_i < \tfrac{4\alpha}{9k}$,
        $$\abs{w_i-w'_i} \leq \frac{2\alpha}{9k},$$
    \end{claim}
    \begin{proof}
        There are two sources of error in this case:
        (1) adding Laplace noise; and (2) by having
        lost $\tau$ points from $X$. Let $\wt{n}_i$
        be as defined in the algorithm, $\wb{n}_i = \abs{C_i}$,
        and $n_i$ be the number of points of this
        component in $X$.

        First, we want to show that the if the number
        of points is large enough, then the added noise
        does not perturb the empirical weights too much.
        Now, using Lemma \ref{lem:lap-conc} with our bound
        on $n$, and applying the union bound over all calls
        to $\PCount$, we get that with probability at
        least $1-O(\beta)$,
        $$\forall~ i,~~~ \abs{\wt{n}_i - \wb{n}_i}
            \leq \frac{n\alpha}{40k^2}.$$

        Also, we know that $\sum\limits_{i=1}^{k}{\wb{n}_i}  \geq n(1-\tau)$.

        Now, let $e = \tfrac{\alpha}{40k^2}$. From the
        above, using triangle inequality, we get that for
        all $i$,
        \begin{align*}
            \abs{\wh{w}_i-w'_i} &=
                    \abs{\frac{\wt{n}_i}{\sum\limits_{j=1}^{k}{\wt{n}_j}} - \frac{n_i}{n}}\\
                &\leq \abs{\frac{\wt{n}_i}{\sum\limits_{i=1}^{k}{\wt{n}_i}}
                    - \frac{\wb{n}_i}{\sum\limits_{i=1}^{k}{\wb{n}_i}}}
                    + \abs{\frac{\wb{n}_i}{\sum\limits_{i=1}^{k}{\wb{n}_i}}
                    - \frac{n_i}{n}}\\
                &\leq \frac{e+ke}{1-ke} + \frac{n_i}{n}\abs{\frac{1}{1-\tau} - 1}\\
%                &= \frac{e+ke}{1-ke} + \frac{n_i}{n}\frac{\tau}{1-\tau}\\
                &\leq \frac{\alpha}{18k} + \frac{\alpha}{18k}\\
                &\leq \frac{\alpha}{9k},
        \end{align*}
        Where the second to last inequality holds
        because $\tau \leq \tfrac{\alpha}{20k}$.
    \end{proof}

    Because $\abs{w_i-w'_i} \leq \tfrac{\alpha}{9k}$,
    using triangle inequality, we get the required result.
\end{proof}

Combining these statements is sufficient to conclude
Theorem~\ref{thm:pgme-accurate}.

%% file: TerrificBall.tex
\subsection{Finding a Secluded Ball}
\label{sec:TerrificBall}

In this section we detail a key building
block in our algorithm for learning Gaussian
mixtures. This particular subroutine is an
adaptation of the work of Nissim and
Stemmer~\cite{NissimS18} (who in turn built
on~\cite{NissimSV16}) that finds a ball that
contains many datapoints. In this section we
show how to tweak their algorithm so that it
now produces a ball with a few more additional
properties. More specifically, our goal in this
section is to privately locate a ball $\ball{p}{r}$
that (i) contains many datapoints, (ii) leaves out
many datapoints (i.e., its complement also contains
many points) and (iii) is secluded in the sense
that $\ball{p}{cr}\setminus \ball p r$ holds very
few (and ideally zero) datapoints for some constant
$c > 1$. More specifically, we are using the following
definition.

\begin{defn}[Terrific Ball]\label{def:TerrificBall}
    Given a dataset $X$ and an integer $t >0$, we
    say a ball $\ball{p}{r}$ is $(c,\Gamma)$-terrific
    for a constant $c>1$ and parameter $\Gamma\geq 0$
    if all of the following three properties hold: (i)
    The number of datapoints in $\ball{p}{r}$ is at
    least $t-\Gamma$; (ii) The number of datapoints
    outside the ball $\ball{p}{cr}$ is least $t-\Gamma$;
    and (iii) The number of datapoints in the annulus
    $\ball{p}{cr}\setminus\ball{p}{r}$ is at most $\Gamma$.
\end{defn}
We say a ball is $c$-terrific if it is $(c,0)$-terrific,
and when $c$ is clear from context we call a ball
terrific. In this section, we provide a differentially
private algorithm that locates a terrific ball. 

\paragraph{The 1-Cluster Algorithm of Nissim-Stemmer.}
First, we give an overview of the Nissim-Stemmer algorithm
for locating a ball $\ball p r$ that satisfy only property
(i) above, namely a ball that holds at least $t$ points
from our dataset containing $n$ points.
Their algorithm is actually composed of two subroutines that
are run sequentially. The first, {\tt GoodRadius} privately
computes some radius $\tilde r$ such that
$\tilde r \leq 4r_{\rm opt}$ with $r_{\rm opt}$ denoting the
radius of the smallest ball that contains $t$ datapoints.
Their next subroutine, {\tt GoodCenter}, takes $\tilde r$ as
an input and produces a ball of radius $\gamma\tilde r$ that
holds (roughly) $t$ datapoints with $\gamma$ denoting some
constant $>2$. The {\tt GoodCenter} procedure works by first
cleverly combining locality-sensitive hashing (LSH) and
randomly-chosen axes-aligned boxes to retrieve a $\poly(n)$-length
list of candidate centers, then applying $\AboveThreshold$ to
find a center point $p$ such that the ball $\ball {p}{\tilde r}$
satisfies the required condition (holding enough datapoints).

In this section, we detail how to revise the two above-mentioned subprocedures in order to retrieve a terrific ball, satisfying all three properties (rather than merely holding many points). The revision isn't difficult, and it is particularly straight-forward for the {\tt GoodCenter} procedure. In fact, we keep {\tt GoodCenter} as is, except for the minor modification of testing for all $3$ properties together. Namely, we replace the na\"ive counting query in $\AboveThreshold$ (i.e., ``is $|\{x\in X:~ x\in\ball{p}{r}\}| \geq t$?'') with the query
\begin{equation}
\label{eq:TerrificBallQuery}
\min\big\{  ~|\{x\in X: x\in \ball{p}{r}\}|, ~~|\{x\in X:~ x\notin \ball{p}{cr}\}|,~~ t - |\{x\in X: x\in \ball{p}{cr}\setminus \ball{p}{r}\}|  \big  \} \stackrel{\rm ?}\geq t
\end{equation}
Note that this query is the minimum of $3$ separate counting queries, and therefore its global sensitivity is $1$. Thus, our modification focuses on altering the first part where we find a good radius, replacing the {\tt GoodRadius} procedure with the {\tt TerrificRadius} procedure detailed below. Once a terrific radius is found, we apply the revised {\tt GoodCenter} procedure and retrieve a center. 

\begin{remark}
    In the work of Nissim-Stemmer~\cite{NissimS18}
    {\tt GoodCenter}, the resulting ball has radius
    $\leq \gamma\tilde r$ since the last step of the
    algorithm is to average over all points in a
    certain set of radius $\leq \gamma\tilde r$. 
    In our setting however it holds that $\tilde r$
    is a radius of a terrific ball, and in particular,
    in a setting where $\gamma < c$ (which is the
    application we use, where $\gamma\approx2.5$ whereas
    $c\geq 4$), this averaging is such that effectively
    all points come from $\ball{\tilde r}{p}$, and so
    the returned ball is of radius $\approx \tilde r$. 
\end{remark}

\paragraph{The TerrificRadius procedure.} Much like the
work of~\cite{NissimSV16}, we too define a set of possible
$r$'s, traverse each possible value of $r$ and associate
a score function that measures its ability to be the
sought-after radius. However, we alter the algorithm is
two significant ways. The first is modifying the score
function to account for all three properties of a terrific
ball, and not just the one about containing many points;
the second is that we do not apply the recursive
algorithm to traverse the set of possible $r$'s, but rather
try each one of these values ourselves using $\AboveThreshold$.
The reason for the latter modification stems from the fact
that our terrific radius is no longer upward closed
(i.e. it is not true that if $r$ is a radius of a ball
satisfying the above three properties then any $r'>r$ is
also a radius of a ball satisfying these properties) and
as a result our scoring function is \emph{not} quasi-concave.

The modification to the scoring function is very similar
to the one detailed in Equation~\eqref{eq:TerrificBallQuery},
with the exception that rather than counting exact sizes, we
cap the count at $t$. Formally, given a (finite) set $S$ we
denote $\#^t S \stackrel{\rm def}= \min\{ |S|, t\}$, and so
we define the counting query 
\begin{align}
    Q_X(p,r) \stackrel{\rm def}{=}
            \min\big\{ &\#^t \{x\in X:~ x\in \ball{p}{r}\},~~~
            \#^t\{x\in X:~ x\notin \ball{p}{cr} \},\nonumber\\
        &t - \#^t\{x\in X:~ x\in \ball{p}{cr}\setminus \ball{p}{r} \} \big\}
\end{align}
It is evident that for any dataset $X$, $p$ and $r>0$ it holds that $Q_X(p,r)$ is an integer between $0$ and $t$. Moreover, for any $p$ and $r$ and any two neighboring datasets $X$ and $X'$ it holds that $|Q_X(p,r)-Q_{X'}(p,r)|\leq 1$. We now have all ingredients for introducing the Terrific Radius subprocedure below.

\begin{algorithm}[h!] \label{alg:TerrificRadius}
	\caption{Find Terrific Radius
		$(X, t, c, largest; \eps, U, L)$}
	\KwIn{Dataset $X \in \R^{n\times d}$.
		Privacy parameters $\eps, \delta>0$;
        failure probability $\beta > 0$; candidate size $t>0$.
        Upper- and lower-bounds on the radius $L<U$. Boolean flag $largest$.}
	\KwOut{A radius of a terrific ball.} 
	\vspace{10pt}
	
	Denote $L_X(r) = \frac 1 t
        \max\limits_{\textrm{distinct }x_1,x_2,...,x_t\in X}\sum\limits_{j=1}^t Q_X(x_{j},r)$.\
	
    Denote $T = \lceil\log_2(U/L)\rceil+1$ and $\Gamma \stackrel{\rm def}=
        \frac {16}{\epsilon} \left(\log(T) + \log(2/\beta)\right)$.
	
	Set $r_0 = L, r_1 = 2L, r_2 = 4L, ..., r_i = 2^iL,..., r_T = U$.

	\If {$largest =${\tt TRUE}} {reverse the order of $r_i$'s.}
    \vspace{10pt}
    
    Run $\AboveThreshold$ on queries $L_X(r_i)$ for all
    $0 \leq i \leq T$, with
    threshold $t-\Gamma$ and sensitivity $2$,
    and output the $r_i$ for which $\AboveThreshold$
    returns $\top$.\\	
	\vspace{10pt}
\end{algorithm}

\begin{claim}\label{clm:TerrificRadiusPrivate}
    Algorithm~\ref{alg:TerrificRadius} satisfies
    $(\eps,0)$-Differential Privacy.
\end{claim}
\begin{proof}
    Much like~\cite{NissimSV16}, we argue that
    $L_X(r)$ has $\ell_1$-sensitivity at most $2$.
    With this sensitivity bound, the algorithm is
    just an instantiation of $\AboveThreshold$
    (Theorem~\ref{thm:above-threshold}), so it is
    $(\eps,0)$-Differentially Private.

    Fix two neighboring datasets $X$ and $X'$. Fix $r$.
    Let $x_1^1,...,x_t^1\in X$ be the $t$ points on which
    $L_X(r)$ is obtained and let $x_1^2,...,x_t^2\in X'$
    be the $t$ points on which $L_{X'}(r)$ is obtained.
    Since $X$ and $X'$ differ on at most $1$ datapoint,
    then at most one point from each $t$-tuple can be a
    point that doesn't belong to $X\cap X'$. Without loss
    of generality, if such a point exists then it is $x_1^1$
    and $x_1^2$. For all other points it follows that
    $|Q_X(x,r)-Q_{X'}(x,r)|\leq 1$ as mentioned above.
    Thus we have
    \begin{align*}
        L_X(r) &= \frac 1 t \sum_{i=1}^t Q_X(x_i^1,r) =
            \frac 1 t Q_X(x_1^1,r) + \frac 1 t \sum_{i=2}^t Q_X(x_i^1,r)
            \cr & \leq \frac t t +  \frac 1 t
            \sum_{i=2}^t \left(Q_{X'}(x_i^1,r) + 1\right)
            \leq 1 + \left(\frac 1 t Q_{X'}(x_1^2,r) +
            \frac 1 t \sum_{i=2}^t Q_{X'}(x_i^1,r)\right) + 1 \leq 2 + L_{X'}(r)
    \end{align*}
    based on the fact that $0\leq Q(x,r)\leq t$ for any
    $x$ and any $r$ and the definition of $L_X$ as a
    $\max$-operation. The inequality $L_{X'}(r)\leq L_X(r)+2$
    is symmetric.
\end{proof}

\begin{claim}
	\label{clm:TerrificRadiusCorrectness_sufficient}
	Fix $\beta>0$ and denote $\Gamma$ as in Algorithm~\ref{alg:TerrificRadius}. With probability $\geq 1 - \beta$, if Algorithm~\ref{alg:TerrificRadius} returns a radius $r$ (and not $\bot$), then there exists a ball $\ball p r$ which is $(c,2\Gamma)$-terrific.
\end{claim}
\begin{proof}
    Let $T+1 = (\lceil\log_2(U/L)\rceil+1)$ be the number
    of queries posed to $\AboveThreshold$. From
    Theorem~\ref{thm:above-threshold}, we know that with
    probability at least $1-\beta$, if radius $r$ is returned,
    then $L_X(r) \geq t-\Gamma-\Gamma=t-2\Gamma$, where
    $\Gamma = \tfrac{16}{\eps}\left(\log(T)+\log(2/\beta)\right)$.
    It follows that one of the $t$ distinct datapoints on
    which $L_X(r)$ is obtained must satisfy $Q(x_i,r) \geq t-2\Gamma$.
    Thus, the ball $\ball {x_i} r$ is the sought-after ball.
\end{proof}

Next we argue that if the data is such that
it has a terrific ball, our algorithm indeed
returns a ball of comparable radius. Note
however that the data could have multiple
terrific balls. Furthermore, given a $c$-terrific
ball, there could be multiple different balls
of different radii that yield the same partition.
Therefore, given dataset $X$ which has a
$c$-terrific ball $\ball{p}{r}$, denote
$A = X\cap \ball{p}{r}$ and define $r_A$ as
the \emph{minimal terrific radius} of a
$c$-terrific ball forming $A$, i.e., for any
other $c$-terrific $\ball{p'}{r'}$ such that
$A = X\cap \ball{p'}{r'}$, we have that
$r'\geq r_A$. Let ${\cal R}_c$ be the set of
all minimal terrific radii of $c$-terrific balls
forming subsets of $X$.

\begin{claim}
	\label{clm:TerrificRadiusCorrectness_necessary}
    Fix $c>1$ and assume we apply Algorithm~\ref{alg:TerrificRadius}
    with this $c$ as a parameter.
    Given $X$ such that ${\cal R}_{4c+1}$
    is not empty, then with probability $\geq 1-\beta$
    it holds that Algorithm~\ref{alg:TerrificRadius}
    returns $r$ such that if $largest={\tt TRUE}$ then
    $r\geq \max\{r'\in {\cal R}_{4c+1}\}$ and if
    $largest={\tt FALSE}$ then
    $r \leq 4\cdot \min\{r'\in {\cal R}_{4c+1}\}$.
\end{claim}
\begin{proof}
    Given $T+1 = (\lceil\log_c(U/L)\rceil+1)$
    queries and threshold $t-\Gamma$
    (using the same definition of $\Gamma$ as in
    Algorithm~\ref{alg:TerrificRadius}), from
    Theorem~\ref{thm:above-threshold},
	it holds with probability $\geq 1-\beta$ that
    $\AboveThreshold$ must halt by the time it considers
    the very first $r$ for which $L_X(r)\geq t$. We show
    that for any $r\in \cR_{4c+1}$ there exists a query
    among the queries posed to $\AboveThreshold$ over
    some $r'$ such that (a) $r \leq r' \leq 4 r$  and
    (b) $L_X(r')=t$. Since the order in which the
    $\AboveThreshold$ mechanism iterates through the
    queries is determined by the boolean flag $largest$
    the required follows.

    Fix $r\in \cR_{4c+1}$. Let $\ball p r$  be a
    terrific ball of radius $r$ that holds at
    least $t$ points, and let $x_1,...,x_{t'}$
    denote the set of $t'\geq t$ distinct datapoints
    in $\ball p r$. Denote $D = \max_{i\neq j}\|x_i-x_j\|$.
    Clearly, $r\leq D\leq 2r$, the lower-bound follows
    from the minimality of $r$ as a radius that separates
    these datapoints from the rest of $X$ and the upper-bound
    is a straight-forward application of the triangle
    inequality. Next, let $r^*$ denote the radius in the
    range $[D, 2D] \subset [r,4r]$ which is posed as a
    query to $\AboveThreshold$. We have that $Q(x_j,r^*)=t$
    for each $x_j$. Indeed, the ball $\ball{x_j}{r^*}$ holds
    $t'\geq t$ datapoints. Moreover,
    $\ball{x_j}{cr^*} \subset \ball{p}{r+cr^*} \subset \ball{p}{(4c+1)r}$
    and therefore
    $\left( \ball{x_j}{2r^*}\setminus \ball{x_j}{r^*}\right) \subset
    \left( \ball{p}{(4c+1)r}\setminus\{x_1,...,x_{t'}\}\right)$
    and thus it is empty too. Therefore, any datapoint outside
    of $\ball{p}{r}$ which must also be outside $\ball{p}{(4c+1)r}$
    is contained in the compliment of $\ball{x_j}{r^*}$,
    and so the compliment also contains $t$ points.
    As this holds for any $x_j$, it follows that $L_X(r^*)=t$
    and thus the required is proven.	
\end{proof}

\paragraph{The entire {\tt PTerrificBall} procedure.} The overall algorithm that tries to find a $c$-terrific ball is the result of running {\tt TerrificRadius} followed by the {\tt GoodCenter} modified as discussed above: we conclude by running $\AboveThreshold$ with the query given by~\eqref{eq:TerrificBallQuery} for a ball of radius $(1+\frac c {10})\tilde r$. Its guarantee is as follows.

\begin{lemma}\label{lem:terrificball}
    The {\tt PTerrificBall} procedure is a
    $(2\epsilon,\delta)$-DP algorithm which, if run
    using size-parameter
    $t \geq \frac{1000c^2}\epsilon n^a\sqrt d
    \log(nd/\beta)\log(1/\delta)\log\log(U/L)$ for
    some arbitrary small constant $a>0$ (say $a=0.1$),
    and is set to find a $c$-terrific ball with
    $c>\gamma$ ($\gamma$ being the parameter fed into
    the LSH in the {\tt GoodCenter} procedure), then
    the following holds. With probability at least
    $1-2\beta$, if it returns a ball $\ball{r}{p}$, then
    this ball is a $(c,2\Gamma)$-terrific ball of radius
    $r \leq (1+\frac c {10})\tilde{r}$,
    where $\tilde{r}$ denotes the radius obtained from
    its call to {\tt TerrificRadius}, and $\Gamma$ is $\frac {16}{\epsilon} \left(\log( \lceil\log_2(U/L)\rceil+1) + \log(2/\beta)\right)$.
\end{lemma}
\begin{proof}
    By Claim~\ref{clm:TerrificRadiusCorrectness_sufficient},
    it follows that $\tilde r$ is a radius for some
    $(c,2\Gamma)$-terrific ball. The analysis in~\cite{NissimS18}
    asserts that with probability at least $1-\beta$,
    {\tt GoodCenter} locates a $(n^{-a}/2)$-fraction of
    the points inside the ball, and uses their average. 
    Note that $t$ is set such that $t\cdot n^{-a}/2 > 10c^2\cdot 2\Gamma$. 
    Fix $x$ to be any of the $\geq t-2\Gamma$ points
    inside the ball. Due to the quality function we use
    in {\tt TerrificRadius}, at most $2\Gamma$ of the
    points, which got the same hash value as $x$, are within
    distance $\gamma \tilde r < c\tilde r$, and the remaining
    $t-2\Gamma$ are within distance $\tilde r$ from $x$.
    It follows that their average is within distance at most
    $\tilde r + \frac{c\tilde r}{10c^2} \leq \tilde r(1+\frac c {10})$.
    The rest of the~\cite{NissimS18} proof follows as
    in the {\tt GoodCenter} case.
\end{proof}

%% file: sa.tex
\section{Sample and Aggregate}
\label{sec:sa}
In this section, we detail methods based on sample and aggregate, and derive their sample complexity.
This will serve as a baseline for comparison with our methods.

A similar sample and aggregate method was considered in~\cite{NissimRS07}, but they focused on a restricted case (when all mixing weights are equal, and all components are spherical with a known variance), and did not explore certain considerations (i.e., how to minimize the impact of a large domain).
We provide a more in-depth exploration and attempt to optimize the sample complexity.

%\bigvnote{In this paragraph, we cite papers like [NSV16]
%and [NS18] without describing what the aggregate approach
%looks like. It might be good to say what it means to ``aggregate''
%in this setting.}
%\bigvnote{Also, it is not clear why we use these aggregation
%methods. One might ask if there is no other known way of aggregating
%stuff, even in the spherical case.}
%\bigvnote{That said, we do it in the proof anyway.
%So, I'm not sure if it's worth doing that informally
%here as well.}
The main advantage of the sample and aggregate method we describe here is that it is extremely flexible: given any non-private algorithm for learning mixtures of Gaussians, it can immediately be converted to a private method.
However, there are a few drawbacks, which our main algorithm avoids.
First, by the nature of the approach, it will increase the sample complexity multiplicatively by $\Omega(\sqrt{d}/\eps)$, thus losing any chance of the non-private sample complexity being the dominating term in any parameter regime.
Second, it is not clear on how to adapt this method to non-spherical Gaussians.
We rely on the methods of~\cite{NissimSV16,NissimS18}, which find a small $\ell_2$-ball which contains many points.  
The drawback of these methods is that they depend on the $\ell_2$-metric, rather than the (unknown) Mahalanobis metric as required by non-spherical Gaussians.
We consider aggregation methods which can handle settings where the required metric is unknown to be a very interesting direction for further study.

Our main sample-and-aggregate meta-theorem is the following.

\begin{theorem}
  \label{thm:sa}
  Let $m = \tilde \Theta\left(\frac{\sqrt{kd} + k^{1.5}}{\eps} \log^{2}(1/\delta) \cdot 2^{O\left(\log^*\left(\frac{dR\sigma_{\max}}{\alpha\sigma_{\min}}\right)\right)}\right)$.
  Suppose there exists a (non-private) algorithm with the following guarantees.
  The algorithm is given a set of samples $X_1, \dots, X_n$ generated i.i.d.\ from some mixture of $k$ Gaussians $\cD = \sum_{i=1}^k w_i \cN(\mu_i, \sigma_i^2\id_{d \times d} )$, with the separation condition that $\|\mu_i - \mu_j\|_2 \geq (\sigma_i + \sigma_j)\tau_{k,d}$, where $\tau_{k,d}$ is some function of $k$ and $d$, and $\tau_{k,d} \geq c\alpha$, for some sufficiently large constant $c$.
  With probability at least $1 - m/100$, it outputs a set of points $\{\hat \mu_1, \dots, \hat \mu_k\}$ and weights $\{\hat w_1, \dots, \hat w_k\}$ such that $\|\hat \mu_{\pi(i)} - \mu_i\|_2 \leq O\left(\frac{\alpha \sigma_i}{\sqrt{\log mk}}\right)$ and $|\hat w_{\pi(i)} - w_i| \leq O(\alpha/k)$ for all $i \in [k]$, where $\pi: [k] \rightarrow [k]$ is some permutation.

  Then there exists a $(\eps,\delta)$-differentially private algorithm which takes $mn$ samples from the same mixture, and input parameters $R, \sigma_{\min}, \sigma_{\max}$ such that $\|\mu_i\|_2 \leq R$ and $\sigma_{\min} \leq \sigma_i \leq \sigma_{\max}$ for all $i \in [k]$.
  With probability at least $9/10$, it outputs a set of points $\{\hat \mu_1, \dots, \hat \mu_k\}$ and weights $\{\hat w_1, \dots, \hat w_k\}$ such that $\|\hat \mu_{\pi(i)} - \mu_i\|_2 \leq O\left(\alpha\sigma_i\right)$ and $|\hat w_{\pi(i)} - w_i| \leq O(\alpha/k)$ for all $i \in [k]$, for some permutation $\pi$.
\end{theorem}
\begin{proof}
  In short, the algorithm will repeat the non-private algorithm several times, and then aggregate the findings using the 1-cluster algorithm from~\cite{NissimSV16}.
  We will focus on how to generate the estimates of the means, and sketch the argument needed to conclude the accuracy guarantees for the mixing weights.

  In more detail, we start by discretizing the space where all the points live, which is a set of diameter $\poly\left(R,\sigma_{\max},d, \log n \right)$, at granularity $\poly\left(\frac{\alpha\sigma_{\min}}{d}\right)$, and every point we consider will first be rounded to the nearest point in this discretization.
  This will allow us to run algorithms which have a dependence on the size of the domain.
  Since this dependence will be proportional to the exponentiation of the iterated logarithm, we can take the granularity to be very fine by increasing the exponent of the polynomial, at a minimal cost in the asymptotic runtime.
  For clarity of presentation, in the sequel we will disregard accounting for error due to discretization.

  Now, we use the following theorem of~\cite{NissimSV16}.

  \begin{theorem}[\cite{NissimSV16}]\label{thm:ploc-log}
    Suppose $X_1,\dots,X_m$ are points from
    $S^d \subset \R^d$, where $S^d$ is finite.
    Let $m,t,\beta,\eps,\delta$ be such that,
    $$t = \Omega\left(\frac{\sqrt{d}}{\eps}
        \log\left(\frac{1}{\beta}\right)
        \log\left(\frac{md}{\beta\delta}\right)
        \sqrt{\log\left(\frac{1}{\beta\delta}\right)}
        \cdot 9^{\log^*(2\abs{S}\sqrt{d})}\right).$$
    Let $r_{opt}$ be the radius of the smallest
    ball that contains at least $t$ points
    from the sample.
    There exists an $(\eps,\delta)$-DP algorithm
    that returns a ball of radius at most
    $w \cdot r_{opt}$ such that it contains at
    least $t-\Delta$ points from the sample with
    error probability $\beta$, where $w = O(\sqrt{\log m})$
    and
    $$\Delta = O\left(\frac{1}{\eps}
        \log\left(\frac{1}{\beta}\right)
        \log\left(\frac{m}{\delta}\right)
        \cdot 9^{\log^*(2\abs{S}\sqrt{d})}\right).$$
\end{theorem}
  Compare with the slightly different guarantees of Theorem~\ref{thm:ploc}, which is the 1-cluster algorithm from~\cite{NissimS18}.
  We will use Theorem~\ref{thm:ploc-log} $k$ times, with the following settings of parameters: their $\varepsilon$ is equal to our $O(\varepsilon/\sqrt{k \log(1/\delta)})$, their $\beta$ is equal to $O(1/k)$, their $|S|$ is our $\poly\left(\frac{dR\sigma_{\max}}{\sigma_{\min}\alpha}\right)$, their $m$ is equal to $mk$, and all other parameters are the same.

  We start by taking the results of running the non-private algorithm $m$ times.
  We will run the algorithm of Theorem~\ref{thm:ploc-log} with $t = m$ to obtain a ball (defined by a center and a radius).
  We remove all points within this ball, and repeat the above process $k$ times.
  Note that by advanced composition, the result will be $(\eps, \delta)$-differentially private.
  We spend the rest of the proof arguing that we satisfy the conditions of Theorem~\ref{thm:ploc-log}, and that if we choose the $i$th mean to be an arbitrary point from the $i$th output ball, these will satisfy the desired guarantees.

  First, we confirm that our choice of $t$ satisfies the conditions of the theorem statement.
  Since we set $t$ to be equal to $m = \tilde \Theta\left(\frac{\sqrt{kd} + k^{1.5}}{\eps} \log^{2}(1/\delta) \cdot 2^{O\left(\log^*\left(\frac{dR\sigma_{\max}}{\alpha\sigma_{\min}}\right)\right)}\right)$, the ``first term'' (the one with the leading $\sqrt{kd}/\eps$) is large enough so that $t$ will satisfy the necessary condition.

    For the rest of the proof, we will reason about
    the state of the points output by the $m$ runs
    of the non-private algorithm. 
    Note that the non-private algorithm will learn to a slightly better accuracy than our final private algorithm ($\alpha/\sqrt{\log{mk}}$, rather than $\alpha$).
    %\bigvnote{Do we want 
    %to say that the accuracy we're looking for here is
    %$\alpha/\sqrt{\log(mk)}$ and not $\alpha$? This would
    %sort of explain why we have a ball of radius
    %$\alpha\sigma_i/\sqrt{\log(mk)}$.} 
    By a union bound,
    we know that all $m$ runs of the non-private algorithm will output mean estimates which are close to the true means with probability at least $99/100$, an event we will condition on for the rest of the proof.
  This implies that, around the mean of component $i$, there is a ball of radius $O\left(\frac{\alpha \sigma_i}{\sqrt{\log mk}}\right)$ which contains a set of $m$ points: we will call each of these point sets a \emph{mean-set}.
  We will say that a mean-set is \emph{unbroken} if all $m$ of its points still remain, i.e., none of them have been removed yet.
    %\bigvnote{In the above sentence, we use the word
    %``dataset''. What dataset is it referring to? So
    %far, we've been using it to refer to our set of
    %samples from the mixture. But it's different here,
    %so might want to be clear about that.}
  %We will prove by induction that, before we run the algorithm of Theorem~\ref{thm:ploc-log} for the $j$th time, there will be $k-j+1$ unbroken mean-sets, and during the $j$th run, we will identify a ball containing points from exactly one unbroken mean-set.

  %Note that before the $j = 1$ run, the correctness of the non-private algorithm implies that we have $k$ unbroken mean-sets, as desired.
  %Suppose that at some iteration $j$, we have $k-j+1$ unbroken mean-sets. 

  We claim that, during every run of the algorithm of
  Theorem~\ref{thm:ploc-log}, we will identify and remove
  points belonging solely to a single unbroken mean-set.
  First, we argue that the smallest ball containing $m$
  points will consist of points solely from a single unbroken
  mean-set. There are two cases which could be to the
  contrary: if it contains points from one unbroken mean-set
  and another mean-set (either broken or unbroken), and
  if it contains points from only broken mean-sets. In
  the first case, the separation condition and triangle
  inequality imply that the ball consisting of points
  solely from the unbroken mean-set within this ball
  would be smaller. %\bigvnote{Are we arguing that by the
  %end of the $i^{\text{th}}$ iteration, this would happen
  %for the second case? If so, the number of points left from
  %broken sets is at most $i\Delta$? Writing $k\Delta$ might
  %be a bit confusing.} 
  The second case is also impossible:
  this is because we require at least $m$ points in the
  ball, and during the $i$th iteration, there will be at most $(i-1)\Delta \leq k\Delta$ points leftover
  from broken mean-sets (assuming that the claim at the
  start of this paragraph holds by strong induction).
  The ``second term'' of $m$ (the one with the leading
  $k^{1.5}/\eps$) enforces that $m > k\Delta$, preventing
  this situation. Arguments similar to those for these two cases
  imply that any ball with radius equal to this minimal
  radius inflated by a factor of $w = O(\sqrt{\log mk})$
  and containing $m - \Delta$ points must consist of points
  belonging solely to a single unbroken mean-set.

  It remains to conclude that any point within a ball has the desired
  accuracy guarantee. 
  Specifically, using any point within a ball as a candidate mean will approximate the true mean of that component up to $O\left(\frac{\alpha \sigma_i}{\sqrt{\log mk}}\right)$.
  This is because the smallest ball containing an unbroken mean-set has radius at most $O\left(\frac{\alpha \sigma_i}{\sqrt{\log mk}}\right)$ (and we know that every point within this ball has the desired accuracy guarantee with respect to the true mean), and the algorithm will inflate this radius by a factor of at most $w = O(\sqrt{\log mk})$.

  At this point, we sketch the straightforward argument to estimate the values of the mixing weights.
  The output of the non-private algorithm consists of pairs of mean and mixing weight estimates.
  By a union bound, all of the (non-private) mixing weights are sufficiently accurate with probability at least $99/100$.
  In order to aggregate these non-private quantities into a single private estimate, we can use a stability-based histogram (see~\cite{KorolovaKMN09, BunNSV15}, and~\cite{Vadhan17} for a clean presentation).
  More precisely, for all the mean estimates contained in each ball, we run a stability-based histogram (with bins of width $O(\alpha/k)$) on the associated mixing weight estimates, and output the identity of the most populated bin.

  We claim that the aggregated mixing weight estimates will all fall into a single bin with a large constant probability, simultaneously for each of the histograms. 
  This is because all the mixing weight estimates are correct with probability $99/100$, and the argument above (i.e., each run of the~\cite{NissimSV16} algorithm removes points belonging solely to a single unbroken mean-set) implies that we will aggregate mixing weight estimates belonging only to a single component.
  This guarantees the accuracy we desire.
  The cost in the sample complexity is dominated by the cost of the $k$ runs of the algorithm of~\cite{NissimSV16}.
\end{proof}

%\bigvnote{Do we want to use $\beta$ instead of
%$\delta$ for error probability in the following?}
The following lemma can be derived from~\cite{VempalaW02}.
The first term of the sample complexity is the complexity of clustering from Theorem 4 of~\cite{VempalaW02}, the second and third terms are for learning the means and mixing weights, respectively.
\begin{lemma}[From Theorem 4 of~\cite{VempalaW02}]
   There exists a (non-private) algorithm with the following guarantees.
  The algorithm is given a set of samples $X_1, \dots, X_n$ generated i.i.d.\ from some mixture of $k$ Gaussians $\cD = \sum_{i=1}^k w_i \cN(\mu_i, \sigma_i^2\id_{d \times d} )$, with the separation condition that $\|\mu_i - \mu_j\|_2 \geq 14 (\sigma_i + \sigma_j)(k \ln (4n/\beta))^{1/4}$.
  With probability at least $1 - \beta$, it outputs a set of points $\{\hat \mu_1, \dots, \hat \mu_k\}$ and weights $\{\hat w_1, \dots, \hat w_k\}$ such that $\|\hat \mu_{\pi(i)} - \mu_i\|_2 \leq O\left(\alpha \sigma_i\right)$ and $|\hat w_{\pi(i)} - w_i| \leq O(\alpha/k)$ for all $i \in [k]$, where $\pi: [k] \rightarrow [k]$ is some permutation.
  The number of samples it requires is $n = \tilde O\left(\frac{d^3}{\mw^2}\log\left(\max_i \frac{|\mu_i|^2}{\sigma_i^2}\right) + \frac{d}{\mw\alpha^2} + \frac{k^2}{\alpha^2}\right)$\footnote{We note that~\cite{NissimRS07} states a similar version of this result, though their coverage omits dependences on the scale of the data.}.
\end{lemma}

From Theorem~\ref{thm:sa}, this implies the following private learning algorithm.
\begin{theorem}
  There exists an $(\eps,\delta)$-differentially private algorithm which takes $n$ samples from some mixture of $k$ Gaussians $\cD = \sum_{i=1}^k w_i \cN(\mu_i, \sigma_i^2\id_{d \times d} )$, with the separation condition that $\|\mu_i - \mu_j\|_2 \geq (\sigma_i + \sigma_j)\tilde \Omega(k^{1/4}\cdot \poly\log\left(k,d,1/\eps,\log(1/\delta),\log^*(\frac{R\sigma_{\max}}{\alpha\sigma_{\min}})\right))$ , and input parameters $R, \sigma_{\min}, \sigma_{\max}$ such that $\|\mu_i\|_2 \leq R$ and $\sigma_{\min} \leq \sigma_i \leq \sigma_{\max}$ for all $i \in [k]$.
  With probability at least $9/10$, it outputs a set of points $\{\hat \mu_1, \dots, \hat \mu_k\}$ and weights $\{\hat w_1, \dots, \hat w_k\}$ such that $\|\hat \mu_{\pi(i)} - \mu_i\|_2 \leq O\left(\alpha\sigma_i\right)$ and $|\hat w_{\pi(i)} - w_i| \leq O(\alpha/k)$ for all $i \in [k]$, for some permutation $\pi$.
  The number of samples it requires is $n = \tilde O\left(\frac{\sqrt{kd} + k^{1.5}}{\eps} \log^{2}(1/\delta) \cdot 2^{O\left(\log^*\left(\frac{dR\sigma_{\max}}{\alpha\sigma_{\min}}\right)\right)}\left(\frac{d^3}{\mw^2} \log\left(\max_i \frac{|\mu_i|^2}{\sigma_i^2}\right)+ \frac{d}{\mw\alpha^2} + \frac{k^2}{\alpha^2}\right)\right)$.

\end{theorem}

We note that plugging more recent advances in learning mixtures of Gaussians~\cite{HopkinsL18,KothariSS18,DiakonikolasKS18b} into Theorem~\ref{thm:sa} allows us to derive computationally and sample efficient algorithms for separations which are $o(k^{1/4})$.
However, we note that even non-privately, the specific sample and time complexities are significantly larger than what we achieve from our more direct construction.

%% file: appendix.tex
\section{Proofs for Deterministic Regularity Conditions}
    \label{sec:regularityproofs}
    
\begin{lemma}\label{lem:min-points-mult}
    Suppose $X^L=((X_1,\rho_1),\dots,(X_n,\rho_n))$
    are labelled samples from
    $\cD \in \cG(d,k,s,R,\kappa,\mw)$. If
    $$n \geq \frac{12}{\mw}\ln(2k/\beta),$$
    then with probability at least $1-\beta$,
    for every label $\rho_i$, for $1 \leq i \leq k$,
    the number of points having label $\rho_i$ is in
    $$\left[\frac{nw_i}{2},\frac{3nw_i}{2}\right].$$
\end{lemma}
\begin{proof}
    It follows directly from Lemma \ref{lem:chernoff-mult}
    by setting $p = w_i$, and taking the union
    bound over all mixture components.
\end{proof}

\begin{lemma}\label{lem:min-points-add}
    Suppose $X^L=((X_1,\rho_1),\dots,(X_n,\rho_n))$
    are labelled samples from
    $\cD \in \cG(d,k,s,R,\kappa,\mw)$. If
    $$n \geq \frac{405k^2}{2\alpha^2}\ln(2k/\beta),$$
    then with probability at least $1-\beta$,
    for every label $\rho_i$, for $1 \leq i \leq k$,
    such that $w_i \geq \tfrac{4\alpha}{9k}$,
    the number of points having label $\rho_i$ is in
    $$\left[n\left(w_i-\frac{\alpha}{9k}\right),
        n\left(w_i + \frac{\alpha}{9k}\right)\right].$$
\end{lemma}
\begin{proof}
    It follows directly from Lemma \ref{lem:chernoff-add}
    by setting $p = w_i$ and
    $\epsilon=\tfrac{4\alpha}{9k}$, and taking
    the union bound over all mixture components.
\end{proof}

\begin{lemma}\label{lem:easy-intra-gaussian}
    Let $X^{L} = ((X_1,\rho_1),\dots,(X_n,\rho_n))$
    be a labelled sample from a Gaussian mixture
    in $\cD \in \cS(\ell,k,\kappa,s)$,
    where $\ell \geq 512 \ln(nk/\beta)$ and $s > 0$.
    Then with probability at least $1-\beta$,
    For every $1 \leq u \leq k$, the radius
    of the smallest ball containing the set of
    points with label $u$ (i.e.~$\{X_i : \rho_i = u\}$)
    is in $[\sqrt{d} \sigma_u / 2 , \sqrt{3 d} \sigma_u]$.
\end{lemma}
\begin{proof}
    For a given $u$, if $X_i,X_j$ are samples from
    $G_{u}$, then by Lemma \ref{lem:gauss-conc}, we have
    with probability at least $1-4e^{-t^2/8}$,
    $$2\ell\sigma^2_{u} - 2t\sigma^2_{u}\sqrt{\ell}
        \leq \llnorm{X_i - X_j}^2 \leq
        2\ell\sigma^2_{u} + 2t\sigma^2_{u}\sqrt{\ell}.$$
    Setting $t = 8\sqrt{2\ln(mk/\beta)}$, and
    taking a union bound over all Gaussians and
    all pairs of samples from every Gaussian,
    we have that with probability at least $1-\beta$,
    for any $u$ and all pairs of points $X_i,X_j$
    from $G_{u}$,
    $$
    2\ell\sigma^2_{u} -
        16\sigma^2_{u}\sqrt{2\ell\ln\left(\frac{mk}{\beta}\right)}
        \leq \llnorm{X_i - X_j}^2 \leq
        2\ell\sigma^2_{u} + 16\sigma^2_{u}\sqrt{2\ell\ln\left(\frac{mk}{\beta}\right)}.
    $$
    By our assumption on $\ell$, we get
    $$\ell\sigma^2_{u} \leq \llnorm{X_i - X_j}^2 \leq
        3\ell\sigma^2_{u}.$$
    Now, because the distance between any two points
    from the same Gaussian
    is at least $\sqrt{\ell}\sigma_u$, the smallest ball
    containing all the points must have radius at least
    $\sqrt{\ell}\sigma_u / 2$.\\
    For the upper bound, for any $u$, let the 
    mean of all points,
    $P = \tfrac{1}{m_{u}}\sum\limits_{i=1}^{m_{u}}{X_i}$
    be a candidate center for the required ball,
    where $X_i$ are samples from $G_{u}$,
    and $m_{u}$ is the number of
    samples from the Gaussian.
    Then for any point $X$ sampled from
    the Gaussian,
    \begin{align*}
        \llnorm{X - P} &\leq
                \llnorm{\frac{1}{m_{u}}
                    \sum\limits_{i=1}^{m_{u}}
                    {X_i} - X}
            \leq \frac{1}{m_{u}}\sum\limits_{i=1}^{m_{u}}
                {\llnorm{X_i - X}}
            \leq \frac{1}{m_{u}} \cdot
                m_{u}\sqrt{3\ell}\sigma_{u}
            = \sqrt{3\ell}\sigma_{u}.
    \end{align*}
    This proves the lemma.
\end{proof}

\begin{lemma}\label{lem:easy-inter-gaussian}
    Let $X^{L} = ((X_1,\rho_1),\dots,(X_n,\rho_n))$
    be a labelled sample from a Gaussian mixture
    in $\cD \in \cS(\ell,k,\kappa,C\sqrt{\ell})$,
    where $\ell \geq 512 \max\{\ln(nk/\beta),k\}$
    and $C>1$ is a universal constant. Then with
    probability at least $1-\beta$, For every
    $\rho_i \neq \rho_j$,
    $$\llnorm{X_i - X_j} \geq \frac{C}{2}
        \sqrt{\ell} \max\{\sigma_{\rho_i}, \sigma_{\rho_j}\}.$$
\end{lemma}
\begin{proof}
    Let $x,y$ be points as described in the
    statement of the Lemma. From Lemma
    \ref{lem:gauss-conc}, we know that
    with probability at least $1-4e^{-t^2/8}$,
    \begin{align*}
        \llnorm{x-y}^2 &\geq \ex{}{\llnorm{x-y}^2} -
                t \left((\sigma_i^2 + \sigma_j^2)\sqrt{\ell} +
                2\llnorm{\mu_i - \mu_j}
                \sqrt{\sigma_i^2 + \sigma_j^2}\right)\\
            &= (\sigma^2_i + \sigma^2_j)\ell +
                \llnorm{\mu_i-\mu_j}^2 -
                t \left((\sigma_i^2 + \sigma_j^2)\sqrt{\ell} +
                2\llnorm{\mu_i - \mu_j}
                \sqrt{\sigma_i^2 + \sigma_j^2}\right).
    \end{align*}
    Setting $t = 16\sqrt{\ln(mk/\beta)}$,
    $\sigma = \max\{\sigma_i,\sigma_j\}$,
    taking the union bound over all pairs
    of Gaussians, and all pairs of points
    from any two Gaussians,
    and using the assumption on $d$,
    we get the following with probability at
    least $1-\beta$.
    \begin{align*}
        \llnorm{x-y}^2 &\geq (\sigma^2_i + \sigma^2_j)\ell +
                \frac{\llnorm{\mu_i-\mu_j}^2}{2} -
                16(\sigma_i^2 + \sigma_j^2)
                \sqrt{\ell\ln\left(\frac{mk}{\beta}\right)}\\
            &\geq \left(1 + \frac{C^2}{2} -
                \frac{C}{2\sqrt{2}}\right)
                \sigma^2 \ell
            > \frac{C^2}{4}\sigma^2 \ell
    \end{align*}
    This proves the lemma.
\end{proof}

\section{Private Location for Mixtures of Gaussians}
    \label{sec:missingproofs-pgloc}

\begin{theorem}[\cite{NissimS18}]\label{thm:ploc}
    Suppose $X_1,\dots,X_m$ are points from
    $S^\ell \subset \R^\ell$, where $S^\ell$ is finite.
    Let $m,t,\beta,\eps,\delta$ be such that,
    $$t \geq O\left(\frac{m^{a}\cdot\sqrt{\ell}}{\eps}
        \log\left(\frac{1}{\beta}\right)
        \log\left(\frac{m\ell}{\beta\delta}\right)
        \sqrt{\log\left(\frac{1}{\beta\delta}\right)}
        \cdot 9^{\log^*(2\abs{S}\sqrt{\ell})}\right),$$
    where $0<a<1$ is a constant that could be arbitrarily
    small. Let $r_{opt}$ be the radius of the smallest
    ball that contains at least $t$ points
    from the sample.
    There exists an $(\eps,\delta)$-DP algorithm
    (called $\PLoc_{\eps,\delta,\beta}$)
    that returns a ball of radius at most
    $w \cdot r_{opt}$ such that it contains at
    least $t-\Delta$ points from the sample with
    error probability $\beta$, where $w = O(1)$ (where the constant depends on the value of $a$),
    and
    $$\Delta = O\left(\frac{m^{a}}{\eps}
        \log\left(\frac{1}{\beta}\right)
        \log\left(\frac{1}{\beta\delta}\right)
        \cdot 9^{\log^*(2\abs{S}\sqrt{\ell})}\right).$$
\end{theorem}

\begin{algorithm}[h!] \label{alg:pgloc}
\caption{Private Gaussians Location
    $\PGLoc(S,t; \eps,\delta,R,\minsigma,\maxsigma)$}
\KwIn{Samples $X_1,\dots,X_{m} \in \R^{\ell}$
    from a mixture of Gaussians.
    Number of points in the target ball: $t$.
    Parameters $\eps, \delta, \beta > 0$.}
\KwOut{Center $\vec{c}$ and radius $r$
    such that $B_r(\vec{c})$ contains
    at least $t/2$ points from $S$.}\vspace{10pt}
    
Set parameters: $\lambda \gets 0.1$\vspace{10pt}

Let $\cX$ be a grid in
    $[ -R - 3\sqrt{\ell} \maxsigma\kappa,
    R + 3\sqrt{\ell} \maxsigma\kappa]^{\ell}$
    of width $\lambda = \frac{\minsigma}{10}$.\\
Round points of $S$ to their nearest points in
    the grid to get dataset $S'$\\
$(\vec{c}, r') \gets \PLoc_{\eps,\delta,
    \beta}(S', t, \cX)$\\
Let $r \gets r' + \lambda\sqrt{\ell}$
\vspace{10pt}

\Return $(\vec{c},r)$
\vspace{10pt}
\end{algorithm}

\begin{proof}[Proof of Theorem \ref{thm:pgloc}]
    We show that Algorithm \ref{alg:pgloc}
    satisfies the conditions in the theorem.

    Privacy follows from Theorem \ref{thm:ploc}, and
    post-processing (Lemma~\ref{lem:post-processing}).

    The first part of the theorem follows
    directly from Theorem \ref{thm:ploc}
    by noting that $\Delta \leq \tfrac{t}{2}$
    for large enough $n,\ell,\abs{S}$, and
    $t = \gamma n$, where $0 < \gamma \leq 1$.
    For all $x \in X$, with high probability, it
    holds that
    $\llnorm{x} \leq R + O(\sqrt{\ell}\maxsigma)$
    by applying Lemma \ref{lem:spectralnormofgaussians}
    after rescaling $x$ appropriately by its covariance,
    then applying the triangle inequality, and noting
    that the empirical mean of a set of points
    lies in their convex hull.\\
    Now, we move to the second part of the
    lemma. Because of the discretization, we
    know that,
    \begin{align}
        \llnorm{x-x'} \leq \lambda\sqrt{\ell}.
            \label{eq:ploc1}
    \end{align}
    Therefore,
    \begin{align*}
        \llnorm{\vec{p}-x} &\leq
                \llnorm{\vec{p}-x'} + \llnorm{x'-x}\\
            &\leq r' + \lambda\sqrt{\ell}.
    \end{align*}
    Let $x,y \in S$ and $x',y' \in S'$ be
    their corresponding rounded points. Then
    from Equation \ref{eq:ploc1}, we know that
    \begin{align}
        \llnorm{x'-y'} &\leq \llnorm{x'-x} +
                \llnorm{x-y} + \llnorm{y-y'}\nonumber\\
            &\leq \llnorm{x-y} + 2\lambda\sqrt{\ell}
                \label{eq:ploc2}.
    \end{align}
    Let $r'_{opt}$ be the radius of the smallest
    ball containing at least $t$ points
    in $S'$. Because of Equation \ref{eq:ploc2},
    we can say that
    $$r'_{opt} \leq r_{opt} + 2\lambda\sqrt{\ell}.$$
    From the correctness of $\PLoc$,
    we can conclude the following,
    \begin{align*}
        r' &\leq cr'_{opt}\\
            &\leq c \left(r_{opt} +
                2\lambda\sqrt{\ell}\right),
    \end{align*}
    where $c>4$ is an absolute constant.
    This gives us,
    \begin{align*}
        r &= r' + \lambda\sqrt{\ell}\\
            &\leq c\left(r_{opt} +
                \frac{9}{4}\lambda\sqrt{\ell}\right)\\
            &\leq c\left(r_{opt} +
                \frac{1}{4}\sqrt{\ell}\minsigma\right).
    \end{align*}
This completes the proof.
\end{proof}

\section{Private Estimation of a Single Spherical Gaussian}
    \label{sec:missingproofs-psge}

\begin{algorithm}[t!] \label{alg:psge}
\caption{Private Spherical Gaussian Estimator
    $\PSGE(X ; \vec{c}, r, \eps,\delta)$}
\KwIn{Samples $X_1,\dots,X_{m} \in \R^{\ell}$.
    Center $\vec{c} \in \R^{\ell}$ and radius $r > 0$ of target component.
    Parameters $\eps, \delta > 0$.}
\KwOut{Mean and variance of the Gaussian.} \vspace{10pt}

Set parameters:
$
    \Delta_{\frac{\eps}{3},\sigma}
        \gets
        \frac{6r^2}{\eps}~~~
        \Delta_{\frac{\eps}{3},\delta,\mu}
        \gets
        \frac{6r\sqrt{2\ln(1.25/\delta)}}{\eps}
$\\\vspace{10pt}

Let $X' \gets X \cap B_r(\vec{c})$\\
For each $i$ such that $X_{2i}, X_{2i-1} \in X'$,
    let $Y_i \gets \frac{1}{\sqrt{2}}(X_{2i} - X_{2i-1})$,
    and let $Y \gets Y_1,\dots,Y_{m'}$\\
Let $m'_X = \PCount_{\frac{\eps}{3}}(X,B_r(\vec{c}))$
    and $m'_Y  = \tfrac{m'_X}{2}$
\vspace{10pt}

\tcp{Private Covariance Estimation}
$\wh{\sigma}^2 \gets \frac{1}{m'_Y \ell}\left(
    \sum\limits_{i=1}^{m'}
    {\llnorm{Y_i}^2} + z_{\sigma}\right)$,
    where $z_{\sigma} \sim
    \Lap\left(\Delta_{\frac{\eps}{2},\sigma}\right)$
\vspace{10pt}

\tcp{Private Mean Estimation}
%For each $i$, $X'_i \gets X_i - \vec{c}$\\
$\wh{\mu} \gets \frac{1}{m'_X}\left(
    \sum\limits_{i=1}^{\abs{X'}}{X'_i} +
    z_{\mu}\right)$, where $z_{\mu} \sim
    \cN\left(0, \Delta^2_{\frac{\eps}{3},
    \delta,\mu}\id_{\ell \times \ell}\right)$\\
%Let $\wh{\mu} \gets \wb{\mu} + \vec{c}$
\vspace{10pt}

\Return $(\wh{\mu},\wh{\sigma}^2)$
\vspace{10pt}
\end{algorithm}

\begin{proof}[Proof of Theorem \ref{thm:easy-learn-gaussian}]
    We show that Algorithm \ref{alg:psge}
    satisfies the conditions in the theorem.

    For the privacy argument for Algorithm
    \ref{alg:psge}, note that we truncate the
    dataset such that all points in $X$ lie
    within $\ball{\vec{c}}{r}$. Now, the
    following are the sensitivities of the
    computed functions.
    \begin{enumerate}
        \item $\sum\limits_{i=1}^{\abs{X'}}{X_i}$:
            $\ell_2$-sensitivity is $2r$.
        \item $\sum\limits_{i=1}^{m'}{Y^2_i}$:
            $\ell_1$-sensitivity is
            $2r^2$.
        \item Number of points in $X'$:
            $\ell_1$-sensitivity is $1$.
    \end{enumerate}
    Therefore, by Lemmata \ref{lem:gaussiandp},
    \ref{lem:laplacedp}, and \ref{lem:composition},
    the algorithm is $(\eps,\delta)$-DP.

    Firstly, either not all the points of $X$ lie
    within $B_{r}(\vec{c})$, that is, $\abs{X'} < m$,
    in which case, we're done, or all the points
    do lie within that ball, that is, $\abs{X'} = m$.
    Since we can only provide guarantees
    when we have a set of untampered random
    points from a Gaussian, we just deal
    with the second case.
    So, $m'_X = m + e$, where $e$ is
    the noise added by $\PCount$. Because
    $m \geq \tfrac{6}{\eps}\ln(5/\beta)$,
    using Lemma \ref{lem:lap-conc}, we know that
    with probability at least
    $1-\tfrac{\beta}{5}$, $\abs{e} \leq
    \tfrac{3\ln(5/\beta)}{\eps}$,
    which means that $m \geq 2\abs{e}$.
    Because of this, the following holds.
    \begin{align}
        \frac{1}{m}
            \left(1-\frac{\abs{e}}{2m}\right)
            \leq \frac{1}{m} \leq
            \frac{1}{m}
            \left(1+\frac{2\abs{e}}{m}\right)
        ~~~\text{and}~~~
        \frac{1}{m}
            \left(1-\frac{\abs{e}}{2m}\right)
            \leq \frac{1}{m + e} \leq
            \frac{1}{m}
            \left(1+\frac{2\abs{e}}{m}\right)
            \label{eq:pge1}
    \end{align}
            
    We start by proving the first part
    of the claim about the estimated mean.
    Let $S_{X'} = \sum\limits_{i=1}^{m}{X'_i}$
    and $S^{\vec{c}}_{X'} = \sum\limits_{i=1}^{m}{(X'_i-\vec{c})}$.
    Then,
    $$\llnorm{\frac{1}{m}S^{\vec{c}}_{X'} -
        \frac{1}{m+e}S^{\vec{c}}_{X'}} =
        \frac{\llnorm{S^{\vec{c}}_{X'}}
        \abs{e}}{m(m+e)}.$$
    We want the above to be at most
    $\tfrac{\sigma\alpha_{\mu}}{4}$. This
    gives us the following.
    $$m(m+e) \geq
        \frac{4\llnorm{S^{\vec{c}}_{X'}}\abs{e}}{\alpha_{\mu}\sigma}$$
    Because $\abs{e} \leq \tfrac{m}{2}$, it
    is sufficient to have the following.
    \begin{align*}
        m^2 &\geq
            \frac{8\llnorm{S^{\vec{c}}_{X'}}\abs{e}}{\alpha_{\mu}\sigma}
    \end{align*}
    But because $X_i \in B_r(\vec{c})$ for all $i$,
    $\llnorm{X_i -\vec{c}} \leq r$.
    So, due to our bound on $\abs{e}$,
    it is sufficient to have the following.
    \begin{align*}
        m^2 \geq
            \frac{24mr\ln(5/\beta)}{\eps\alpha_{\mu}\sigma}
        \iff m \geq
            \frac{24r\ln(5/\beta)}{\eps\alpha_{\mu}\sigma}
    \end{align*}
    This gives us,
    $$\llnorm{\frac{1}{m}S_{X'} -
        \frac{1}{m+e}S_{X'}} \leq
        \frac{\sigma\alpha_{\mu}}{4}.$$
    Now, we want to bound the distance
    between $\tfrac{1}{m+e}S_{X'}$ and
    $\tfrac{1}{m+e}(S_{X'}+z_{\mu})$
    (by $\tfrac{\sigma\alpha_{\mu}}{4}$).
    \begin{align*}
        \llnorm{\frac{1}{m+e}S_{X'} -
            \frac{1}{m+e}(S_{X'}+z_{\mu})} \leq
            \frac{\sigma\alpha_{\mu}}{4}
        \iff \llnorm{\frac{z_{\mu}}{m+e}} \leq
            \frac{\sigma\alpha_{\mu}}{4}
        \iff m+e \geq
            \frac{4\llnorm{z_{\mu}}}{\sigma\alpha_{\mu}}
    \end{align*}
    Because $m \geq 2\abs{e}$, it is sufficient
    to have the following.
    \begin{align*}
        \frac{m}{2} \geq
            \frac{4\llnorm{z_{\mu}}}{\sigma\alpha_{\mu}}
        \iff m \geq
            \frac{8\llnorm{z_{\mu}}}{\sigma\alpha_{\mu}}
    \end{align*}
    Using Lemma \ref{lem:chi-conc}, we
    and noting that $\ell \geq 8\ln(10/\beta)$, we
    know that with probability at least
    $1-\tfrac{\beta}{5}$,
    \begin{align*}
        \llnorm{z_{\mu}} \leq
                \frac{6r\sqrt{2\ln(1.25/\delta)}
                \sqrt{2\ell}}{\eps}
            = \frac{12r\sqrt{\ell
                \ln(1.25/\delta)}}{\eps}
    \end{align*}
    Therefore, it is sufficient to have
    the following.
    \begin{align*}
        m \geq
            \frac{96r\sqrt{\ell
            \ln(1.25/\delta)}}{\sigma\alpha_{\mu}\eps}
    \end{align*}
    %Using Lemma \ref{lem:gauss-conc-1d},
    %and taking the union bound over all
    %$\ell$ dimensions, we know that with
    %probability at least $1-\tfrac{\beta}{5}$,
    %\begin{align*}
    %    \llnorm{z_{\mu}} \leq
    %            \frac{6r\sqrt{2\ln(2.5/\delta)}
    %            \sqrt{2\ell\ln(10\ell/\beta)}}{\eps}
    %        = \frac{12r\sqrt{\ell\ln(10\ell/\beta)
    %            \ln(2.5/\delta)}}{\eps}
    %\end{align*}
    %Therefore, it is sufficient to have
    %the following.
    %\begin{align*}
    %    m \geq
    %        \frac{96r\sqrt{\ell\ln(10\ell/\beta)
    %        \ln(2.5/\delta)}}{\sigma\alpha_{\mu}\eps}
    %\end{align*}
    To complete the proof about the
    accuracy of the estimated mean,
    we need to bound the distance between
    $\tfrac{1}{m}S_{X'}$ and $\mu$
    (by $\tfrac{\sigma\alpha_{\mu}}{2}$).
    Let $\wt{\mu} = \tfrac{1}{m}S_{X'}$.
    Using Lemma \ref{lem:spectralnormofgaussians},
    and the fact that
    $m \geq \tfrac{c_1\ell + c_2\log(1/\beta)}{\alpha^2_{\mu}}$
    for universal constants $c_1,c_2$,
    we get that with
    probability at least $1-\tfrac{\beta}{5}$,
    \begin{align*}
        \llnorm{\mu - \wt{\mu}} \leq \frac{\alpha_{\mu}\sigma}{2}.
    \end{align*}
    We finally apply the triangle inequality
    to get,
    \begin{align*}
        \llnorm{\mu - \wh{\mu}} \leq
                \frac{\sigma\alpha_{\mu}}{2} +
                \frac{\sigma\alpha_{\mu}}{4} +
                \frac{\sigma\alpha_{\mu}}{4} 
            = \sigma\alpha_{\mu}.
    \end{align*}
            
    We now prove the lemma about the
    estimated covariance. Note that
    $m'=\tfrac{m}{2}$. Let
    $\Sigma_Y = \sum\limits_{i=1}^{m'}{\llnorm{Y_i}^2}$.
    We want to show the following.
    \begin{align*}
        \left(1-\alpha_{\sigma}\right)^{1/3}\sigma^2
            \leq \frac{1}{m'\ell}\Sigma_Y &\leq
            \left(1+\alpha_{\sigma}\right)^{1/3}\sigma^2\\
        \left(1-\alpha_{\sigma}\right)^{1/3}
            \frac{1}{m'\ell}\Sigma_Y
            \leq \frac{1}{m'_{Y}\ell}\Sigma_Y &\leq
            \left(1+\alpha_{\sigma}\right)^{1/3}\frac{1}{m'\ell}
            \Sigma_Y\\
        \left(1-\alpha_{\sigma}\right)^{1/3}
            \frac{1}{m'_{Y}\ell}\Sigma_Y
            \leq \frac{1}{m'_{Y}\ell}(\Sigma_Y + z_{\sigma})
            &\leq
            \left(1+\alpha_{\sigma}\right)^{1/3}
            \frac{1}{m'_{Y}\ell}\Sigma_Y
    \end{align*}
    The claim would then follow by substitution.
    We use the fact that for any $t \in [0,1]$,
    \begin{align}
        (1-t)^{1/3} \leq 1-\frac{t}{6} \leq
            1+\frac{t}{6} \leq (1+t)^{1/3}.
            \label{eq:pge2}
    \end{align}
    We start by proving the first inequality.
    Note that for each $i,j$,
    $Y^j_i \sim \cN(0,\sigma^2)$ is i.i.d. Using
    Lemma \ref{lem:chi-conc} and the fact
    that $m \geq \tfrac{576}{\alpha^2_{\sigma}\ell}\ln(10/\beta)$,
    we know that with probability at
    least $1-\tfrac{\beta}{5}$,
    $$\left(1-\frac{\alpha_{\sigma}}
        {6}\right)\sigma^2
        \leq \frac{1}{m'\ell}\Sigma_Y \leq
        \left(1+\frac{\alpha_{\sigma}}
        {6}\right)\sigma^2.$$
    Combining the above with Inequality
    \ref{eq:pge2}, we get the first result.\\
    To prove the second result, since
    $m' = \tfrac{m}{2}$, using Inequality
    \ref{eq:pge1}, it is enough to show
    that,
    $$\frac{\Sigma_Y}{m'\ell}\left(1+\frac{\abs{e}}{m'}\right)
        \leq
        \frac{\Sigma_Y}{m'\ell}\left(1+\frac{\alpha_{\sigma}}
        {6}\right)~~~ \text{and}~~~
        \frac{\Sigma_Y}{m'\ell}\left(1-\frac{\alpha_{\sigma}}
        {6}\right) \leq
        \frac{\Sigma_Y}{m'\ell}\left(1-\frac{\abs{e}}{4m'}
        \right).$$
    Since with high probability,
    $\abs{e} \leq \tfrac{3}{\eps}\ln(5/\beta)$,
    having $m \geq \tfrac{36}
    {\eps\alpha_{\sigma}}\ln(5/\beta)$ satisfies
    the two conditions. This gives us,
    $$\left(1-\frac{\alpha_{\sigma}}
        {6}\right)
        \frac{1}{m'\ell}\Sigma_Y
        \leq \frac{1}{m'_{Y}\ell}\Sigma_Y \leq
        \left(1+\frac{\alpha_{\sigma}}
        {6}\right)\frac{1}{m'\ell}
        \Sigma_Y,$$
    which gives us the required result
    after combining with Inequality
    \ref{eq:pge2}.\\
    To prove the third result, it
    is sufficient to show the following.
    $$\frac{1}{m'_Y\ell}(\Sigma_Y+\abs{z_{\sigma}})
        \leq
        \frac{\Sigma_Y}{m'_Y\ell}\left(1+\frac{\alpha_{\sigma}}
        {6}\right)~~~ \text{and}~~~
        \frac{\Sigma_Y}{m'_Y\ell}\left(1-\frac{\alpha_{\sigma}}
        {6}\right) \leq
        \frac{1}{m'_Y\ell}(\Sigma_Y-\abs{z_{\sigma}})$$
    Note that from Lemma \ref{lem:lap-conc},
    with probability at least
    $1-\tfrac{\beta}{10}$,
    \begin{align*}
    \abs{z_{\sigma}} &\leq
            \frac{6r^2\ln(10/\beta)}{\eps}.
    \end{align*}
    From Lemma \ref{lem:easy-intra-gaussian},
    we know that for any $i,j$, with probability
    at least $1-\tfrac{\beta}{10}$,
    $$\llnorm{Y_i - Y_j} \geq
        \tfrac{\sqrt{\ell}\sigma}{2}.$$
    This means that at least half
    of the points of $Y$ must have $L_2$
    norms at least $\tfrac{\sqrt{\ell}\sigma}{4}$,
    which implies that $\Sigma_Y \geq
    \tfrac{m'\ell\sigma^2}{32} =
    \tfrac{m\ell\sigma^2}{64}$.
    Then the two conditions above will
    be satisfied if,
    \begin{align*}
        \frac{m\ell\sigma^2}{64} \geq
            \frac{6}{\alpha_{\sigma}} \cdot
            \frac{6r^2\ln(10/\beta)}{\eps}
        \iff m \geq
            \frac{2304r^2\ln(10/\beta)}
            {\alpha_{\sigma}\eps\sigma^2 \ell}.
    \end{align*}
    This gives us,
    $$\left(1-\frac{\alpha_{\sigma}}
        {6}\right)
        \frac{1}{m'_{Y}\ell}\Sigma_Y
        \leq \frac{1}{m'_{Y}\ell}(\Sigma_Y + z_{\sigma})
        \leq
        \left(1+\frac{\alpha_{\sigma}}
        {6}\right)
        \frac{1}{m'_{Y}\ell}\Sigma_Y,$$
    which when combined with Inequality
    \ref{eq:pge2}, gives us the third
    result. Combining the three
    results via substitution,
    we complete the proof for the accuracy
    of the estimated variance.
\end{proof}

\section{Additional Useful Concentration Inequalities}

Throughout we will make use of a number of
concentration results, which we collect here
for convenience.  We start with standard tail
bounds for the univariate Laplace and Gaussian
distributions.
\begin{lemma}[Laplace Tail Bound]\label{lem:lap-conc}
    Let $Z \sim \Lap(t)$. Then
    $\pr{}{\abs{Z} > t\cdot\ln(1/\beta)} \leq \beta.$
\end{lemma}

\begin{lemma}[Gaussian Tail Bound]\label{lem:gauss-conc-1d}
    Let $X \sim \cN(\mu,\sigma^2)$. Then
    $\pr{}{\abs{X-\mu} > \sigma\sqrt{2\ln(2/\beta)}} \leq \beta.$
\end{lemma}

We also recall standard bounds on the sums of well behaved random variables.
\begin{lemma}[Multiplicative Chernoff]\label{lem:chernoff-mult}
    Let $X_1,\dots,X_m$ be independent Bernoulli random variables
    taking values in $\zo$. Let $X$ denote their sum and
    let $p = \ex{}{X_i}$. Then for $m \geq \frac{12}{p}\ln(2/\beta)$,
    $$\pr{}{X \not\in \left[ \frac{mp}{2}, \frac{3mp}{2} \right]} \leq
        2e^{-mp/12} \leq \beta.$$
\end{lemma}

\begin{lemma}[Bernstein's Inequality]\label{lem:chernoff-add}
    Let $X_1,\dots,X_m$ be independent Bernoulli random variables
    taking values in $\zo$. Let $p = \ex{}{X_i}$.
    Then for $m \geq \frac{5p}{2\epsilon^2}\ln(2/\beta)$ and
    $\eps \leq p/4$,
    $$\pr{}{\abs{\frac{1}{m}\sum{X_i}-p} \geq \epsilon}
        \leq 2e^{-\epsilon^2m/2(p+\epsilon)}
        \leq \beta.$$
\end{lemma}

\begin{lemma}[Concentration of Empirical Variance]\label{lem:chi-conc}
    Let $X_1,\dots,X_m \sim \cN(0,\sigma^2)$
    be independent. If
    $m \geq \frac{8}{\eps^2}
        \ln\left(\frac{2}{\beta}\right)$
    and $\eps \in (0,1)$, then
    $$\pr{}{\abs{\frac{1}{m}\sum\limits_{i=1}^{m}
        {X^2_i} - \sigma^2} >
        \eps\sigma^2} \leq \beta.$$
\end{lemma}

Finally, we have a concentration lemma from
\cite{VempalaW02} for the distance between two
points drawn from not-necessarily identical
spherical Gaussians.
\begin{lemma}\label{lem:gauss-conc}
    Let $X \sim \cN(\mu_1,\sigma_1^2 \id_{d \times d})$
    and $y \sim \cN(\mu_2,\sigma_2^2 \id_{d \times d})$.
    For $t > 0$,
    $$\pr{}{\abs{\llnorm{x - y}^2 - \ex{}{\llnorm{x-y}^2}} >
        t \left((\sigma_1^2 + \sigma_2^2)\sqrt{d} +
        2\llnorm{\mu_1 - \mu_2}\sqrt{\sigma_1^2 + \sigma_2^2}\right)} \leq 4e^{-t^2/8}.$$
\end{lemma}